\documentclass[12pt]{article}
\usepackage{lMac}
\usepackage{becMac}

\def\showintremarks{n}   
\def\showissues{n}       

\begin{document}
\title{Operators for Parabolic Block Spin Transformations}

\author{Tadeusz Balaban}
\affil{\small Department of Mathematics \authorcr
       Rutgers, The State University of New Jersey \authorcr
       tbalaban@math.rutgers.edu\authorcr
       \  }

\author{Joel Feldman\thanks{Research supported in part by the Natural 
                Sciences and Engineering Research Council 
                of Canada and the Forschungsinstitut f\"ur 
                Mathematik, ETH Z\"urich.}}
\affil{Department of Mathematics \authorcr
       University of British Columbia \authorcr
       feldman@math.ubc.ca \authorcr
       http:/\hskip-3pt/www.math.ubc.ca/\squig feldman/\authorcr
       \  }

\author{Horst Kn\"orrer}
\author{Eugene Trubowitz}
\affil{Mathematik \authorcr
       ETH-Z\"urich \authorcr
       knoerrer@math.ethz.ch, trub@math.ethz.ch \authorcr
       http:/\hskip-3pt/www.math.ethz.ch/\squig knoerrer/}


\maketitle

\begin{abstract}
\noindent
This paper is a contribution to a program to see symmetry breaking in a
weakly interacting many Boson system on a three dimensional lattice at 
low temperature.  It is part of an analysis of the ``small field''  approximation 
to the ``parabolic flow'' which exhibits the formation of a ``Mexican hat''
potential well. Bounds on  the fluctuation integral covariance, as well
as on some other  linear operators, are an important ingredient
in the renormalization group step analysis of  \cite{PAR1,PAR2}.
These bounds are proven here.

\end{abstract}

\newpage
\tableofcontents

\newpage
\section{Introduction}

In \cite{PAR1,PAR2}, we exhibit, for a many particle system of 
weakly interacting Bosons in three space dimensions, the 
formation of a potential well of the type that typically leads to 
symmetry breaking in the thermodynamic limit.  To do so, we 
use the block spin renormalization group approach. In previous papers 
\cite{UV,fnlint1,fnlint2,CPC} (followed by a simple change of
variables) we have written the partition function of such a system on a discrete torus\footnote{All bounds achieved so far are uniform in the 
volume of this torus.} 
in terms of a functional integral on a $1+3$ dimensional space
\begin{equation*}
\cX_0 
= \big( \bbbz / L_\tp\bbbz\big) \times 
  \big(\bbbz^3 / L_\sp\bbbz^3 \big)
\end{equation*}
with positive integers $ L_\tp$, $ L_\sp$. Up to corrections which are 
exponentially small in the coupling constant, and up to a multiplicative 
normalization factor, this representation is of the form
\begin{equation}\label{eqnPINTstartingpoint}
   \int \Big[ \smprod_{x\in \cX_0} 
        \sfrac{ d\psi(x)^\ast\wedge d\psi(x)}{2\pi \imath}\Big] \,
  e^{\cA_0(\psi^*,\psi) }\chi_0(\psi)
\end{equation}
with an action $\cA_0$ of the form
\begin{equation}\label{eqnPINaIn}
\cA_0(\psi_*,\psi) = -\<\psi_*,\,D_0\psi\>_0
       -\cV_0(\psi_*,\psi)
       +\mu_0 \<\psi_*,\,\psi\>_0
       +\cE'_0(\psi_*,\psi)
\end{equation}
Here
\begin{itemize}[leftmargin=*, topsep=2pt, itemsep=0pt, parsep=0pt]
\item
$
D_0=\bbbone - e^{-\oh_0} -e^{-\oh_0} \partial_0\,$,
 where $\,\partial_0$  the forward time derivative  (see 
\eqref{eqnPBSforwardDeriv} below and $\oh_0 $ is -- up to a scaling --  
the single particle Hamiltonian.
\item
$
\cV_0(\psi_*,\psi)
$ 
is a quartic monomial that describes the coupling between the particles
\item 
$\mu_0 $ is related to the chemical potential of the system
\item
 $\cE'_0(\psi_*,\psi)$ is perturbatively small
\item
 $\,\chi_0(\psi)\,$ is a ``small field cut off function''.
\end{itemize}

\noindent See \cite[(\eqnHTstartingpoint), (\eqnHTaIn)]{PAR1}.

For the block spin renormalization group action, we pick a 
 ``block rectangle'' of length $L^2$ in the ``time direction'' 
and $L$ in ``space directions'', where $L$ is a sufficiently 
large odd positive integer, and a corresponding nonnegative, 
compactly supported function $q(x)$ on $\bbbz\times \bbbz^3$ 
(the averaging profile). The choice of this kind of rectangle is 
characteristic of ``parabolic scaling''. See 
\cite[Definition \defHTscaling, Remark \remHTbasicremarkonscaling,
Definition \defHTbasicnorm.d]{PAR1}. 
For simplicity we assume that $L_\sp$ and $L_\tp$ are powers of $L$. 

The block spin averaging operator, which we denote $Q$, maps functions
on the lattice  $\cX_0$ to functions on the ``coarse'' lattice
$\cX_{-1}^{(1)}= \big( L^2\bbbz / L_\tp\bbbz\big) \times 
  \big(L\bbbz^3 / L_\sp\bbbz^3 \big)$. After each renormalization group
step we scale, to again give functions on a unit lattice. After the
first RG step this unit lattice is  
$\cX_{0}^{(1)}= \big( \bbbz /\sfrac{1}{L^2}L_\tp \bbbz\big) \times 
  \big(\bbbz^3 / \sfrac{1}{L}L_\sp\bbbz^3 \big)$. The ``scaled''
block spin averaging operator maps functions
on the lattice  
$\cX_1=\big(\sfrac{1}{L^2}\bbbz /\sfrac{1}{L^2}L_\tp \bbbz\big) \times 
  \big(\sfrac{1}{L}\bbbz^3 / \sfrac{1}{L}L_\sp\bbbz^3 \big)$ to functions 
on the  unit lattice  $\cX_{0}^{(1)}$.

In the $n^{\rm th}$ renormalization group step, we end up considering functions
on the chain of lattices
\begin{equation*}
\cX_{-1}^{(n+1)} \subset \cX_0^{(n)} \subset \cX_1^{(n-1)} 
\subset \cdots \subset  \cX_{n-1}^{(1)} \subset  \cX_{n}^{(0)}
\end{equation*}
where, for integers $\,j\ge -1\,$ and $\,n\ge 0\,$,
\begin{equation*}
\cX_{j}^{(n)} 
= \big(\veps_j^2 \bbbz / \veps_{n+j}^2 L_\tp\bbbz \big)  
   \times \big(\veps_j\bbbz^3 / \veps_{n+j} L_\sp  \bbbz^3 \big)
\qquad\text{with } \veps_j=\sfrac{1}{L^j}
\end{equation*}
The subscript in $\cX_{j}^{(n)}$ determines the ``coarseness'' of the lattice
--- nearest neighbour points are a distance $\veps_{2j}=\sfrac{1}{L^{2j}}$ 
apart in the time direction and a distance $\veps_j=\sfrac{1}{L^j}$ apart in
spatial directions. The superscript in $\cX_{j}^{(n)}$ determines the number
of points in the lattice --- $|\cX_{j}^{(n)}|=|\cX_0|/L^{5n}$ for all $j$.
 We usually write $\cX_{n}^{(0)}= \cX_n$.
See \cite[Definition \defHTbackgrounddomaction.a or Appendix \appDEFlattices]{PAR1}.

The $(n+1)^{\rm st}$ block spin transformation involves the passage from
$\cX_0^{(n)}$ to its sublattice $\cX_{-1}^{(n+1)}$. The averaging operations
determine linear maps
\begin{equation*}
Q: \cH^{(n)}_0 \mapsto  \cH^{(n+1)}_{-1} \qquad \text{and}
\qquad 
Q_n: \cH_n=\cH^{(0)}_n \mapsto  \cH^{(n)}_0
\end{equation*}
where $\cH^{(n)}_j =L^2\big(\cX_{j}^{(n)}\big)$ denotes the (finite dimensional)
Hilbert space  of functions on $\cX_{j}^{(n)} $ with integral 
$\
\int_{X^{(n)}_j} du = \veps_j^5 \smsum_{u\in\cX_j^{(n)} }
\ $
and the \emph{real} inner product
\begin{equation*}
\< \al_1,\al_2\>_j
= \int_{X^{(n)}_j} \al_1(u) \,\al_2(u) \ du
\end{equation*}
Again see \cite[Definition \defHTbackgrounddomaction.a or 
   Appendix \appDEFblockspinops]{PAR1}.
In \S\ref{secPOblockSpin}  we pick a specific averaging profile 
and give bounds on the operators
$Q$, $Q_n$, their Fourier transforms, and related operators.

Scaling is performed by the linear isomorphisms
\begin{equation*}
\bbbl : \cX_j^{(n)}\rightarrow \cX_{j-1}^{(n)}
    \qquad  \qquad
  (u_0,\bu) \mapsto (L^2u_0, L\bu)
\end{equation*}
 For a function $\,\al\in\cH_j^{(n)}\,$,
define  the function $\,\bbbl_*(\al) \in \cH_{j-1}^{(n)}\,$  by 
$\,
\bbbl_*(\al)(\bbbl u) = \al(u)
\,$. 
See  \cite[Appendix \appDEFscaling]{PAR1}. 
In particular, after rescaling and multiplication with the ``scaling factor'' $ L^{2n}$,  the differential operator $D_0$ in \eqref{eqnPINaIn} becomes the operator
\begin{equation*}
D_n = L^{2n}\ \bbbl_*^{-n}
         \big(\bbbone - e^{-\oh_0} -e^{-\oh_0} \partial_0\big)\bbbl_*^n
\end{equation*}
on $\cH_n$. 
This operator is discussed in \S\ref{secPOdiffOps}.

As mentioned above, the passage from a functional integral on $\cX_0^{(n)}$ 
to a functional integral on $\cX_{-1}^{(n+1)}$ is an averaging procedure over,
roughly speaking, a rectangle of size $L^2$ in the time direction and
size $L$ in the spatial directions. This passage is analyzed using 
stationary phase techniques that involve 
\begin{itemize}[leftmargin=*, topsep=2pt, itemsep=0pt, parsep=0pt]
\item the determination of critical fields on $\cX_0^{(n)}$
(that are functions of external fields on $\cX_{-1}^{(n+1)}$) for an 
appropriate action, and 
\item a functional integral over ``fluctuation fields'' around the critical field. 
\end{itemize}
The covariance for the integral over the fluctuation fields has been 
identified in \cite[(\eqnHTcn)]{PAR1} and is bounded in \S\ref{secPOcovariance}.

The composition of the critical fields of $n$ renormalization group steps
is -- after rescaling -- a field on $\cX_n$, called the ``background field'', 
that is a function of an external field on $\cX_0^{(n)}$ .
It is crucial in our representation of the partition function. See 
\cite[Theorem \thmTHmaintheorem]{PAR1}.
The ``leading order'' part of the background field is linear in the 
external field and has been identified in 
\cite[Proposition \propHTexistencebackgroundfields]{PAR1}. It is the 
composition of an operator, from $\cH_0^{(n)}$ to $\cH_n$ determined 
by the averaging profile $q$,
and an operator $S_n$ on $\cH_n$ which can be viewed as a Green's function 
for the differential operator $D_n$ (plus a mass term). This operator, 
$S_n$, is discussed in \S\ref{secPOgreens}.

To get bounds on  the critical fields in the fluctuation integral at step 
$n+1$, we use a well known algebraic relation between these critical 
fields and the background fields at step $n+1$ given in 
\cite[Proposition \propFormalFldSlns]{BlockSpin}  
and \cite[Proposition \propBGAomnibus.a]{PAR1}. 
The operators 
in the linearization of this relation, and various other linearizations, 
are studied in \S\ref{secPOleadingOrder}.

By construction, many of the operators discussed in this paper are 
linear operators defined on the Hilbert space of functions on a lattice 
that are invariant under translations with respect to a sublattice. 
It is natural to use Bloch/Floquet decompositions and Fourier transforms 
for an analysis of such operators.  We use the abstract basic results about such 
decompositions given in \cite{Bloch}. 

\bigskip

Most operator estimates we obtain in this paper are with respect to a 
norm of the following kind.

\begin{definition}\label{defPINoperatornorm}
For any operator $A:\cH_j^{(n-j)}\rightarrow\cH_k^{(n-k)}$, with kernel
$A(u,u')$, and for any mass $m\ge 0$, we define the norm
\begin{equation*}
\|A\|_m
=\max\Big\{\sup_{u\in\cX_k^{(n-k)}}\,\int_{\cX_j^{(n-j)}} \hskip-15pt du'\ 
                  e^{m|u-u'|}|A(u,u')|\ ,\ 
\sup_{u'\in\cX_j^{(n-j)}}\int_{\cX_k^{(n-k)}}\hskip-15pt du \ 
                           e^{m|u-u'|}|A(u,u')|
\Big\}
\end{equation*}
In the special case that $m=0$, this is just the usual $\ell^1$--$\ell^\infty$
norm of the kernel.
\end{definition}

As we point out in \cite[Lemmas \lemBOlonelinfty\ and \lemBOuniqueness]{Bloch} 
this norm is related to the analyticity properties of the Fourier transform.
In this paper we use the following Fourier transform conventions.
\smallskip
 
\noindent
The dual lattice of $\cX_{j}^{(n)}$ is
\begin{equation*}
\hat\cX_j^{(n)}
  =\big(\sfrac{2\pi}{\veps_{n+j}^2L_\tp}\bbbz/\sfrac{2\pi}{\veps_j^2}\bbbz\big)
    \!\times\!
 \big(\sfrac{2\pi}{\veps_{n+j}L_\sp}\bbbz^3/\sfrac{2\pi}{\veps_j}\bbbz^3\big)
\end{equation*}
For a function
$\al\in \cH_{j}^{(n)}$
\begin{equation*}
\hat \al(p)
 =\int_{\cX_{j}^{(n)}} \al(u) e^{-i p\cdot u }\ du
\qquad 
\al(u)
  =\int_{\hat\cX_{j}^{(n)}}  \hat \al(p) e^{i u\cdot p }\ \sfrac{dp}{(2\pi)^4}
\end{equation*}
where
$\ 
\int_{\hat\cX_{j}^{(n)}} \sfrac{dp}{(2\pi)^4} 
= \sfrac{1}{\veps_{n+j}^5L_\tp L_\sp^3}    \sum_{p\in\hat\cX_{j}^{(n)}}
\ $.
The maps
\begin{equation*}
\bbbl : \hat \cX_{j-1}^{(n)}\rightarrow \hat\cX_j^{(n)}
    \qquad  \qquad
  (q_0,\bq) \mapsto (L^2q_0, L\bq)
\end{equation*}
are again linear isomorphisms, and, for a function $\,\al\in\cH_j^{(n)}\,$,
\begin{equation}\label{eqnPINTlft}
\widehat{\bbbl_*(\al)} (q) = L^ 5 \hat \al(\bbbl q)
\end{equation}
The quotient map dual to the inclusion
$\cX_j^{(n)}\subset \cX_{j+k}^{(n-k)}$ is
\begin{equation}\label{eqnPINTproj}
\hat \pi_{n+j}^{(j+k,j)} : \hat\cX_{j+k}^{(n-k)} \rightarrow \hat\cX_j^{(n)}
\end{equation}
When the indices are clear from the context we suppress them and 
write $\hat\pi$. 

\bigskip

The estimates of this paper are used in \cite{PAR1,PAR2}. In particular, the 
construction of the background fields and the critical fields in 
\cite{BGE} uses a contraction mapping 
argument around the linearizations of \S\ref{secPOleadingOrder}
and \S\ref{secPOgreens} in this paper.

For the readers' convenience we have included, in Appendix \ref{appLatOpSummary}, 
a list of most of the operators and lattices that appear in this paper.

\begin{convention}\label{convPOconstants}
Most estimates in this paper are bounds on norms of operators as in 
Definition \ref{defPINoperatornorm}. The (finite number of) constants that 
appear in these bounds are consecutively labelled
$\Gam_1, \Gam_2, \cdots $, $\gam_1, \gam_2, \cdots $, $m_1, m_2, \cdots$.
All of these constants $\Gam_j$, $\gam_j$, $m_j$ are independent
of $L$ and the scale index $n$. We define $\Gam_\op$ to be the maximum 
of the $\Gam_j$'s, and, in \cite{PAR1,PAR2, BGE}, refer to the estimates 
using only this constant $\Gam_\op$. 
\end{convention}

\newpage
\section{Block Spin Operators}\label{secPOblockSpin}
In this chapter, we analyze the block spin ``averaging operators'' 
$Q$ of \cite[Definitions \defHTblockspintr.a and \defHTbasicnorm.d]{PAR1}
and $Q_n$ of \cite[Definition \defHTbasicnorm.d]{PAR1} as well as the
operator $\fQ_n$ of \cite[Definition \defHTbackgrounddomaction.b]{PAR1}. 
Recall that $Q:\cH_0^{(n)}\rightarrow  \cH_{-1}^{(n+1)}$ is defined by
\begin{equation}\label{eqnPBSaveop}
(Q\psi)(y) = \smsum_{x\in \bbbz\times\bbbz^3} q(x) \psi(y+[x])
\end{equation}
where $\,[x]\,$ denotes the class of $\,x\in\bbbz\times\bbbz^3\,$ in the
quotient space $\,\cX_0^{(n)}\,$.
The averaging profile $q$ is the $\fq$--fold convolution of 
the characteristic function, $1_{\sq}(x)$, of the rectangle 
$\big[-\sfrac{L^2-1}{2}, \sfrac{L^2-1}{2} \big]
\times \big[-\sfrac{L-1}{2}, \sfrac{L-1}{2} \big]^3$, normalized
to have integral one. That is,
\begin{equation*}
q=\sfrac{1}{L^{5\fq}} 
\overbrace{ 1_{\sq}*1_{\sq}*\cdots*1_{\sq} }^{\fq\ {\rm times}}
\end{equation*}
See \cite[Example \exBOnaive\ and Remark \remBOlessnaive]{Bloch}.
Except where otherwise stated, we shall assume that $\fq\ge4$ 
is a fixed even natural number.\footnote
{See Remark \ref{remPBOqfour} for a discussion of the condition $\fq>2$.
 The condition that $\fq$ be even is imposed purely for convenience.}

The operator 
\begin{equation}\label{eqnPBSqn}
Q_n = Q^{(1)} \cdots  Q^{(n)} 
   = \big(\bbbl_*^{-1}Q\big)^n \bbbl_*^n
         :  \cH_n=\cH_n^{(0)} \rightarrow \cH_0^{(n)} 
\end{equation}
where $Q^{(j)} = \bbbl_*^{-j}\,Q \,\bbbl_*^j:
        \cH^{(n-j)}_j\rightarrow\cH^{(n-j+1)}_{j-1}$. The operator
\begin{equation*}
\fQ_n=a\Big(\bbbone
             +\sum_{j=1}^{n-1}\sfrac{1}{L^{2j}}Q_jQ_j^*\Big)^{-1} 
\end{equation*}

The Fourier transform of the characteristic function $1_{\sq}$ is
$\sfrac{\si(\bbbl k)}{\si(k)}$ with $k\in \hat\cX_0^{(n)}$ and with
\begin{equation}\label{eqnPBSsidef}
\si(k) = \sin\big(\half k_0\big)
          \prod_{\nu=1}^3 \sin\big(\half \bk_\nu\big)
\end{equation}
Therefore
\begin{equation}\label{eqnPBSuplusdef}
\hat q(k)=u_+(k)^\fq\qquad\text{with}\qquad
u_+(k) =\frac{\si(\bbbl k)}{L^5\si(k)}
\end{equation}
and, by \cite[Lemma \lemBOfourier.a]{Bloch}
\begin{equation}\label{eqnPBSqaction}
\widehat{(Q\psi)}(\fk)
  =\sum_{ \atop{k\in\hat\cX_0^{(n)}}{\hat\pi(k)=\fk}} 
                       \hat q(k)\hat\psi(k)
\end{equation}
for all $\psi\in\cH_0^{(n)}$ and $\fk\in \hat\cX_{-1}^{(n+1)}$.

\begin{remark}\label{remPBSqnft}
\ 
\begin{enumerate}[label=(\alph*), leftmargin=*]
\item 
Since $\fq$ is even, $\si(k)^\fq$ is an entire function of 
$k\in\bbbc\times\bbbc^3$ that is periodic with respect to the lattice 
$2\pi (\bbbz\times\bbbz^3)$. 
Also
\begin{align*}
\si(p_j)^\fq = \si\big(\hat \pi^{(j,0)}_n(p_j)\big)^\fq
\qquad\text{for all }p_j\in \hat\cX_j^{(n-j)}
\end{align*}

\item 
For all $\phi\in\cH_n$ and $k\in\hat\cX_0^{(n)}$,
\begin{align*}
\widehat{(Q_n\phi)}(k)
  &=\sum_{ \atop{p\in\hat\cX_n}{\hat\pi(p)=k}}
                         u_n(p)^\fq\hat\phi(p)\qquad\text{with}\qquad
u_n(p)
=\veps_n^5\frac{\si(p)}{\si(\bbbl^{-n}p)}
\end{align*}

\item For all $\psi\in\cH_0^{(n)}$ and $k\in\hat\cX_0^{(n)}$,
$\widehat{\fQ_n\psi}(k)=\hat\fQ_n(k)\hat\psi(k)$ where
\begin{equation*}
\hat \fQ_n(k) = a\bigg[1+\sum_{j=1}^{n-1}
    \sum_{ \atop{p_j\in\hat\cX_j^{(n-j)}}
                {\hat\pi(p_j)= k} }
         \sfrac{1}{L^{2j}}u_j(p_j)^{2\fq}\bigg]^{-1}
\end{equation*}

\item The functions $u_n(p)$ and $u_+(p)$ are entire in $p$  and are
invariant under $p_\nu\rightarrow-p_\nu$ for each $0\le\nu\le 3$
and under $p_\nu\leftrightarrow p_{\nu'}$ for all $1\le\nu,\nu'\le 3$.

\item  Set, with the notation of \eqref{eqnPINTproj}, 
the ``single period'' lattices and their duals
\begin{alignat*}{3}
\cB^+\!&=\big(\bbbz/L^2\bbbz\big)\times
         \big(\bbbz^3/L\bbbz^3\big) &\qquad
\hat\cB^+\!
  &=\big(\sfrac{2\pi}{L^2}\bbbz/2\pi\bbbz\big)
    \!\times\!
 \big(\sfrac{2\pi}{L}\bbbz^3/2\pi\bbbz^3\big) 
  =\ker \hat \pi_{n-1}^{(0,-1)}
\\
\cB_j&=\big(\veps_j^2\bbbz/\bbbz\big)\times
         \big(\veps_j\bbbz^3/\bbbz^3\big) &\qquad
\hat\cB_j
  &=\big(2\pi\bbbz/\sfrac{2\pi}{\veps_j^2}\bbbz\big)
    \!\times\!
 \big(2\pi\bbbz^3/\sfrac{2\pi}{\veps_j}\bbbz^3\big) 
 = \ker \hat \pi_n^{(j,0)}
\end{alignat*}
for each integer $j\ge 0$. In this notation, the representations of 
$Q$, $Q_n$ and $\fQ_n$ of \eqref{eqnPBSqaction} and parts (b) and (c) are
\begin{align*}
\widehat{(Q\psi)}(\fk)
  &=\sum_{ \ell\in\hat\cB^+} 
                       u_+(\fk+\ell)^\fq\hat\psi(\fk+\ell)
\displaybreak[0]\\
\widehat{(Q_n\phi)}(k)
  &=\sum_{ \ell\in\hat\cB_n} 
                         u_n(k+\ell)^\fq\hat\phi(k+\ell)
\displaybreak[0]\\
\hat \fQ_n(k) &= a\bigg[1+\sum_{j=1}^{n-1}
    \smsum_{\ell_j\in\hat\cB_j}\sfrac{1}{L^{2j}}u_j(k+\ell_j)^{2\fq}\bigg]^{-1}
\end{align*}
Here in $\widehat{(Q\psi)}(\fk)=\sum_{ \ell\in\hat\cB^+} 
                       u_+(\fk+\ell)^\fq\hat\psi(\fk+\ell)$, for example,
$\fk\in\hat\cX^{(n+1)}_{-1}$ is represented by the element of 
$\sfrac{2\pi}{\veps_n^2L_\tp}\bbbz\times \sfrac{2\pi}{\veps_n L_\sp}\bbbz^3$
having minimal components and $\ell$ is represented by the element of 
$\sfrac{2\pi}{\veps_{-1}^2}\bbbz\times \sfrac{2\pi}{\veps_{-1}}\bbbz^3$
having minimal components. Similarly
\begin{align*}
\big(\widehat{Q^* \th}\big)(\fk+\ell)
   &=u_+(\fk+\ell)^\fq \hat\th(\fk)\\
\big(\widehat{Q_n^* \psi}\big)(k+\ell_n)
   &=u_n(k+\ell_n)^\fq \hat\psi(k)
\end{align*}
\end{enumerate}
\end{remark}

\begin{proof}
(a) Any two points of $\hat\cX_j^{(n-j)}$ with the same image 
in $\hat\cX_0^{(n)}$ under $\hat \pi^{(j,0)}_n$ differ by $2\pi$ times 
an integer vector. The formula follows.

\Item (b)
By \eqref{eqnPINTlft} and \eqref{eqnPBSqaction}, we have, 
for $\al\in\cH_j^{(n-j)}$ 
and $p_{j-1}\in \hat\cX_{j-1}^{(n-j+1)}$
\begin{alignat*}{3}
\widehat{(Q^{(j)}\al)}(p_{j-1})
  &= \sfrac{1}{L^{5j}}\widehat{(QL_*^j\al)}(\bbbl^{-j}p_{j-1}) 
   = \sfrac{1}{L^{5j}}\sum_{ \atop{k\in\hat\cX_0^{(n-j)}}
                                  {\hat\pi(k)=\bbbl^{-j} p_{j-1}} } 
                       \hat q(k)\widehat{(\bbbl_*^j\al)}(k)
\\
  &= \frac{1}{L^{5\fq}}\sum_{ \atop{k\in\hat\cX_0^{(n-j)}}
                                   {\hat\pi(k)=\bbbl^{-j} p_{j-1}} }
              \frac{\si(\bbbl k)^\fq}{\si(k)^\fq}\hat\al\big(\bbbl^jk\big) 
   = \frac{1}{L^{5\fq}}\sum_{ \atop{p_j\in\hat\cX_j^{(n-j)}}
                                   {\hat\pi(p_j)= p_{j-1}} }
           \frac{\si(\bbbl^{-j+1} p_j)^\fq}{\si(\bbbl^{-j}p_j)^\fq}\hat\al(p_j)
\end{alignat*}
so that, by part (a),
\begin{align*}
\widehat{(Q_n\phi)}(p_0)
&= \frac{1}{L^{5\fq n}}\sum_{ \atop{\atop{p_j\in\hat\cX_j^{(n-j)}} 
                                         {\hat\pi(p_j)= p_{j-1}} }
                                   {1\le j\le n} }
           \frac{\si(p_1)^\fq}{\si(\bbbl^{-1}p_1)^\fq}\ 
           \frac{\si(\bbbl^{-1} p_2)^\fq}{\si(\bbbl^{-2}p_2)^\fq}\ 
    \cdots\ 
             \frac{\si(\bbbl^{-n+1} p_n)^\fq}{\si(\bbbl^{-n}p_n)^\fq}\ 
           \hat\phi(p_n) \\
&= \veps_n^{5\fq}\sum_{ \atop{\atop{p_j\in\hat\cX_j^{(n-j)}}
                                   {\hat\pi(p_j)= p_{j-1}} }
                              {1\le j\le n} }
           \frac{\si(p_n)^\fq}{\si(\bbbl^{-1}p_n)^\fq}\ 
           \frac{\si(\bbbl^{-1} p_n)^\fq}{\si(\bbbl^{-2}p_n)^\fq}\ 
    \cdots\ 
             \frac{\si(\bbbl^{-n+1} p_n)^\fq}{\si(\bbbl^{-n}p_n)^\fq}\ 
           \hat\phi(p_n) \\
&= \veps_n^{5\fq}\sum_{ \atop{p_n\in\hat\cX_n}
                             {\hat\pi(p_n)= p_0} }
           \frac{\si(p_n)^\fq}{\si(\bbbl^{-n}p_n)^\fq}\ 
           \hat\phi(p_n) 
\end{align*}

\Item (c)
 follows from part (b) and \cite[Lemma \lemBOfourier.a]{Bloch}.

\Item (d)
 is obvious since $\sfrac{\sin z}{\sin\sfrac{z}{m}}$
 is even and entire for any nonzero integer $m$.
\end{proof}

In Lemmas \ref{lemPBSunppties} and \ref{lemPBSuplusppties} we derive a number of
bounds on the kernels $u_n$ and $u_+$ that appear in the representations
for $Q_n$ and $Q$ of Remark \ref{remPBSqnft}.e. In Proposition \ref{propPBSAnppties} 
we analyze the operator $\hat\fQ_n$. Then in Remark \ref{remPBSunderivAlg} and
Lemma \ref{lemPBSunderiv} we study how to move derivatives past $Q$ and $Q_n$.

When dealing with the asymmetry between ``temporal'' and 
``spatial'' scaling we set, for convenience,
\begin{align}\label{eqnPINTlnu}
L_\nu&=\left.\begin{cases} L^2 & \text{for $\nu=0$} \\
               L &   \text{for $\nu=1,2,3$}
             \end{cases}\right\}
\qquad\veps_{n,\nu}=\left.\begin{cases}
    \veps_n^2=\sfrac{1}{L_0^n}=\sfrac{1}{L^{2n}} & \text{for $\nu=0$}\\
    \veps_n=\sfrac{1}{L_\nu^n}=\sfrac{1}{L^n} & \text{for $\nu=1,2,3$} 
     \end{cases}\right\}
\end{align}

\begin{lemma}\label{lemPBSunppties}
Let $\fq\in\bbbn$.
Assume that $|\Re k_\nu|\le\pi$, $|\Im k_\nu|\le 2$  for each $0\le\nu\le3$.

\begin{enumerate}[label=(\alph*), leftmargin=*]
\item 
$
\big|u_n(k+\ell)\big|\le \smprod_{\nu=0}^3\,
\sfrac{24}{|\ell_\nu|+\pi}
$ 
for all $\ell\in\hat\cB_n$. We use $|\ell_\nu|$ to denote the magnitude of
the smallest representative of $\ell_\nu$ in its equivalence class,
as an element of $\hat\cB_n$. There is a constant 
$\Gam_1$, depending only on $\fq$, such that
$
\|Q_n\|_{m=1}\le\Gam_1
$.

\item
$
\big|u_n(k+\ell)\big|\le 
   \Big[\prod\limits_{\atop{0\le\nu\le 3}{\ell_\nu\ne 0}}|k_\nu| \Big]
   \prod_{\nu=0}^3  \sfrac{24}{|\ell_\nu|+\pi}
$ 
if $0\ne \ell\in\hat\cB_n$.

\item
 $\big|u_n(k)-1\big|\le 4^3 |k|^2$.

\item 
If $\ell\in\hat\cB_n$ and $\ell_{\tilde\nu}\ne 0$
for some $0\le\tilde\nu\le3$, then
$u_n(k+\ell)
     =\sin\big(\half k_{\tilde\nu}\big) v_{n,\tilde\nu}(k+\ell)$ with
$
\big|v_{n,\tilde\nu}(k+\ell)\big|
      \le \smprod_{\nu=0}^3  \sfrac{24}{|\ell_\nu|+\pi}
$.

\item 
For all $\ell\in\hat\cB_n$,
\begin{align*}
\big|\Im u_n(k+\ell)\big|
&\le 12\ |\Im k|\smprod_{\nu=0}^3  \sfrac{24}{|\ell_\nu|+\pi}\\
\big|\Im u_n(k+\ell)^\fq\big|
&\le 12\,\fq\ |\Im k|\ 
  \Big[\smprod_{\nu=0}^3  \sfrac{24}{|\ell_\nu|+\pi}\Big]^\fq
\end{align*}

\item
Recall that $|\Re k_\nu|\le\pi$, $|\Im k_\nu|\le 2$  
for each $0\le\nu\le3$. We have
\begin{equation*}
\sfrac{1}{4\pi^4}  \le |u_n(k)| \le 4\pi^4
\end{equation*}
If, in addition, $k$ is real 
\begin{equation*}
\big(\sfrac{2}{\pi}\big)^4  \le u_n(k)\le\big(\sfrac{\pi}{2}\big)^4
\end{equation*}
\end{enumerate}
\end{lemma}

\begin{proof} 
 Set $s(x)=\sfrac{\sin x}{x}$. By the definitions of $u_n(p)$ 
in Remark \ref{remPBSqnft}.b,  $\si(k)$ in \eqref{eqnPBSsidef} and
$\veps_{n,\nu}$ in \eqref{eqnPINTlnu},
\begin{equation}\label{eqnPBSunagain}
u_n(p) 
=\frac{\sin\half p_0}{\sfrac{1}{\veps_n^2}\sin\half\veps_n^2 p_0}
\prod_{\nu=1}^3 
      \frac{\sin\half\bp_\nu}{\sfrac{1}{\veps_n}\sin\half\veps_n\bp_\nu}
=\prod_{\nu=0}^3 \frac{s(p_\nu/2)}{s(\veps_{n,\nu} p_\nu/2)}
\end{equation}

\noindent (a) 
 We may assume without loss of generality that $\ell_\nu$
is bounded, as a real number, by $\sfrac{\pi}{\veps_{n,\nu}}-\pi$. (Recall
that $\sfrac{1}{\veps_{n,\nu}}$ is an odd natural number.) 
So $\big|\Re k_\nu+\ell_\nu\big|$ is always bounded by 
$\sfrac{\pi}{\veps_{n,\nu}}$.
Consequently,  the hypotheses of Lemma \ref{lemPBSsinxoverx}.c, 
with $x+iy=k_\nu+\ell_\nu$ and $\veps=\veps_{n,\nu}$, are satisfied
and
\begin{align*}
\bigg|\frac{\sin\half (k_\nu+\ell_\nu)}
   {\sfrac{1}{\veps_{n,\nu}}\sin\half\veps_{n,\nu} (k_\nu+\ell_\nu)}\bigg|
\le\left.\begin{cases}
            4 &\text{if $\ell_\nu=0$}\\
           \sfrac{8}{|\Re k_\nu+\ell_\nu|}  
                  &\text{if $|\ell_\nu|\ge 2\pi$}
            \end{cases}\right\}
\le\frac{24}{|\ell_\nu|+\pi}
\end{align*}
since  $|\Re k_\nu|\le\pi$ and $\ell_\nu\in 2\pi\bbbz$. 

When $\fq>1$, $|u_n(k+\ell)|^\fq$ is summable in $\ell$ and the bound
on $\|Q_n\|_{m=1}$ follows from \cite[Lemma \lemBOlonelinfty.c]{Bloch}.
When $\fq=1$, we use that the action of
\begin{equation*}
Q_n:\cH_n^{(0)}\rightarrow  \cH_0^{(n)}
\end{equation*}
in position space is
\begin{equation}\label{eqnELBSavejop}
(Q_n\phi)(x) = \smsum_{v\in \veps_n^2\bbbz\times\veps_n\bbbz^3} u_n(v) \phi(x+[v])
\end{equation}
where $\,[v]\,$ denotes the class of 
$\,v\in\veps_n^2\bbbz\times\veps_n\bbbz^3\,$ 
in the quotient space 
\begin{equation*}
\cX_{n}^{(0)}=
     \big(\veps_n^2\bbbz/\veps_n^2 L_\tp\bbbz\big)
     \times\big(\veps_n\bbbz^3/\veps_n L_\sp\bbbz^3\big)
\end{equation*}
and $x$ runs over
\begin{equation*}
\cX_0^{(n)}= \big(\bbbz/\veps_n^2 L_\tp\bbbz\big)
     \times\big(\bbbz^3/\veps_nL_\sp\bbbz^3\big)
\end{equation*}
When $\fq=1$, the averaging profile $u_n$ is the characteristic function, 
$1_{\sq}(v)$ (the dependence on $n$ is suppressed in the notation), 
of the rectangle 
\begin{equation*}
\Big(\big[-\sfrac{1}{2}, \sfrac{1}{2} \big]
\times \big[-\sfrac{1}{2}, \sfrac{1}{2} \big]^3\Big)\cap
\big(\veps_n^2\bbbz\times\veps_n\bbbz^3\big)
\end{equation*}
Note that
$
u_n= 1_{\sq}
$
is already normalized to have integral one. The $L^1$--$L^\infty$ norm
of $u_n$ (with mass zero) is exactly one. The $L^1$--$L^\infty$ norm,
with mass $m=1$,  of $u_n$ is bounded by 
$\exp\big\{\big|\big(\half,\half,\half,\half\big)\big|\big\}=e$.

\Item (b)
If $\ell_\nu\ne 0$, 
\begin{equation*}
\half\big|k_\nu+\ell_\nu\big|
\ge\half\big|\Re k_\nu+\ell_\nu\big|
\ge \sfrac{1}{6}\big(\pi+|\ell_\nu|\big)
\end{equation*}
so that, by Lemma \ref{lemPBSsinxoverx}.a, the denominator
\begin{equation*}
\sfrac{1}{\veps_{n,\nu}}\big|\sin\half\veps_{n,\nu}(k_\nu+\ell_\nu)\big|
\ge \sfrac{1}{\veps_{n,\nu}} \sfrac{\sqrt{2}}{\pi}
    \sfrac{1}{6}\veps_{n,\nu}\big(\pi+|\ell_\nu|\big)
= \sfrac{1}{3\sqrt{2}\pi}\big(\pi+|\ell_\nu|\big)
\end{equation*}
On the other hand, the numerator, by Lemma \ref{lemPBSsinxoverx}.b,
\begin{equation*}
\big|\sin\sfrac{1}{2} (k_\nu+\ell_\nu)\big|
=\big|\sin\big(\sfrac{1}{2} k_\nu\big)\big|
\le  |k_\nu|
\end{equation*}
As $\ell\ne 0$, there is at least one $\nu$ with $\ell_\nu\ne 0$. 
For each $\nu$ with $\ell_\nu\ne 0$ bound the factor
\begin{align*}
 \bigg|\frac{\sin\sfrac{1}{2} (k_\nu+\ell_\nu)}
   {\sfrac{1}{\veps_{n,\nu}}
        \sin\half\veps_{n,\nu}(k_\nu+\ell_\nu)}\bigg|
\le  \frac{3\sqrt{2}\pi|k_\nu|}{|\ell_\nu|+\pi}
\le  \frac{24|k_\nu|}{|\ell_\nu|+\pi}
\end{align*}
Bound the remaining factors, with $\nu$ having $\ell_\nu= 0$, by 
$\sfrac{24}{|\ell_\nu|+\pi}$ as in part (a).  All together
\begin{align*}
\prod_{\nu=0}^3 \bigg|\frac{\sin\sfrac{1}{2} (k_\nu+\ell_\nu)}
   {\sfrac{1}{\veps_{n,\nu}}\sin\half\veps_{n,\nu}(k_\nu+\ell_\nu)}\bigg|
\le \bigg[\prod_{\atop{0\le\nu\le 3}{\ell_\nu\ne 0}}|k_\nu| \bigg]
    \prod_{\nu=0}^3 \frac{24}{|\ell_\nu|+\pi}
\end{align*}

\Item (c)
For both $z=\half \veps_{n,\nu} k_\nu$ and $z=\half k_\nu$, 
$|z|^2\le\sfrac{\pi^2}{4}+1 <4$ so that, by Lemma \ref{lemPBSsinxoverx}.e,
\begin{equation*}
\big|\sfrac{\sin z}{z}-1\big|\le\half |z|^2
\end{equation*}
Using
$
\sfrac{a}{b}-1=\sfrac{(a-1)-(b-1)}{b}
$
we have, by Lemma \ref{lemPBSsinxoverx}.a, 
\begin{align*}
\bigg|\frac{\sin(\sfrac{1}{2} k_\nu)} {\sfrac{1}{\veps_{n,\nu}}\sin(\half\veps_{n,\nu} k_\nu)}-1\bigg|
\le\frac{\half|\half k_\nu|^2+\half|\sfrac{1}{2}\veps_{n,\nu} k_\nu|^2}
          {|\sin(\half \veps_{n,\nu} k_\nu)|/|\half\veps_{n,\nu} k_\nu|}
\le\sfrac{\pi}{\sqrt{2}}\sfrac{1}{8}(1+\veps_{n,\nu}^2)|k_\nu|^2
\le |k_\nu|^2
\end{align*}
Finally, using
\begin{equation*}
\prod_{\nu=0}^3 A_\nu -1
=(A_0-1)\prod_{\nu=1}^3 A_\nu
+(A_1-1)\prod_{\nu=2}^3 A_\nu
+(A_2-1) A_3
+(A_3-1)
\end{equation*}
we have, by Lemma \ref{lemPBSsinxoverx}.c,
\begin{equation*}
\big|u_n(k)-1\big|
\le 4^3|k_0|^2+4^{3-1}|k_1|^2+4|k_2|^2+|k_3|^2
\le 4^3 |k|^2
\end{equation*}

\Item (d)
The proof is the same as that for part (b), except that
the factor $\sin\big(\sfrac{1}{2} k_{\tilde \nu}\big)$ in the numerator
\begin{equation*}
\sin\sfrac{1}{2} (k_{\tilde \nu}+\ell_{\tilde \nu})
=(-1)^{\sfrac{\ell_{\tilde\nu}}{2\pi}}
    \sin\big(\sfrac{1}{2} k_{\tilde \nu}\big)
\end{equation*}
is pulled out of $u_n$, leaving $v_{n,\tilde\nu}$, rather than being bounded 
by $|k_{\tilde\nu}|$.

\Item (e)
By Lemma \ref{lemPBSsinxoverx}.c
\begin{align*}
\bigg|\Im\frac{\sin\half (k_\nu+\ell_\nu)}
   {\sfrac{1}{\veps_{n,\nu}}\sin\half\veps_{n,\nu} (k_\nu+\ell_\nu)}\bigg|
\le 6|\Im k_\nu| \frac{24}{|\ell_\nu|+\pi}
\end{align*}
In general, for any complex numbers $z_j=r_je^{i\th_j}$, $1\le j\le J$,
\begin{align*}
\Big|\Im\smprod_{j=1}^Jz_j\Big|
=\Big|\sin\Big(\smsum_{j=1}^J\th_j\Big)\Big|\smprod_{j=1}^Jr_j
\end{align*}
Repeatedly using
\begin{equation*}
\big|\sin(\th+\th')\big|
=\big|\sin(\th)\cos(\th')+\cos(\th)\sin(\th')\big|
\le|\sin(\th)|+|\sin(\th')|
\end{equation*}
we have
\begin{align*}
\Big|\Im\smprod_{j=1}^Jz_j\Big|
\le \sum_{j=1}^J\big|\sin(\th_j)\big|\smprod_{j=1}^Jr_j
=\sum_{j=1}^J\big|\Im z_j\big|\smprod_{j'\ne j}|z_{j'}|
\end{align*}
So
\begin{align*}
\big|\Im u_n(k+\ell)\big|
\le  \Big(\smsum_{\nu=0}^3  6|\Im k_\nu| \Big)
         \smprod_{\nu=0}^3  \sfrac{24}{|\ell_\nu|+\pi}
\le 12\ |\Im k|\smprod_{\nu=0}^3  \sfrac{24}{|\ell_\nu|+\pi}
\end{align*}
and 
\begin{align*}
\big|\Im u_n(k+\ell)^\fq\big|
\le 12\,\fq\ |\Im k|\ 
\Big[\smprod_{\nu=0}^3  \sfrac{24}{|\ell_\nu|+\pi}\Big]^\fq
\end{align*}

\Item (f)
Just apply Lemma \ref{lemPBSsinxoverx}.a,b separately to all of the 
numerators and denominators in the right hand side of \eqref{eqnPBSunagain}.
When $k$ is real, apply Lemma \ref{lemPBSsinxoverx}.d instead.

\end{proof}

\pagebreak[2]
\begin{lemma}\label{lemPBSuplusppties}
Let $\fq\in\bbbn$. 
Assume that $|\Re \fk_\nu|\le\sfrac{\pi}{L_\nu}$  and 
$|\Im \fk_\nu|\le \sfrac{2}{L_\nu}$ for each $0\le\nu\le3$.

\begin{enumerate}[label=(\alph*), leftmargin=*]
\item 
$
\big|u_+(\fk+\ell)\big|\le \smprod\limits_{\nu=0}^3 
                               \sfrac{24}{L_\nu|\ell_\nu|+\pi}
$ 
for all $\ell\in\hat\cB^+$. We use $|\ell_\nu|$ to denote the magnitude of
the smallest representative of $\ell_\nu$ in its equivalence class,
as an element of $\hat\cB^+$. 

\item 
$
\big|u_+(\fk+\ell)\big|\le 
        \Big[\smprod\limits_{\nu\in\cI_+} 
                                                L_\nu|\fk_\nu|\Big]
         \Big[\smprod\limits_{\nu=0}^3 
                               \sfrac{24}{L_\nu|\ell_\nu|+\pi}\Big]
$ 
for all $\ell\in\hat\cB^+$ with $\ell\ne 0$. Here $\cI_+$ is any subset
of $\set{\nu}{0\le\nu\le3,\ \ell_\nu\ne 0}$.

\item 
$\big|u_+(\fk)-1\big|\le 4^3 \sum\limits_{\nu=0}^3 L_\nu^2 |\fk_\nu|^2$.

\item 
If $\ell\in\hat\cB^+$ and $\ell_{\tilde\nu}\ne 0$
for some $0\le\tilde\nu\le3$, then
$u_+(\fk+\ell)=\sin\big(\sfrac{L_{\tilde\nu}}{2}\fk_{\tilde\nu}\big) v_{+,\tilde\nu}(\fk+\ell)$ with
$
\big|v_{+,\tilde\nu}(\fk+\ell)\big|\le \smprod\limits_{\nu=0}^3 
                               \sfrac{24}{L_\nu|\ell_\nu|+\pi}
$.

\item 
For all $\ell\in\hat\cB^+$,
\begin{align*}
\big|\Im u_+(\fk+\ell)\big|
&\le 12\ |\Im\bbbl \fk|
       \smprod_{\nu=0}^3  \sfrac{24}{L_\nu|\ell_\nu|+\pi}\\
\big|\Im u_+(\fk+\ell)^\fq\big|
&\le 12\,\fq\ |\Im \bbbl\fk|\ 
      \Big[\smprod_{\nu=0}^3  \sfrac{24}{L_\nu|\ell_\nu|+\pi}\Big]^\fq
\end{align*}
\end{enumerate}
\end{lemma}
\begin{proof}  
(a)
We may assume without loss of generality that $\ell_\nu$
is bounded, as a real number, by $\pi-\sfrac{\pi}{L_\nu}$. (Recall that $L$
is an odd natural number.) So we will always have 
$\big|\Re \fk_\nu+\ell_\nu\big|\le\pi$.
So,  by Lemma \ref{lemPBSsinxoverx}.c with $\veps=\sfrac{1}{L_\nu}$ and
$x+iy=L_\nu(\fk_\nu+\ell_\nu)$, 
\begin{align*}
\bigg|\frac{\sin\sfrac{L_\nu}{2} (\fk_\nu+\ell_\nu)}
   {L_\nu\sin\half(\fk_\nu+\ell_\nu)}\bigg|
&\le\left.\begin{cases}
              4 &\text{if $\ell_\nu=0$}\\
              \sfrac{8}{L_\nu|\Re \fk_\nu+\ell_\nu|}  
                  &\text{if $|L_\nu\ell_\nu|\ge 2\pi$}
          \end{cases}\right\}
\le \frac{24}{L_\nu|\ell_\nu|+\pi}
\end{align*}
since
$|\Re \fk_\nu|\le\sfrac{\pi}{L_\nu}$, 
$|\Im \fk_\nu|\le \sfrac{2}{L_\nu}$ 
and $\ell_\nu\in \sfrac{2\pi}{L_\nu}\bbbz$.

\Item (b)
If $\ell_\nu\ne 0$, 
\begin{equation*}
\half\big|\fk_\nu+\ell_\nu\big|
\ge\half\big|\Re \fk_\nu+\ell_\nu\big|
\ge \sfrac{1}{6}\big(\sfrac{\pi}{L_\nu}+|\ell_\nu|\big)
\end{equation*}
and $\half\big|\Re \fk_\nu+\ell_\nu\big|\le\sfrac{\pi}{2}$
so that, by Lemma \ref{lemPBSsinxoverx}.a, the denominator
\begin{equation*}
L_\nu\big|\sin\half(\fk_\nu+\ell_\nu)\big|
\ge \sfrac{\sqrt{2}}{6\pi}\big(\pi+L_\nu|\ell_\nu|\big)
\end{equation*}
On the other hand, the numerator, by Lemma \ref{lemPBSsinxoverx}.b,
\begin{equation*}
\big|\sin\sfrac{L_\nu}{2} (\fk_\nu+\ell_\nu)\big|
=\big|\sin\big(\sfrac{L_\nu}{2} \fk_\nu\big)\big|
\le 2\sfrac{L_\nu}{2} |\fk_\nu|
= L_\nu |\fk_\nu|
\end{equation*}
As $\ell\ne 0$, there is at least one $\nu$ with $\ell_\nu\ne 0$. 
Bound each factor with $\nu\in\cI_+$  by
\begin{align*}
 \bigg|\frac{\sin\sfrac{L_\nu}{2} (\fk_\nu+\ell_\nu)}
   {L_\nu\sin\half(\fk_\nu+\ell_\nu)}\bigg|
\le  \frac{6\pi}{\sqrt{2}} \frac{L_\nu|\fk_\nu|}{L_\nu|\ell_\nu|+\pi} 
\end{align*}
Bound the remaining factors by $\sfrac{24}{L_\nu|\ell_\nu|+\pi}$.

\Item (c)
has the same proof as that of Lemma \ref{lemPBSunppties}.c, with
$k_\nu$ replaced by $L_\nu\fk_\nu$ and $\sfrac{1}{\veps_{n,\nu}}$ 
replaced by $L_\nu$.

\Item (d)
The proof is similar to that for part (b), except that
the factor  $\sin\big(\sfrac{L_{\tilde \nu}}{2} \fk_{\tilde \nu}\big)$ in 
the numerator
\begin{equation*}
\sin\sfrac{L_{\tilde\nu}}{2} (\fk_{\tilde\nu}+\ell_{\tilde\nu})
=(-1)^{\sfrac{L_{\tilde\nu}\ell_{\tilde\nu}}{2\pi}}
   \sin\big(\sfrac{L_{\tilde\nu}}{2} \fk_{\tilde\nu}\big)
\end{equation*}
is pulled out of $u_+$, leaving $v_{+,\tilde\nu}$, rather than being bounded 
by $L_{\tilde\nu}|\fk_{\tilde\nu}|$. Also, each 
$\Big|\sfrac{\sin\sfrac{L_\nu}{2} (\fk_\nu+\ell_\nu)}
   {L_\nu\sin\half(\fk_\nu+\ell_\nu)}\Big|$ with $\nu\ne\tilde\nu$
is bounded by $\sfrac{24}{L_\nu|\ell_\nu|+\pi}$.

\Item (e)
By Lemma \ref{lemPBSsinxoverx}.c
\begin{align*}
\bigg|\Im\frac{\sin\sfrac{L_\nu}{2} (\fk_\nu+\ell_\nu)}
   {L_\nu\sin\half(\fk_\nu+\ell_\nu)}\bigg|
\le 8|\Im L_\nu \fk_\nu| \frac{24}{L_\nu|\ell_\nu|+\pi}
\end{align*}
The proof now continues as in Lemma \ref{lemPBSunppties}.e, just by substituting
$k_\nu\rightarrow L_\nu\fk_\nu$, $\ell_\nu\rightarrow L_\nu\ell_\nu$
and $\veps_{n,\nu}=\sfrac{1}{L_\nu}$.
\end{proof}

\begin{proposition}\label{propPBSAnppties}
Let $\fq\in\bbbn$.
There are constants $\Gam_2$, depending only on $\fq$, and $\Gam_3$,
depending only on $a$, such that
the following hold for all $L>\Gam_2$.

\begin{enumerate}[label=(\alph*), leftmargin=*]
\item  
On the domain 
$\set{k\in\bbbc\times\bbbc^3}{|\Im k_\nu|< 2 \text{ for each }
0\le\nu\le3}$ $\hat \fQ_n(k)$ is analytic,  and invariant 
under $k_\nu\rightarrow-k_\nu$ for each $0\le\nu\le 3$
and under $k_\nu\leftrightarrow k_{\nu'}$ for all $1\le\nu,\nu'\le 3$,
and obeys
\begin{equation*}
\sfrac{5}{6}a\le |\hat \fQ_n(k)| \le \sfrac{5}{4} a\qquad
\Re \hat \fQ_n(k)\ge \sfrac{a}{2}
\end{equation*}
If $k$ is real $\sfrac{5}{6}a\le \hat \fQ_n(k)\le a$.

\item
If $|\Re k_\nu|\le \pi$ and $|\Im k_\nu|\le 2$ for each $0\le\nu\le3$, then
\begin{equation*}
\big|\hat \fQ_n(k)-a_n\big|\le \sfrac{a}{3003}|k|^2\qquad\text{where}\qquad
a_n=a\sfrac{1-L^{-2}}{1-L^{-2n}}
\end{equation*}

\item 
 $\|\fQ_n\|_{m=1}\le\Gam_3$

\end{enumerate}
\end{proposition}

\begin{proof} 
(a)
Recall, from Remark \ref{remPBSqnft}.e, that 
$
\hat \fQ_n(k) = a\bigg[1+\sum\limits_{j=1}^{n-1}
 \smsum\limits_{\ell_j\in\hat\cB_j}
 \sfrac{1}{L^{2j}}u_j(k+\ell_j)^{2\fq}\bigg]^{-1}
$.
By Lemma \ref{lemPBSunppties}.a,
\begin{align*}
\sum_{j=1}^{n-1}
    \smsum_{\ell_j\in\hat\cB_j}\sfrac{1}{L^{2j}}|u_j(k+\ell_j)|^{2\fq}
&\le \sum_{j=1}^\infty \sfrac{1}{L^{2j}}
    \smsum_{\ell\in 2\pi\bbbz\times 2\pi\bbbz^3}\smprod_{\nu=0}^3
      \big(\sfrac{24}{|\ell_\nu|+\pi}\big)^{2\fq}
= \sfrac{c_\fq}{L^2-1}
\end{align*}
where 
$c_\fq 
= \Big[\sum_{j\in\bbbz}\big(\sfrac{24/\pi}{[2|j|+1}\big)^{2\fq}\Big]^{4}$.
Just pick $L$ large enough that $\sfrac{c_\fq}{L^2-1}<\sfrac{1}{5}$ and 
use that, if $|z|\le\sfrac{1}{5}$
\begin{equation*}
\Re\sfrac{1}{1+z}
=\sfrac{\Re(1+\overline{z})}{|1+z|^2}
\ge\sfrac{4/5}{(6/5)^2}
=\sfrac{20}{36}
\end{equation*}

\Item (b)
Using $O(|k|^2)$ to denote any function that is bounded
by a constant, depending only on $\fq$,
\begin{align*}
1+\sum_{j=1}^{n-1}
    \smsum_{\ell_j\in\hat\cB_j}\sfrac{1}{L^{2j}}u_j(k+\ell_j)^{2\fq}
&=\sum_{j=0}^{n-1}\sfrac{1}{L^{2j}}
  +\sum_{j=1}^{n-1}\sfrac{1}{L^{2j}}[u_j(k)^{2\fq}-1] \\
  &\hskip1.5in
  +\sum_{j=1}^{n-1}
    \smsum_{\atop{\ell_j\in\hat\cB_j}{\ell\ne0}}
              \sfrac{1}{L^{2j}}u_j(k+\ell_j)^{2\fq}
    \\
&\le \sfrac{1-L^{-2n}}{1-L^{-2}}
      +\sum_{j=1}^{n-1}\sfrac{1}{L^{2j}}O\big(|k|^2\big)\qquad
    \text{by Lemma \ref{lemPBSunppties}.b,c}\cr
&\le \sfrac{1-L^{-2n}}{1-L^{-2}}
      +\sfrac{1}{L^2}O\big(|k|^2\big)
\end{align*}
So
\begin{align*}
\hat \fQ_n(k)
=a\big[\sfrac{1-L^{-2n}}{1-L^{-2}}+\sfrac{1}{L^2}O\big(|k|^2\big)\big]^{-1}
=a_n\big[1+\sfrac{1}{L^2}\sfrac{1-L^{-2}}{1-L^{-2n}}O\big(|k|^2\big)\big]^{-1}
\end{align*}
and it suffices to choose $\Gam_2$ large enough that
\begin{equation*}
\sfrac{9}{10}\le \sfrac{1-L^{-2}}{1-L^{-2n}}\le\sfrac{11}{10}\quad
\Big|\sfrac{1}{L^2}\sfrac{1-L^{-2}}{1-L^{-2n}}O\big(|k|^2\big)\Big|
\le \sfrac{10}{22}\sfrac{1}{3003}|k|^2\quad
\Big|\sfrac{1}{L^2}\sfrac{1-L^{-2}}{1-L^{-2n}}O\big(|k|^2\big)\Big|
\le \half
\end{equation*}
for all allowed $k$'s and $L$'s.

\Item (c)
follows immediately from part (a) and \cite[Lemma \lemBOlonelinfty.b]{Bloch}
with $\cX_\fin = \cX_\crs = \cX_0^{(n)}$.
\end{proof}

We now define operators $Q_{n,\nu}^{(\pm)}$ and $Q_{+,\nu}^{(\pm)} 
$ so that the next remark holds.
\begin{remark}\label{remPBSunderivAlg}
Let $0\le\nu\le3$. We have
\begin{equation*}
\partial_\nu Q_n^*= Q_{n,\nu}^{(+)} \partial_\nu\qquad
\partial_\nu Q_n= Q_{n,\nu}^{(-)} \partial_\nu\qquad
\partial_\nu Q^*= Q_{+,\nu}^{(+)} \partial_\nu\qquad
\partial_\nu Q= Q_{+,\nu}^{(-)} \partial_\nu\qquad
\end{equation*}
If $S:\cH_n\rightarrow\cH_n$ 
and $T:\cH_0^{(n)}\rightarrow\cH_0^{(n)}$ 
are linear operators that are translation invariant with respect to $\cX_n$
and $\cX_0^{(n)}$, respectively, then
\begin{equation*}
Q_{n,\nu}^{(-)}S Q_{n,\nu}^{(+)}=Q_nSQ_n^*\qquad
Q_{+,\nu}^{(-)}T Q_{+,\nu}^{(+)}=QTQ^*
\end{equation*}
\end{remark}

To  prepare for the definitions, recall that
the forward derivatives of $\al\in \cH_j^{(n)}$   
are defined by
\begin{equation}\label{eqnPBSforwardDeriv}
(\partial_\nu \al)(x) =  
\sfrac{1}{\veps_{j,\nu}} \big[\al(x+ \veps_{j,\nu} e_\nu)-\al(x)\big]
\end{equation}
where $e_\nu$ is a unit vector in the $\nu^{\mathrm{th}}$ direction. 
The Fourier transforms
\begin{alignat}{5}\label{eqnPBSderivft}
\big(\widehat{\partial_\nu \phi}\big)(p)
   &=2ie^{i\veps_{n,\nu} p_\nu/2}\ 
  \sfrac{\sin(\veps_{n,\nu} p_\nu/2)}{\veps_{n,\nu}} \hat\phi(p) &
\quad&\text{for all }\phi\in\cH_n &
&\text{and }p\in\hat\cX_n   \notag\\
\big(\widehat{\partial_\nu \psi}\big)(k)
   &=2ie^{ik_\nu/2}\ \sin(k_\nu/2) \hat\psi(k) &
&\text{for all }\psi\in\cH_0^{(n)} &
&\text{and }k\in\hat\cX_0^{(n)}\\
\big(\widehat{\partial_\nu \th}\big)(\fk)
   &=2ie^{iL_\nu \fk_\nu/2}\ 
          \sfrac{\sin(L_\nu \fk_\nu/2)}{L_\nu} \hat\th(\fk) &
&\text{for all }\th\in\cH_{-1}^{(n+1)} &\quad
&\text{and }\fk\in\hat\cX_{-1}^{(n+1)}  \notag
\end{alignat}
Set
\begin{alignat*}{5}
u_{n,\nu}^{(+)}(p)&=\prod_{0\le\nu'\le3\atop\nu'\ne\nu} 
      \frac{\sin\half p_{\nu'}}
  {\sfrac{1}{\veps_{n,\nu'}}\sin\half\veps_{n,\nu'}p_{\nu'}} &
u_{+,\nu}^{(+)}(k)
&=\prod_{0\le\nu'\le3\atop\nu'\ne\nu} 
      \frac{\sin\half L_{\nu'} k_{\nu'}}{L_{\nu'}\sin\half k_{\nu'}}
\\
u_{n,\nu}^{(-)}(p)&=
\frac{\sin\half p_\nu}
  {\sfrac{1}{\veps_{n,\nu}}\sin\half\veps_{n,\nu} p_\nu}
\prod_{\nu'=0}^3 
      \frac{\sin\half p_{\nu'}}
  {\sfrac{1}{\veps_{n,\nu'}}\sin\half\veps_{n,\nu'} p_{\nu'}} &\quad
u_{+,\nu}^{(-)}(k)
&=\frac{\sin\half L_\nu k_\nu}{L_\nu\sin\half k_\nu}
\prod_{\nu'=0}^3 
      \frac{\sin\half L_{\nu'} k_{\nu'}}{L_{\nu'}\sin\half k_{\nu'}}
\end{alignat*}
and
\begin{alignat*}{3}
\ze_{n,\nu}^{(+)}(k,\ell_n)
    &=e^{i\veps_{n,\nu} (k+\ell_n)_\nu/2 }e^{-ik_\nu/2 }
    \cos\half\ell_{n,\nu} &\quad
\ze_{+,\nu}^{(+)}(\fk,\ell)
    &=e^{i(\fk+\ell)_\nu/2 }e^{-iL_\nu \fk_\nu/2 }
    \cos\half L_\nu\ell_\nu\\
\ze_{n,\nu}^{(-)}(k,\ell_n)
    &=e^{ik_\nu/2 } e^{-i\veps_{n,\nu} (k+\ell_n)_\nu/2 }
    \cos\half\ell_{n,\nu} &\quad
\ze_{+,\nu}^{(-)}(\fk,\ell)
    &=e^{iL_\nu\fk_\nu/2 } e^{-i(\fk+\ell)_\nu/2 }
    \cos\half L_\nu\ell_\nu
\end{alignat*}
Define the operators $Q_{n,\nu}^{(+)}:\cH_0^{(n)}\rightarrow \cH_n$
and $Q_{n,\nu}^{(-)}: \cH_n\rightarrow \cH_0^{(n)}$ by
\begin{equation}\label{eqnPBSqnplusminus}
\begin{split}
\big(\widehat{Q_{n,\nu}^{(+)} \psi}\big)(k+\ell_n)
   &=\ze_{n,\nu}^{(+)}(k,\ell_n)u_{n,\nu}^{(+)}(k+\ell_n)u_n(k+\ell_n)^{\fq-1} \hat\psi(k)
\\
\big(\widehat{Q_{n,\nu}^{(-)} \phi}\big)(k)
   &=\sum_{\ell_n\in\hat\cB_n} \ze_{n,\nu}^{(-)}(k,\ell_n)
     u_{n,\nu}^{(-)}(k+\ell_n)
     u_n(k+\ell_n)^{\fq-1} \hat\phi(k+\ell_n)\cr
\end{split}
\end{equation}
and the operators 
   $Q_{+,\nu}^{(+)}:\cH_{-1}^{(n+1)}\rightarrow \cH_0^{(n)}$
and 
  $Q_{+,\nu}^{(-)}: \cH_0^{(n)}\rightarrow \cH_{-1}^{(n+1)}$ by
\begin{equation}\label{eqnPBSqplusplusminus}
\begin{split}
\big(\widehat{Q_{+,\nu}^{(+)} \th}\big)(\fk+\ell)
   &=\ze_{+,\nu}^{(+)}(\fk,\ell)u_{+,\nu}^{(+)}(\fk+\ell)u_+(\fk+\ell)^{\fq-1} \hat\th(\fk)
\\
\big(\widehat{Q_{+,\nu}^{(-)} \psi}\big)(\fk)
  &=\sum_{\ell\in\hat\cB^+} \ze_{+,\nu}^{(-)}(\fk,\ell)u_{+,\nu}^{(-)}(\fk+\ell)
     u_+(\fk+\ell)^{\fq-1} \hat\psi(\fk+\ell)
\end{split}
\end{equation} 

\begin{proof}[Proof of Remark \ref{remPBSunderivAlg}]
For the ``$Q^*_n$'' and ``$Q_n$'' cases, it suffices to observe that
\begin{align*}
\big(2ie^{i\veps_{n,\nu} (k+\ell)_\nu/2}\ 
    \sfrac{\sin(\veps_{n,\nu} (k+\ell)_\nu/2)}{\veps_{n,\nu}}\big)
                u_n(k+\ell)
&=\ze_{n,\nu}^{(+)}(k,\ell)u_{n,\nu}^{(+)}(k+\ell)
          \big(2ie^{ik_\nu/2}\ \sin(k_\nu/2)\big)
\end{align*}
and
\begin{align*}
\big(2ie^{ik_\nu/2}\ \sin(k_\nu/2)\big)
                u_n(k+\ell)
&=\ze_{n,\nu}^{(-)}(k,\ell)u_{n,\nu}^{(-)}(k+\ell)
    \big(2ie^{i\veps_{n,\nu} (k+\ell)_\nu/2}\ 
    \sfrac{\sin(\veps_{n,\nu} (k+\ell)_\nu/2)}{\veps_{n,\nu}}\big)
\end{align*}
and
$$
\ze_{n,\nu}^{(-)}(k,\ell)u_{n,\nu}^{(-)}(k+\ell)
\ze_{n,\nu}^{(+)}(k,\ell)u_{n,\nu}^{(+)}(k+\ell)
=u_n(k+\ell)^2
$$
for all $k$, $\ell$, $\nu$. We remark that 
``$Q_{n,\nu}^{(-)} SQ_{n,\nu}^{(+)}=Q_nSQ_n^*$''
should not be surprising since $\partial_\nu Q_nSQ_n^*
=Q_{n,\nu}^{(-)} SQ_{n,\nu}^{(+)}\partial_\nu$
and $Q_nSQ_n^*$ is translation invariant on the
unit scale and so commutes with $\partial_\nu$.
The proof for the ``$Q^*$'' and ``$Q$'' cases are virtually identical.
\end{proof}

\begin{lemma}\label{lemPBSunderiv}
Let $0\le\nu\le3$, $\ell\in\hat\cB^+$ and $\ell_n\in\hat\cB_n$.

\begin{enumerate}[label=(\alph*), leftmargin=*]
\item $\ze_{n,\nu}^{(+)}(k,\ell_n)u_{n,\nu}^{(+)}(k+\ell_n)$ 
and $\ze_{n,\nu}^{(-)}(k,\ell_n)u_{n,\nu}^{(-)}(k+\ell_n)$ are entire in $k$
and \newline
$\ze_{+,\nu}^{(+)}(\fk,\ell)u_{+,\nu}^{(+)}(\fk+\ell)$ 
and $\ze_{+,\nu}^{(-)}(\fk,\ell)u_{+,\nu}^{(-)}(k+\ell)$ are entire in $\fk$

\item
Assume that $|\Re k_{\nu'}|\le\pi$, $|\Im k_{\nu'}|\le 2$,
$|\Re \fk_{\nu'}|\le\sfrac{\pi}{L_{\nu'}}$  and 
$|\Im \fk_{\nu'}|\le \sfrac{2}{L_{\nu'}}$  for each $0\le\nu'\le3$.
Then
\begin{align*}
\big|\ze_{n,\nu}^{(+)}(k,\ell_n)u_{n,\nu}^{(+)}(k+\ell_n)\big|
&\le e\prod_{0\le\nu'\le3\atop\nu'\ne\nu}  
                     \sfrac{24}{|\ell_{n,\nu'}|+\pi}\\
\big|\ze_{+,\nu}^{(+)}(\fk,\ell)u_{+,\nu}^{(+)}(\fk+\ell)\big|
&\le e\prod_{0\le\nu'\le3\atop\nu'\ne\nu}  
             \sfrac{24}{L_\nu|\ell_{\nu'}|+\pi}\\
\big|\ze_{n,\nu}^{(-)}(k,\ell_n)u_{n,\nu}^{(-)}(k+\ell_n)\big|
&\le \sfrac{24 e}{|\ell_{n,\nu}|+\pi}
            \smprod_{\nu'=0}^3 \sfrac{24}{|\ell_{n,\nu'}|+\pi} \\
\big|\ze_{+,\nu}^{(-)}(\fk,\ell)u_{+,\nu}^{(-)}(\fk+\ell)\big|
&\le \sfrac{24 e}{L_\nu|\ell_\nu|+\pi}
            \smprod_{\nu'=0}^3 \sfrac{24}{L_{\nu'}|\ell_{\nu'}|+\pi}
\end{align*}

\item 
There is a constant $\Gam_4$, depending only on $\fq$, 
such that $\big\|Q^{(\pm)}_{n,\nu}\big\|_{m=1}\le\Gam_4$.
\end{enumerate}
\end{lemma}

\begin{proof}
(a)
The proof is virtually identical to that of  Remark \ref{remPBSqnft}.d.

\Item (b) 
The proof is virtually identical to that of 
            Lemmas \ref{lemPBSunppties}.a and \ref{lemPBSuplusppties}.a.

\Item (c)
By \eqref{eqnPBSqnplusminus}, the Fourier transform of 
          $Q^{(\pm)}_{n,\nu}$ is
          $\ze_{n,\nu}^{(\pm)}(k,\ell_n)u_{n,\nu}^{(\pm)}(k+\ell_n)
            u_n(k+\ell_n)^{\fq-1}$, which by part (b) and 
            Lemma \ref{lemPBSunppties}.a, is bounded in magnitude by
\begin{equation*}
e\prod_{0\le\nu'\le3\atop\nu'\ne\nu}  
                     \sfrac{24}{|\ell_{n,\nu'}|+\pi}
   \smprod_{\nu=0}^3 \big(\sfrac{24}{|\ell_\nu|+\pi}\big)^{\fq-1}
\end{equation*}
As $\fq>2$, the claim now follows by \cite[Lemma \lemBOlonelinfty.c]{Bloch}.
\end{proof}

\begin{remark}\label{remPBOqfour}
The principle obstruction to allowing $\fq=1$ arises when a differential
operator $\partial_\nu$ is intertwined with the block spin averaging 
operator $Q_n$, as happens in Remark \ref{remPBSunderivAlg}.
See, for example, the proof of Lemma \ref{lemPBSunderiv}.c. 
We use the condition $\fq>1$ starting at 
Lemma \ref{lemPOCakmatrix} in \S\ref{secPOcovariance} and in \S\ref{secPOgreens}, \ref{secPOleadingOrder}. (See Lemma \ref{lemPOGSnnuppties}.) We use the condition
$\fq>2$ in Proposition \ref{propPOLmain} and Lemma \ref{lemPBSunderiv}.c. 
\end{remark}

\newpage
\section{Differential Operators}\label{secPOdiffOps}

In \cite[Definition \defHTbackgrounddomaction.a]{PAR1} we associated to
an operator $h_0$ on $L^2\big(\bbbz^3/L_\sp\bbbz^3\big)$ the operators
\begin{equation}\label{eqnPDOdndef}
D_n = L^{2n}\ \bbbl_*^{-n}
         \big(\bbbone - e^{-\oh_0} -e^{-\oh_0} \partial_0\big)\bbbl_*^n
\end{equation}
Here $\partial_0$ is the forward time derivative of \eqref{eqnPBSforwardDeriv}.
In this chapter we assume that $h_0$ is the periodization (see
\cite[\S\secBOperiodization]{Bloch}) of a translation
invariant operator $\bh_0$ on $L^2\big(\bbbz^3\big)$ whose Fourier transform
$\hat\bh_0(\bp)$ 
\begin{itemize}[leftmargin=*, topsep=2pt, itemsep=0pt, parsep=0pt]
\item
 is entire in $\bp$ and invariant under 
$\bp_\nu\rightarrow-\bp_\nu$ for each $1\le\nu\le3$
\item
 is nonnegative when $\bp$ is real and is strictly positive
 when $\bp\in\bbbr^3\setminus{2\pi\bbbz^3}$
\item
 obeys $\hat\bh_0(\bZ)
     =\sfrac{\partial\,\hat\bh_0\ }{\partial\bp_\nu}(\bZ)=0$
 for $1\le\nu\le3$ and has strictly positive Jacobian matrix
 $H=\Big[ 
   \sfrac{\partial^2\,\hat\bh_0\ }{\partial\bp_\mu\partial\bp_\nu}(\bZ)   
  \Big]_{1\le\mu,\nu\le3}$. 
\end{itemize}
\begin{remark}\label{remPDOftDn}
\ 
\begin{enumerate}[label=(\alph*), leftmargin=*]
\item
The operator $D_n$ is the periodization of a translation invariant
operator $\bD_n$, acting on $L^2\big(\veps_n\bbbz\times\veps_n^2\bbbz^3\big)$, 
whose Fourier transform is
\begin{align*}
\hat\bD_n(p)
&=\half\veps_n^2p_0^2 e^{-\hat\bh_0(\veps_n\bp)}
\bigg[\frac{\sin\half\veps_n^2p_0}{\half\veps_n^2p_0}\bigg]^2
+\bp^2\frac{1-e^{-\hat\bh_0(\veps_n\bp)}}{\veps_n^2\bp^2}
-ip_0\,e^{-\hat\bh_0(\veps_n\bp)}\frac{\sin\veps_n^2p_0}{\veps_n^2p_0}\cr
\end{align*}
with $p=(p_0,\bp)\in\bbbc\times\bbbc^3$.

\item 
$\hat\bD_n(p)$ is entire in $p$ and invariant under 
$\bp_\nu\rightarrow-\bp_\nu$ for each $1\le\nu\le3$.

\item 
$\hat\bD_n(p)$ has nonnegative real part when $p$ is real.
\end{enumerate}
\end{remark}
\begin{proof} (a) follows from \eqref{eqnPINTlft} and the observation, 
by \eqref{eqnPBSderivft},
that the Fourier transform of $\partial_0$, on $\bbbz$, is
\begin{align*}
2ie^{i k_0/2}\ \sin( k_0/2)
= -2 \sin^2(k_0/2) +i\sin(k_0)
\end{align*}

\noindent (b) and (c) are obvious.
\end{proof}

\begin{lemma}\label{lemPDOhatSzeroppties}
There are constants $\gam_1$, $\Gam_5$ and a function
$\bmm(c)>0$  that depend only on $\hat\bh_0$  and 
in particular are independent of $n$ and $L$, such that the following hold.

\begin{enumerate}[label=(\alph*), leftmargin=*]
\item 
For all $p\in\bbbr\times\bbbr^3$, 
\begin{equation*}
\big|\hat\bD_n(p)\big|
\ge\gam_1\big(|p_0|+\smsum_{\nu=1}^3|\bp_\nu|^2\big)
\end{equation*}
We use $|p_0|$, $|\bp_\nu|$ and $|\bp|$ to refer to the magnitudes of the 
smallest representatives of $\,p_0\in\bbbc$, $\bp_\nu\in\bbbc$ and $\bp\in\bbbc^3$
in $\bbbc/\sfrac{2\pi}{\veps_n^2}\bbbz$, $\bbbc/\sfrac{2\pi}{\veps_n}\bbbz$
and $\bbbc^3/\sfrac{2\pi}{\veps_n}\bbbz^3$, respectively.

\item  
For all $p\in\bbbc\times\bbbc^3$ with $\veps_n^2p_0,\veps_n\bp$
having modulus less than one,
\begin{equation*}
\hat\bD_n(p)
 = -ip_0+\half\veps_n^2p_0^2 
     +\half\! \smsum_{\nu,\nu'=1}^3 \!\!H_{\nu,\nu'}\bp_\nu\bp_{\nu'}
     +O\Big(\veps_n|\bp|^3+\veps_n^4|p_0|^3\Big)
\end{equation*}
The higher order part $O(\ \cdot\ )$ is uniform in $n$ and $L$.

\item
We have, 
for all $p\in\bbbc\times\bbbc^3$ with $\veps_n^2|\Im p_0|\le 1$
and $\veps_n|\Im\bp|\le 1$,
\begin{equation*}
\big|\hat\bD_n(p)\big|
\le\Gam_5\big(|p_0|+\smsum_{\nu=1}^3|\bp_\nu|^2\big)
\end{equation*}
 and 
\begin{align*}
\big|\sfrac{\partial^{\ell_i}\hfill}{\partial p_\nu^{\ell_i}}
      \hat\bD_n(p)\big|
&\le\Gam_5\begin{cases}
                  \sfrac{1}{[1+|p_0|+|\bp|^2]^{\ell_i-1}} & 
                             \text{if $\nu=0$, $\ell_i=1,2$}\\
                   \noalign{\vskip0.05in}
                  \sfrac{1}{[1+|p_0|+|\bp|^2]^{\ell_i/2-1}} & 
                             \text{if $1\le\nu\le3$, $1\le\ell_i\le4$}
           \end{cases}
\end{align*}

\item  
For all $c>0$ and $p\in\bbbc\times\bbbc^3$, with $|p_0|+|\bp|\ge c$ 
and $|\Im p|\le \bmm(c)$,
\begin{equation*}
\big|\hat\bD_n(p)\big|
\ge\gam_1\,\big(|p_0|+\smsum_{\nu=1}^3|\bp_\nu|^2\big)
\end{equation*}

\item
For all $c>0$ and all $p$ in the set
\begin{align*}
&\set{p\in\bbbc\times\bbbc^3}{|\Im p|\le \bmm(c),\ |\bp|\ge c\!}\\
&\hskip2in\cup
\set{p\in\bbbc\times\bbbc^3}{|\Im p|\le \bmm(c),\ |\veps_n p_0|\ge c\!}
\end{align*}
we have
\begin{equation*}
\Re\hat\bD_n(p)
\ge\gam_1\,\big(\veps_n^2|p_0|^2+\smsum_{\nu=1}^3|\bp_\nu|^2\big)
\end{equation*}
\end{enumerate}
\end{lemma}
\begin{proof}  (a) By the hypotheses on $\hat h$,
\begin{equation*}
\frac{1-e^{-\hat\bh_0(\veps_n\bp)}}{\veps_n^2\bp^2}
=\frac{\hat\bh_0(\veps_n\bp)+O(\hat\bh_0(\veps_n\bp)^2)}{\veps_n^2\bp^2}
=\frac{\half\veps_n^2\bp\cdot H\bp+O(|\veps_n\bp|^3)}{\veps_n^2\bp^2}
\end{equation*}
for $\veps_n\bp$ is a real neighbourhood of $\bZ$. This is strictly positive
and bounded away from 0 on some real neighbourhood of $0$, uniformly in $\veps_n$. 
Since $\hat h$ is continuous and strictly positive on 
$\bbbr^3\setminus 2\pi\bbbz^3$, there is a constant $\gam'_1>0$, independent
of $\veps_n$, such that
\begin{equation*}
\frac{1-e^{-\hat\bh_0(\veps_n\bp)}}{\veps_n^2|\bp|^2}\ge\gam'_1
\qquad\text{and}\qquad
e^{-\hat\bh_0(\veps_n\bp)}\ge\gam'_1
\end{equation*}
for all $\bp\in\bbbr^3$. The claim now follows from Lemma \ref{lemPBSsinxoverx}.a.

\Item (b) Expanding
$\sfrac{\sin z}{z}=1+O(|z|^2)$, for $|z|\le 1$, and
\begin{align*}
e^{-\hat h(\veps_n\bp)} 
    &= 1- \half \veps_n^2 \smsum_{\nu,\nu'=1}^3 H_{\nu,\nu'}\bp_\nu\bp_{\nu'}
                    +O\Big(\big(\veps_n|\bp|\big)^3\Big)
                    && &\qquad\text{for}\quad\veps_n|\bp|\le 1
\end{align*}
gives, using Remark \ref{remPDOftDn}.a,
\begin{align*}
\hat\bD_n(p)
&=\Big[\!-\!ip_0+\half\veps_n^2p_0^2 \Big]
     \Big[1+O\Big(\big(\veps_n|\bp|\big)^2\!\!+\!\big(\veps_n^2p_0\big)^2\Big)\Big]
+\half\!\!\smsum_{\nu,\nu'=1}^3 \!\!\! H_{\nu,\nu'}\bp_\nu\bp_{\nu'}\!
                    +O\Big(\veps_n|\bp|^3\Big)\\
&=-ip_0+\half\veps_n^2p_0^2 
+\half\!\smsum_{\nu,\nu'=1}^3 \!\! H_{\nu,\nu'}\bp_\nu\bp_{\nu'}
          +O\Big(\veps_n|\bp|^3+\veps_n^2|p_0||\bp|^2+\veps_n^4|p_0|^3\Big)
\end{align*}
The claim now follows from
\begin{equation*}
\veps_n^2|p_0||\bp|^2\le\big(\veps_n^4|p_0|^3\big)^{1/3}
                              \big(\veps_n|\bp|^3\big)^{2/3}
\end{equation*}

\Item (c) Since  $\hat\bD_n(p)$ is periodic with respect
to $\sfrac{2\pi}{\veps_n^2}\bbbz\times\sfrac{2\pi}{\veps_n}\bbbz^3$,
we may assume that $\veps_n^2\Re p_0$ and each $\veps_n\Re\bp_\nu$, 
$1\le\nu\le3$ is bounded in magnitude by $\pi$. 
By Lemma \ref{lemPBSsinxoverx}.b, $\big|\sfrac{\sin z}{z}\big|\le 2$
for all $z\in\bbbc$ with $|\Im z|\le 1$. Since $\veps_n\bp$ runs over a
compact set (independently of $n$ and $L$), $e^{-\hat h(\veps_n\bp)}$
is bounded. So 
\begin{equation*}
\bigg|\half\veps_n^2p_0^2 e^{-\hat\bh_0(\veps_n\bp)}
\bigg[\frac{\sin\half\veps_n^2p_0}{\half\veps_n^2p_0}\bigg]^2\bigg|
\ ,\ 
\bigg|ip_0\,e^{-\hat\bh_0(\veps_n\bp)}\frac{\sin\veps_n^2p_0}{\veps_n^2p_0}\bigg|
\ \le\ \const |p_0|
\end{equation*}
and
\begin{equation*}
\bigg|\frac{1-e^{-\hat\bh_0(\veps_n\bp)}}{\veps_n^2}\bigg|
\le\const \sfrac{|\veps_n\bp|^2}{\veps_n^2}
\le\const|\bp|^2
\end{equation*}
This gives the bound on $\big|\hat\bD_n(p)\big|$.

The bounds 
\begin{alignat*}{5}
\sfrac{\partial\hfill}{\partial p_0}
   \veps_n^2\Big[\sfrac{\sin(\half\veps_n^2 p_0)}{\half\veps_n^2}\Big]^2
    &=2\sin(\veps_n^2 p_0)=O(1) &
\sfrac{\partial\hfill}{\partial p_0}
   \Big[\sfrac{\sin(\veps_n^2 p_0)}{\veps_n^2}\Big]
    &=\cos(\veps_n^2 p_0)=O(1)\\
\sfrac{\partial^2\hfill}{\partial p_0^2}
   \veps_n^2\Big[\sfrac{\sin(\half\veps_n^2 p_0)}{\half\veps_n^2}\Big]^2
    &=2\veps_n^2 \cos(\veps_n^2 p_0)=O(\veps_n^2) & 
\sfrac{\partial^2\hfill}{\partial p_0^2}
   \Big[\sfrac{\sin(\veps_n^2 p_0)}{\veps_n^2}\Big]
    &=-\veps_n^2\sin(\veps_n^2 p_0)=O(\veps_n^2)\\
\sfrac{\partial\hfill}{\partial \bp_\nu}
   \Big[\sfrac{1}{\veps_n^2}e^{-\hat\bh_0(\veps_n\bp)}\Big]
    &=O(|\bp|) &
\sfrac{\partial^2\hfill}{\partial \bp_\nu^2}
   \Big[\sfrac{1}{\veps_n^2}e^{-\hat\bh_0(\veps_n\bp)}\Big]
    &=O(1+\veps_n^2|\bp|^2) =O(1)\\
\sfrac{\partial^3\hfill}{\partial \bp_\nu^3}
   \Big[\sfrac{1}{\veps_n^2}e^{-\hat\bh_0(\veps_n\bp)}\Big]
    &=O(\veps_n^2|\bp|\!+\!\veps_n^4|\bp|^3)&
\sfrac{\partial^4\hfill}{\partial \bp_\nu^4}
   \Big[\sfrac{1}{\veps_n^2}e^{-\hat\bh_0(\veps_n\bp)}\Big]
    &=O(\veps_n^2\!+\!\veps_n^4|\bp|^2\!+\!\veps_n^6|\bp|^4)\\
    &=O(\veps_n)&
    &=O(\veps_n^2)
\end{alignat*}
together with 
\begin{equation*}
\veps_n^2\le\sfrac{\const}{1+|p_0|+|\bp|^2}\qquad
\veps_n^2|p_0|\le\const\qquad
\end{equation*}
yield the bounds on the derivatives.

\Item (d)
Write $p=\sP+i\sQ$ with 
$\sP=(\sP_0,\sbP),\ \sQ=(\sQ_0,\sbQ)\in\bbbr\times\bbbr^3$. 
We may choose $\bmm(c)$ sufficiently small that, 
if $|\sP+i\sQ|\ge c$ and $|\sQ|\le \bmm(c)$,
then 
\begin{align*}
&|\sP_0+i\sQ_0|+\!\smsum_{\nu=1}^3\!|\sbP_\nu+i\sbQ_\nu|^2 \\
      &\hskip1in\le \Big(|\sP_0|+\!\smsum_{\nu=1}^3\!|\sbP_\nu|^2\Big)
           +\Big(|\sQ_0|+\!\smsum_{\nu=1}^3 |\sbQ_\nu|^2\Big) \\
      &\hskip3in\le 2\Big(|\sP_0|+\!\smsum_{\nu=1}^3|\sbP_\nu|^2\Big)
\end{align*} 
If $\gam_1>0$ is chosen small enough and $\Gam_5$ is chosen large enough,
then, for all such $\sP,\sQ$, we have, by parts (a) and (c),
\begin{align*}
\big|\hat\bD_n(\sP)\big|
&\ge 4\gam_1\big(|\sP_0|+\smsum_{\nu=1}^3|\sbP_\nu|^2\big)\\
&\ge 2\gam_1\big(|\sP_0+i\sQ_0|+\smsum_{\nu=1}^3|\sbP_\nu+i\sbQ_\nu|^2\big)\\
\Big|\hat\bD_n(\sP+i\sQ)
     -\hat\bD_n(\sP)\Big|
&\le2\,
 \Gam_5\Big(1+|\sP_0+i\sQ_0|+\smsum_{\nu=1}^3|\sbP_\nu+i\sbQ_\nu|^2\Big)|\sQ|
\end{align*}
Recalling that
$
|\sP_0+i\sQ_0|+\smsum_{\nu=1}^3|\sbP_\nu+i\sbQ_\nu|^2
\ge\min\Big\{\sfrac{c}{2},\sfrac{c^2}{4}\Big\}
$,
it now suffices to choose $\bmm(c)$ small enough that
\begin{equation*}
2\,\Gam_5 \bmm(c)
\le\sfrac{1}{2}\gam_1\min\Big\{\sfrac{c}{2},\sfrac{c^2}{4} \Big\}
\qquad\text{and}\qquad
2\,\Gam_5 \bmm(c)\le\sfrac{1}{2}\gam_1
\end{equation*}

\Item (e) 
Again write $p=\sP+i\sQ$ with 
$\sP=(\sP_0,\sbP),\ \sQ=(\sQ_0,\sbQ)\in\bbbr\times\bbbr^3$. 
Since  $\hat\bD_n(\sP+i\sQ)$ is periodic with respect
to $\sP\in\sfrac{2\pi}{\veps_n^2}\bbbz\times\sfrac{2\pi}{\veps_n}\bbbz^3$,
we may assume that $\big|\veps_n^2 \sP_0\big|\le\pi$ and 
$\big|\veps_n\sbP_\nu\big|\le\pi$, for each $1\le\nu\le3$. 
If the constant $\gam_1$ was chosen small enough, then, as in part (a), 
\begin{align*}
\Re\hat\bD_n(\sP)
&= \half\veps_n^2\sP_0^2 e^{-\hat\bh_0(\veps_n\sbP)}
      \bigg[\frac{\sin\half\veps_n^2\sP_0}{\half\veps_n^2\sP_0}\bigg]^2
      +\sbP^2\frac{1-e^{-\hat\bh_0(\veps_n\sbP)}}{\veps_n^2\sbP^2}
\ge 2\gam_1 \big(\veps_n^2\sP_0^2+|\sbP|^2\big)
\end{align*} 
Hence it suffices to prove that it is possible to choose $\bmm=\bmm(c)$ so 
that 
\begin{equation}\label{eqnPDOimshifta}
\Big|\Re\hat\bD_n(\sP+i\sQ)
         -\Re\hat\bD_n(\sP)\Big|\le \gam_1|\sbP|^2
\qquad\text{when $|\sbP|\ge c$ and $|\sQ|\le \bmm$}
\end{equation}
and that
\begin{equation}\label{eqnPDOimshiftb}
\Big|\Re\hat\bD_n(\sP+i\sQ)
   -\Re\hat\bD_n(\sP)\Big|\le \gam_1 c^2
\qquad\text{when $|\sbP|\le c$, $|\veps_n \sP_0|\ge c$ and $|\sQ|\le \bmm$}
\end{equation}
This is a consequence of the following bounds on the derivatives of the 
real parts of the three terms making up $\hat\bD_n(\sP+i\sQ)$ in 
Remark \ref{remPDOftDn}.a. For the first term,
\begin{align*}
&\bigg|\sfrac{d\hfill}{dt} \veps_n^2(\sP_0+it\sQ_0)^2 
        e^{-\hat\bh_0(\veps_n\sbP)}
\bigg[\frac{\sin\half\veps_n^2(\sP_0+it\sQ_0)}{\half\veps_n^2(\sP_0+it\sQ_0)}\bigg]^2
\bigg|\\
&\hskip0.1in\le \const \big[\veps_n^2|\sP_0+i\sQ_0|\,|\sQ_0|
            +\veps_n^2|\sP_0+i\sQ_0|^2\ \veps_n^2|\sQ_0|\big]\\
&\hskip0.1in\le\const\, \bmm\\
&\sfrac{d\hfill}{dt}\Re \veps_n^2(\sP_0+it\sQ_0)^2 
        \Big[e^{-\hat\bh_0(\veps_n\sbP+it\veps_n\sbQ)}-e^{-\hat\bh_0(\veps_n\sbP)}\Big]
\bigg[\frac{\sin\half\veps_n^2(\sP_0+it\sQ_0)}{\half\veps_n^2(\sP_0+it\sQ_0)}\bigg]^2
\bigg|_{t=0}=0
\\
&\bigg|\sfrac{d^2\hfill}{dt^2}\veps_n^2(\sP_0+it\sQ_0)^2 
       \Big[e^{-\hat\bh_0(\veps_n\sbP+it\veps_n\sbQ)}-e^{-\hat\bh_0(\veps_n\sbP)}\Big]
\bigg[\frac{\sin\half\veps_n^2(\sP_0+it\sQ_0)}{\half\veps_n^2(\sP_0+it\sQ_0)}\bigg]^2
\bigg|\\
&\hskip0.1in\le \const \big[\veps_n^2|\sQ_0|^2
   +\veps_n^2|\sP_0\!+\!i\sQ_0|\,|\sQ_0|\big(\veps_n|\sbQ|+\veps_n^2|\sQ_0|\big)
   +\veps_n^2|\sP_0\!+\!i\sQ_0|^2\ \big(\veps_n|\sbQ|+\veps_n^2|\sQ_0|\big)^2\big]
 \\
&\hskip0.1in\le\const\, \bmm^2
\end{align*}
For the second term,
\begin{align*}
&\sfrac{d\hfill}{dt}
   \Re \sfrac{1-e^{-\hat\bh_0(\veps_n\sbP+it\veps_n\sbQ)}}{\veps_n^2}
\bigg|_{t=0}=0
\\
&\bigg|\sfrac{d^2\hfill}{dt^2}
   \sfrac{1-e^{-\hat\bh_0(\veps_n\sbP+it\veps_n\sbQ)}}{\veps_n^2}
\bigg|
\le \const\sfrac{1}{\veps_n^2}\big(\veps_n|\sbQ|\big)^2
\le \const\, \bmm^2
\end{align*}
For the third term,
\begin{align*}
&\bigg|\sfrac{d\hfill}{dt}\Re
   i(\sP_0+it\sQ_0)\,e^{-\hat\bh_0(\veps_n\sbP)}
         \frac{\sin\veps_n^2(\sP_0+it\sQ_0)}{\veps_n^2(\sP_0+it\sQ_0)}
\bigg|\\
&\hskip0.3in\le \const|\sQ_0|
   +\const |\sP_0+i\sQ_0|\ \veps_n^2|\sQ_0|\\
&\hskip0.3in\le\const\, \bmm\\
&\bigg|\sfrac{d\hfill}{dt}\Re
   i(\sP_0+it\sQ_0)\,
        \Big[e^{-\hat\bh_0(\veps_n\sbP+it\veps_n\sbQ)}-e^{-\hat\bh_0(\veps_n\sbP)}\Big]
         \frac{\sin\veps_n^2(\sP_0+it\sQ_0)}{\veps_n^2(\sP_0+it\sQ_0)}
\bigg|_{t=0}\bigg|\\
&\hskip0.3in=\bigg|\Re
   i\sP_0\, \Big[\sfrac{d\hfill}{dt}\hat\bh_0(\veps_n\sbP+it\veps_n\sbQ)\Big]_{t=0}
         e^{-\hat\bh_0(\veps_n\sbP)}
         \frac{\sin\veps_n^2(\sP_0)}{\veps_n^2\sP_0}\bigg|\\
&\hskip0.3in=\bigg|
   \veps_n \sP_0\  \big(\sbQ\cdot\hat\bh_0'(\veps_n\sbP)\big)\ 
         e^{-\hat\bh_0(\veps_n\sbP)}
         \frac{\sin\veps_n^2(\sP_0)}{\veps_n^2\sP_0}\bigg|\\
&\hskip0.3in\le\const \veps_n^2|\sP_0|\, |\sbP|\, \bmm\\
&\hskip0.3in\le\const |\sbP|\, \bmm\\
&\bigg|\sfrac{d^2\hfill}{dt^2}
   i(\sP_0+it\sQ_0)\,
        \Big[e^{-\hat\bh_0(\veps_n\sbP+it\veps_n\sbQ)}-e^{-\hat\bh_0(\veps_n\sbP)}\Big]
         \frac{\sin\veps_n^2(\sP_0+it\sQ_0)}{\veps_n^2(\sP_0+it\sQ_0)}\bigg|\\
&\hskip0.3in\le \const \big[
    |\sQ_0|\big(\veps_n|\sbQ|+\veps_n^2|\sQ_0|\big)
   +|\sP_0\!+\!i\sQ_0|\ \big(\veps_n|\sbQ|+\veps_n^2|\sQ_0|\big)^2\big]
 \\
&\hskip0.3in\le\const\, \bmm^2
\end{align*}
Now choose $\bmm=\bmm(c)$ small enough that
\eqref{eqnPDOimshifta} and \eqref{eqnPDOimshiftb} are satisfied.
\end{proof}

\newpage
\section{The Covariance}\label{secPOcovariance}

The covariance for the fluctuation integral in \cite{PAR2} is
\begin{equation*}
C^{(n)}=(\sfrac{a}{L^2}Q^*Q+\De^{(n)})^{-1} 
\end{equation*}
where
\begin{equation*}
\De^{(n)}=\left.
   \begin{cases}
      \big(\bbbone+\fQ_n Q_n D_n^{-1} Q_n^*\big)^{-1}\fQ_n & \text{if $n\ge 1$}\\
      \noalign{\vskip0.05in}
      D_0       & \text{if $n=0$}
    \end{cases}\right\}
    :\cH_0^{(n)}\rightarrow\cH_0^{(n)}
\end{equation*}
See \cite[(\eqnHTcn) and (\eqnHTden)]{PAR1}.
In Lemma \ref{lemPOCDenppties} we study the properties of $\De^{(n)}$
and in Corollary \ref{corPOCsquareroot} we study the properties of $C^{(n)}$
and its square root.

\begin{remark}\label{remPOCde}
Let $n\ge 1$.

\begin{enumerate}[label=(\alph*), leftmargin=*]
\item
 The operator $\De^{(n)}$ is the periodization of a 
translation invariant operator $\bDe^{(n)}$, acting on $L^2\big(\bbbz\times\bbbz^3\big)$, whose Fourier transform
is
\begin{align*}
\hat\bDe^{(n)}(k)
&=\hat\fQ_n(k)\Big(1+\hat\fQ_n(k)\sum_{\ell\in\hat\cB_n} u_n(k+\ell)^{2\fq}\ 
                  \hat\bD_n^{-1}(k+\ell)\Big)^{-1}\\
&=\frac{\hat \fQ_n(k)\ \hat\bD_n(k)}
          {\hat\bD_n(k)
           +\hat \fQ_n(k)\sum_{\ell}
             u_n(k+\ell)^{2\fq}\ 
          \hat\bD_n^{-1}(k+\ell)
           \hat\bD_n(k)}
\end{align*}
with $k\in\bbbc\times\bbbc^3$, where $u_n(p)$ and $\hat\cB_n$ were defined 
in parts (b) and (e) of Remark \ref{remPBSqnft}, respectively.

\item 
$\hat\bDe^{(n)}(k)$ is invariant under 
$\bk_\nu\rightarrow-\bk_\nu$  for each $1\le\nu\le3$.

\item 
$\hat\bDe^{(n)}(k)$ has nonnegative real part when $k$ is real.
\end{enumerate}
\end{remark}

\begin{lemma}\label{lemPOCDenppties}
There are constants 
$\mm_1>0$, $\gam_2$, $\Gam_6$ and $\Gam_7$,
such that, for $L>\Gam_2$, the following hold.

\begin{enumerate}[label=(\alph*), leftmargin=*]
\item 
$\hat\bDe^{(n)}(k)$ is analytic on $|\Im k| < 3\mm_1$.

\item  
For all $k\in\bbbc\times\bbbc^3$ with $|\Im k|\le 3\mm_1$.
\begin{align*}
\hat\bDe^{(n)}(k)
&=-ik_0+\big(\sfrac{1}{a_n}+\sfrac{\veps_n^2}{2}\big)k_0^2
    +\half\! \smsum_{\nu,\nu'=1}^3 \!\!H_{\nu,\nu'}\bk_\nu\bk_{\nu'}
        +O\big(|k|^3\big)\displaybreak[0]\\
\hat\bDe^{(n)}\!(k)\,\hat\bD_n^{-1}(k)
    &=1+\sfrac{ik_0}{a_n} +O\big(|k|^2\big)
\end{align*} 
The higher order part $O(\ \cdot\ )$ is uniform in $n$ and $L$.

\item  
$\big|\hat\bDe^{(n)}(k)\big|\le 2a$ and
$\big|\sfrac{\partial\hfill}{\partial k_\nu}\hat\bDe^{(n)}(k)\big|,
\big|\sfrac{\partial^2\hfill}{\partial k_\nu\partial k_{\nu'}}
                                       \hat\bDe^{(n)}(k)\big|\le \Gam_7$
for all $0\le\nu,\nu'\le3$ and $k\in\bbbc\times\bbbc^3$ with $|\Im k|<3\mm_1$.

\item 
There is a function $\rho(c)>0$, which is defined for all 
$c>0$ and which depends only on $\mm_1$, $\fq$, $\hat\bh_0$ and $a$ 
and, in particular, is independent of $n$ and $L$, such that
\begin{equation*}
\Re\hat\bDe^{(n)}(k)\ge \rho(c)
\end{equation*}
for all $k\in\bbbc\times\bbbc^3$ with $|k|\ge c$ and $|\Im k|\le 3\mm_1$.

\item 
For all $k\in\bbbc\times\bbbc^3$ with $|\Im k|\le 3\mm_1$ 
and $|\Re k_\nu|\le\pi$ for all $0\le\nu\le3$,
\begin{equation*}
\big|\hat\bDe^{(n)}(k)\big|\ge \gam_2 \big|\hat\bD_n(k)\big| 
\end{equation*}

\item
$\hat\bD_n^{-1}(p)\,\hat\bDe^{(n)}(p)$ is analytic on $|\Im p|\le 3\mm_1$.
Furthermore, for all $p\in\bbbc\times\bbbc^3$ with $|\Im p|\le 3\mm_1$,
\begin{equation*}
\big|\hat\bD_n^{-1}(p)\,\hat\bDe^{(n)}(p)\big|
\le \frac{\Gam_6}{1+|p_0|+\smsum_{\nu=1}^3 |\bp_\nu|^2}
\end{equation*}
Here, as usual, $|p_0|$ and $|\bp_\nu|$ refer to the magnitudes of the 
smallest representatives of $p_0\in\bbbc$ and $\bp_\nu\in\bbbc$
in $\bbbc/\sfrac{2\pi}{\veps_n^2}\bbbz$ and $\bbbc/\sfrac{2\pi}{\veps_n}\bbbz$
respectively.

\item
$\big|\sfrac{\partial\hfill}{\partial \bk_\nu}\hat\bDe^{(n)}(k)\big|
\le \Gam_7|\bk_\nu|$ for all $1\le\nu\le3$ and $k\in\bbbc\times\bbbc^3$ with 
$|\Im k|<3\mm_1$.
\end{enumerate}
\end{lemma}
\begin{proof} 
We first prove part (b).
Using that
\begin{itemize}[leftmargin=*, topsep=2pt, itemsep=0pt, parsep=0pt]
\item 
  $u_n(k)^{2\fq}=1+O\big(|k|^2\big)$   by Lemma \ref{lemPBSunppties}.b
\item
  $\hat \fQ_n(k)=a_n+O\big(|k|^2\big)$   by Proposition \ref{propPBSAnppties}.b
\item
   $|\hat\bD_n(k)|  \le\const\big(|k_0|+ |\bk|^2\big)$    
                      by Lemma \ref{lemPDOhatSzeroppties}.c
\end{itemize}
and that, for $\ell\ne 0$,
\begin{itemize}[leftmargin=*, topsep=2pt, itemsep=0pt, parsep=0pt]
\item
  $|u_n(k+\ell)|^{2\fq}\le\const|k|^{2\fq}
    \prod_{\nu=0}^3\sfrac{1}{(|\ell|_\nu|+\pi)^{2\fq}}$
   by Lemma \ref{lemPBSunppties}.a
\item
   $| \hat\bD_n^{-1}(k+\ell)|
   \le\const$ by Lemma \ref{lemPDOhatSzeroppties}.d
\end{itemize}
 
\noindent
we obtain, by Remark \ref{remPOCde}.a and Lemma \ref{lemPDOhatSzeroppties}.b, 
\begin{align*}
\hat\bDe^{(n)}(k)
&=\frac{\hat \fQ_n(k)\hat\bD_n(k)}{
   \hat \fQ_n(k) u_n(k)^{2\fq} +\hat\bD_n(k)
    +O\big(|k|^3\big)}\displaybreak[0]\\
&=\hat\bD_n(k)\frac{a_n+O\big(|k|^2\big)}{
   a_n +\hat\bD_n(k)
     +O\big(|k|^2\big)}\displaybreak[0]\\
&=\hat\bD_n(k)
     \Big\{1-\sfrac{1}{a_n}\hat\bD_n(k)
              +O\big(|k|^2\big)\Big\}\displaybreak[0]\\
&=-ik_0+\half\veps_n^2k_0^2 
     +\half\! \smsum_{\nu,\nu'=1}^3 \!\!H_{\nu,\nu'}\bk_\nu\bk_{\nu'}
    -\sfrac{1}{a_n}\hat\bD_n(k)^2
    +O\big(|k|^3\big)\displaybreak[0]\\
&=-ik_0+\big(\sfrac{1}{a_n}+\sfrac{\veps_n^2}{2}\big)k_0^2
    +\half\! \smsum_{\nu,\nu'=1}^3 \!\!H_{\nu,\nu'}\bk_\nu\bk_{\nu'}
        +O\big(|k|^3\big)
\end{align*}
This also shows that, in a neighbourhood of the origin,
$\hat\bDe^{(n)}(k)$ is analytic and bounded in magnitude by $2a$.
For the second expansion,
\begin{align*}
\hat\bDe^{(n)}\!(k)\,\hat\bD_n^{-1}(k)
&=\frac{a_n+O\big(|k|^2\big)}{
   a_n +\hat\bD_n(k)
     +O\big(|k|^2\big)}
=1-\sfrac{1}{a_n}\hat\bD_n(k)
              +O\big(|k|^2\big) \\
&=1+\sfrac{ik_0}{a_n} +O\big(|k|^2\big)
\end{align*}

\Item (a), (c) 
We first prove the analyticity and the bound $\big|\hat\bDe^{(n)}(k)\big|\le 2a$
of part (c). We have already done so for a neighbourhood of the origin.
So it suffices to consider $|k|>c_0$, for some suitably small $c_0$. 
Since $\hat\bDe^{(n)}$ is periodic with respect 
to $2\pi\bbbz\times 2\pi\bbbz^3$, it suffices to consider $k$ in the set
\begin{equation*}
M(\mm_1)=\set{k\in\bbbc\times\bbbc^3}{
         |\Re k_\nu|\le\pi\text{ for all $0\le\nu\le3$},
         |\Im k|\le 3\mm_1,\ 
         |k|>c_0 }
\end{equation*}

Recall, from Remarks \ref{remPBSqnft}.d and \ref{remPDOftDn}.b, that $u_n(p)$ and  
$\hat\bD_n(p)$ are entire. If $\mm_1$ is small enough, then, by 
Lemma \ref{lemPDOhatSzeroppties}.d, the functions $k\mapsto \hat\bD_n(k+\ell)$,
$\ell\in\hat\cB_n$, $n\ge 1$, have no zeroes in $M(\mm_1)$. 
Hence each term in the infinite sum
\begin{equation*}
1+\hat \fQ_n(k)\sum_{\ell\in\hat\cB_n} u_n(k+\ell)^{2\fq}\ 
                  \hat\bD_n^{-1}(k+\ell)
\end{equation*}
is analytic and it suffices to prove that, for $k\in M(\mm_1)$, the sum 
converges uniformly
and $\big|1+\hat \fQ_n(k)\sum_{\ell\in\hat\cB_n} u_n(k+\ell)^{2\fq}\ 
                  \hat\bD_n^{-1}(k+\ell)\big|\ge\sfrac{6}{10}$.

By Proposition \ref{propPBSAnppties}.a, Remark \ref{remPBSqnft}.d and 
Lemma \ref{lemPDOhatSzeroppties}.d, there is an $l>0$ such that
\begin{align}\label{eqnPOCtailDe}
|\hat \fQ_n(k)|\sum_{\ell\in\hat\cB_n\atop |\ell|\ge l} \big|u_n(k+\ell)^{2\fq}\ 
                  \hat\bD_n^{-1}(k+\ell)\big|
&\le \sfrac{6}{5}a\sum_{\ell\in\hat\cB_n\atop|\ell|\ge l}\sfrac{1}{\gam_1\pi} \smprod_{\nu=0}^3\big(\sfrac{24}{|\ell_\nu|+\pi}\big)^{2\fq}
\le\sfrac{\const}{l}
\le\sfrac{2}{10} 
\end{align}
Hence we have uniform convergence and
\begin{equation*}
\Big|1+\hat \fQ_n(k)\sum_{\ell\in\hat\cB_n} u_n(k+\ell)^2\ 
                  \hat\bD_n^{-1}(k+\ell)\Big|
\ge \Big|1+\hat \fQ_n(k)\sum_{\ell\in\hat\cB_n\atop |\ell|<l} u_n(k+\ell)^2\ 
                  \hat\bD_n^{-1}(k+\ell)\Big|-\sfrac{2}{10}
\end{equation*}
For real $k$ and every $\ell\in\hat\cB_n$, the real part of
$
\hat \fQ_n(k) u_n(k+\ell)^{2\fq}\ 
                  \hat\bD_n^{-1}(k+\ell)
$
is nonnegative.
Consequently, if $\mm_1$ is small enough and $|\Im k|\le 3\mm_1$
\begin{equation}\label{eqnPOCsmallellDe}
\Re \hat \fQ_n(k)\sum_{\ell\in\hat\cB_n\atop |\ell|<l} u_n(k+\ell)^{2\fq}\ 
                  \hat\bD_n^{-1}(k+\ell)\ge-\sfrac{2}{10}
\end{equation}
since, for all $k\in M(\mm_1)$ and $\ell\in\hat\cB_n$ with $|\ell|<l$, 
\begin{itemize}[leftmargin=*, topsep=2pt, itemsep=0pt, parsep=0pt]
\item
$\hat\bD_n(k+\ell)$ is
bounded away from zero (by Lemma \ref{lemPDOhatSzeroppties}.d) and has 
bounded first derivative (by Lemma \ref{lemPDOhatSzeroppties}.c)
and 
\item 
$u_n(k+\ell)$ and $\hat \fQ_n(k)$ are bounded with bounded 
first derivatives (by analyticity, Remark \ref{remPBSqnft}.d and 
Proposition \ref{propPBSAnppties}.a). 
\end{itemize} 
So, by \eqref{eqnPOCsmallellDe},
\begin{equation*}
 \Big|1+\hat \fQ_n(k)\sum_{\ell\in\hat\cB_n\atop |\ell|<l} u_n(k+\ell)^{2\fq}\ 
                  \hat\bD_n^{-1}(k+\ell)\Big|\ge\sfrac{8}{10}
\end{equation*}
All of this is uniform in $n$ and $L$. 

Shrinking $\mm_1$ by a factor of $2$, the bounds 
$\big|\sfrac{\partial\hfill}{\partial k_\nu}\hat\bDe^{(n)}(k)\big|,
\big|\sfrac{\partial^2\hfill}{\partial k_\nu\partial k_{\nu'}}
                                       \hat\bDe^{(n)}(k)\big|\le \Gam_7$,
with $\Gam_7$ being the maximum of $4a$ divided by the original $3\mm_1$
(for first order derivatives) and $16a$ divided by the square of the
original $3\mm_1$ (for second order derivatives),
follow by the Cauchy integral formula.

\Item (d)
Denote the real and imaginary parts of the numerator and denominator by
\begin{alignat*}{3}
R_n&=\Re \hat \fQ_n(k) &
R_d&=\Re\Big[1+\hat \fQ_n(k)\sum_{\ell\in\hat\cB_n} u_n(k+\ell)^{2\fq}\ 
                  \hat\bD_n^{-1}(k+\ell)\Big]
\\
I_n&=\Im \hat \fQ_n(k) &
I_d&=\Im\Big[1+\hat \fQ_n(k)\sum_{\ell\in\hat\cB_n} u_n(k+\ell)^{2\fq}\ 
                  \hat\bD_n^{-1}(k+\ell)\Big]
\end{alignat*}
 By Proposition \ref{propPBSAnppties}.a and \eqref{eqnPOCtailDe}, 
\eqref{eqnPOCsmallellDe},
\begin{equation*}
R_n\ge\sfrac{a}{2}\qquad
R_d\ge\sfrac{6}{10}
\end{equation*}
By Proposition \ref{propPBSAnppties}.a, Remark \ref{remPBSqnft}.d 
and Lemma \ref{lemPDOhatSzeroppties}.d, 
\begin{equation*}
I_d\le \Big|\hat \fQ_n(k)\sum_{\ell\in\hat\cB_n} u_n(k+\ell)^{2\fq}\ 
                  \hat\bD_n^{-1}(k+\ell)\Big|
\le  a\tilde I_d(c)
\end{equation*}
where $\tilde I_d(c)=\sfrac{6}{5}a \sum_{\ell\in\hat\cB_n}
    \sfrac{1}{\gam_1\min\{c/2,c^2/4\}}
   \smprod_{\nu=0}^3 \sfrac{24^{2\fq}}{(|\ell_\nu|+\pi)^{2\fq}}
$.
When $k$ is real, $\hat \fQ_n(k)$ is real. By analyticity and 
Proposition \ref{propPBSAnppties}.a, the first order derivatives of $\hat \fQ_n(k)$
are bounded. So if $\mm_1$ is small enough,
$
|I_n|\le \sfrac{2}{10\tilde I_d(c)}
$.
So
\begin{align*}
\Re\hat\bDe^{(n)}(k)=\sfrac{R_nR_d+I_nI_d}{R_d^2+I_d^2}
   \ge\sfrac{a(3/10)-a(2/10)}{(1+a\tilde I_d(c))^2}
   =\sfrac{a}{10(1+a\tilde I_d(c))^2}
\end{align*}

\Item (e) By Remark \ref{remPBSqnft}.d, Proposition \ref{propPBSAnppties}.a and 
Lemma \ref{lemPDOhatSzeroppties}.d,
\begin{align*}
&\Big|1+\hat \fQ_n(k)\sum_{\ell\in\hat\cB_n} u_n(k+\ell)^{2\fq}\ 
                  \hat\bD_n^{-1}(k+\ell)\Big|\\
&\hskip0.8in\le \big|\hat \fQ_n(k)\ u_n(k)^{2\fq}\ \hat\bD_n^{-1}(k)\big|
 +1 +\sum_{\bZ\ne\ell\in\hat\cB_n} \big|\hat \fQ_n(k)\ u_n(k+\ell)^{2\fq}\ 
                  \hat\bD_n^{-1}(k+\ell)\big|\\
&\hskip0.8in\le \sfrac{6}{5}a\ \big(\sfrac{24}{\pi}\big)^{8\fq}\ \big|\hat\bD_n^{-1}(k)\big|
 +1 +\sfrac{6}{5}a\sum_{\bZ\ne\ell\in\hat\cB_n}\sfrac{1}{\gam_1\pi} \smprod_{\nu=0}^3\big(\sfrac{24}{|\ell_\nu|+\pi}\big)^{2\fq}
\end{align*}
and so, by Lemma \ref{lemPDOhatSzeroppties}.c, 
\begin{align*}
&\big|\hat\bD_n(k)\big|
\Big|1+\hat \fQ_n(k)\sum_{\ell\in\hat\cB_n} u_n(k+\ell)^{2\fq}\ 
                  \hat\bD_n^{-1}(k+\ell)\Big|\\
&\hskip0.75in\le \sfrac{6}{5}a\big(\sfrac{24}{\pi}\big)^{8\fq}
   + \Gam_5\ \big(\pi+1+3(\pi+1)^2\big)
\Big[1 +\sfrac{6}{5}a\smsum_{\bZ\ne\ell\in\hat\cB_n}\sfrac{1}{\gam_1\pi} \smprod_{\nu=0}^3\big(\sfrac{24}{|\ell_\nu|+\pi}\big)^{2\fq}\Big]
\end{align*}

\Item (f) 
By part (c) of this lemma, Remark \ref{remPDOftDn}.b and 
Lemma \ref{lemPDOhatSzeroppties}.d, it suffices to consider $p=k$ with
$|\Im k|\le 3\mm_1$ and $|k|<c_0<\pi$ where $c_0$ is sufficiently small. Now
\begin{align*}
&\hat\bD_n^{-1}(k)\hat\bDe^{(n)}(k) \\
&\hskip0.5in=\frac{\hat \fQ_n(k)}{
   \hat \fQ_n(k) u_n(k)^{2\fq} +\hat\bD_n(k)\big[
    1+\hat \fQ_n(k)\sum_{\bZ\ne\ell\in\hat\cB_n} u_n(k+\ell)^{2\fq}\ 
                  \hat\bD_n^{-1}(k+\ell)\big]}
\end{align*}
By Proposition \ref{propPBSAnppties}.a, Remark \ref{remPBSqnft}.d, 
Lemma \ref{lemPBSunppties}.b
and parts (c) and (d) of Lemma \ref{lemPDOhatSzeroppties}, we may choose $c_0>0$ 
so that
\begin{equation*}
\Big|\hat\bD_n(k)\big[
    1+\hat \fQ_n(k)\sum_{\bZ\ne\ell\in\hat\cB_n} u_n(k+\ell)^{2\fq}\ 
                  \hat\bD_n^{-1}(k+\ell)\big]\Big|
<\sfrac{a}{5}
\end{equation*}
and
\begin{equation*}
\big|\hat \fQ_n(k)\big|\,\big|u_n(k)^{2\fq}-1|<\sfrac{a}{5}
\end{equation*}
for all $|k|<c_0$ with $|\Im k|\le 3\mm_1$. That does it.

\Item (g) Fix any $1\le\nu\le3$. Observe that
\begin{equation*}
\sfrac{\partial\hat\bDe^{(n)}}{\partial k_\nu}(k)
=\bk_\nu\int_0^1 
\sfrac{\partial^2\hat\bDe^{(n)}}{\partial k_\nu^2}\big(k(t)\big)\ dt
\qquad\text{with}\quad
k(t)_{\nu'}=\begin{cases}
               k_{\nu'}& \text{if $\nu'\ne \nu$}\\
               tk_\nu, & \text{if $\nu=\nu'$}
              \end{cases}
\end{equation*}
since $\sfrac{\partial\hat\bDe^{(n)}}{\partial k_\nu}\big(k(0)\big)=0$,
by Remark \ref{remPOCde}.b. Now
apply the bound on the second derivative from part (c).
\end{proof}

We next consider the ``resolvents''
\begin{equation*}
R^{(n)}_\ze=\big(\ze\bbbone -\sfrac{a}{L^2} Q^*Q-\De^{(n)}\big)^{-1}
\end{equation*}
in preparation for studying $C_n$ and $\sqrt{C_n}$.
We shall use \cite[Lemma \lemBOfnbnd]{Bloch}, with
\begin{equation*}
\veps_T= 1\qquad
\veps_X=1 \qquad
L_T=  L^2 \qquad
L_X= L \qquad
\cL_T=L_\tp\qquad
\cL_X=L_\sp
\end{equation*}
and
\begin{equation*}
\cX_\fin=\cX_{0}^{(n)}\quad
\cZ_\fin=\bbbz\times\bbbz^3 \quad
\cX_\crs=\cX_{-1}^{(n+1)}\quad
\cZ_\crs=L^2\bbbz\times L\bbbz^3 \quad
\cB=\cB^+
\end{equation*}
We wish to apply \cite[Lemma \lemBOfnbnd]{Bloch}, with $A\rightarrow 
D_n=\sfrac{a}{L^2} Q^*Q+\De^{(n)}$ (scaled).
By \cite[Lemmas \lemBOfourier.b and \lemBOifkervar.b]{Bloch}, 
for each $\ell, \ell'\in\hat\cB=\hat\cB^+$,
\begin{equation}\label{eqnPOCftCinverse}
a_{\fk}(\ell,\ell')\rightarrow d_{n,\fk}(\ell,\ell')
      =\sfrac{a}{L^2} u_+(\fk+\ell)^\fq u_+(\fk+\ell')^\fq  
                        +\de_{\ell,\ell'}\hat\De^{(n)}(\fk+\ell)
\end{equation}
Here we are denoting
\begin{itemize}[leftmargin=*, topsep=2pt, itemsep=0pt, parsep=0pt]
\item 
momenta dual to the $L$--lattice 
by $\fk\in 
\big(\bbbr/\sfrac{2\pi}{L^2}\bbbz\big)
            \times \big(\bbbr^3/\sfrac{2\pi}{L}\bbbz^3\big)$ and
\item 
momenta dual to the unit lattice 
$
\bbbz\times \bbbz^3$ by
$k\in 
\big(\bbbr/2\pi\bbbz\big) \times \big(\bbbr^3/2\pi\bbbz^3\big)$
and decompose $k=\fk+\ell$ or $k=\fk+\ell'$ with $\fk$
in a fundamental cell for 
$
\big(\bbbr/\sfrac{2\pi}{L^2}\bbbz\big)
            \times \big(\bbbr^3/\sfrac{2\pi}{L}\bbbz^3\big)$
and $\ell,\ell'\in \big(\sfrac{2\pi}{L^2}\bbbz/2\pi\bbbz\big)
    \times\big(\sfrac{2\pi}{L}\bbbz^3/2\pi\bbbz^3\big)
    =\hat\cB^+$.
\end{itemize}
We also use $d_{n,\fk}$ to denote the $\hat\cB^+\times \hat\cB^+$ 
matrix $\big[d_{n,\fk}(\ell,\ell')\big]_{\ell,\ell'\in\hat\cB^+}$.
Observe, by Remark \ref{remPBSqnft}.d and Lemma \ref{lemPOCDenppties}.a, that
$d_{n,\fk}\big(\ell,\ell'\big)$ is 
analytic in the strip $|\Im \fk|<3\mm_1$.

Let $\big[v_\ell\big]_{\ell\in\hat\cB^+}$ and 
 $\big[w_\ell\big]_{\ell\in\hat\cB^+}$ be any vectors in 
$L^2(\hat\cB^+)$. Then, if $\fk=\bbbl^{-1}(k)$, 
\begin{align*}
&\big<\bar v,d_{n,\fk}w\big>\\
&=
\sfrac{a}{L^2}
\Big[\hskip-2pt\smsum_{\ell\in\hat\cB^+}\hskip-5pt
           u_+(\bbbl^{-1}(k)\!+\!\ell)^\fq \overline{v_\ell}\Big]
\Big[\hskip-2pt\smsum_{\ell\in\hat\cB^+}\hskip-5pt
              u_+(\bbbl^{-1}(k)\!+\!\ell)^\fq w_\ell\Big]
 + \hskip-4pt\sum_{\ell\in\hat\cB^+}\hskip-3pt
            \hat\De^{(n)}(\bbbl^{-1}(k)\!+\!\ell)\overline{v_\ell}w_\ell\\
&=\sum_{\ell,\ell'\in\hat\cB_1} 
      \overline{v_{\bbbl(\ell)}}\ d_{n,k}^{(s)}(\ell,\ell')\  w_{\bbbl(\ell')}
\end{align*}
where the ``scaled'' matrix
\begin{equation*}
d_{n,k}^{(s)}(\ell,\ell')
= \sfrac{a}{L^2}  u_+\big(\bbbl^{-1}(k\!+\!\ell)\big)^\fq
                      u_+\big(\bbbl^{-1}(k\!+\!\ell')\big)^\fq  
             +\de_{\ell,\ell'}\hat\De^{(n)}\big(\bbbl^{-1}(k\!+\!\ell)\big)
\end{equation*}

\begin{lemma}\label{lemPOCakmatrix}
There are constants $\mm_2,\la_0,\Gam_8>0$, 
such that, for all $L>\Gam_2$ and
$k\in\bbbc\times\bbbc^3$ with $|\Im k|< 3\mm_2$,
the following hold.

\begin{enumerate}[label=(\alph*), leftmargin=*]
\item 
Write $\fk=\bbbl^{-1}(k)$. For both the operator and 
(matrix) $\ell^1$--$\ell^\infty$ norms 
\begin{equation*}
\|d_{n,\fk}\|\le\Gam_8\qquad 
\|d_{n,k}^{(s)}\|\le\Gam_8
\end{equation*}

\item 
Let $\la\in\bbbc$ be within a distance $\sfrac{\la_0}{L^2}$
of the negative real axis. Then the resolvent
\begin{equation*}
\big\|{\big(\la\bbbone-d_{n,k}^{(s)}\big)}^{-1}\big\|\le \Gam_8 L^2
\end{equation*}
This is true for both the operator and (matrix) $\ell^1$--$\ell^\infty$ norms.
\end{enumerate}
\end{lemma}
\begin{proof}  
Since $d_{n,\fk+p}(\ell,\ell')=d_{n,\fk}(\ell+p,\ell'+p)$, 
for all $p,\ell,\ell'\in\hat\cB^+$, we may always assume that
$|\Re \fk_0|\le\sfrac{\pi}{L^2}$ and $|\Re\fk_\nu|\le\sfrac{\pi}{L}$
for $1\le\nu\le3$.

\Item (a)
 By Lemmas \ref{lemPBSuplusppties}.a and \ref{lemPOCDenppties}.c, the matrices 
\begin{equation*}
\big[\sfrac{a}{L^2} u_+(\fk+\ell)^\fq 
       u_+(\fk+\ell')^\fq\big]_{\ell,\ell'\in\hat\cB^+}\quad\text{and}\quad
\big[\sfrac{a}{L^2}  u_+\big(\bbbl^{-1}(k\!+\!\ell)\big)^\fq
                      u_+\big(\bbbl^{-1}(k\!+\!\ell')\big)^\fq  
             \big]_{\ell,\ell'\in\hat\cB_1}
\end{equation*}
have finite $L^1$--$L^\infty$ norms and the matrix elements of 
\begin{equation*}
\big[\de_{\ell,\ell'}\hat\De^{(n)}(\fk+\ell)\big]_{\ell,\ell'\in\hat\cB^+}
\quad\text{and}\quad
\big[\de_{\ell,\ell'}
    \hat\De^{(n)}\big(\bbbl^{-1}(k\!+\!\ell)\big)\big]_{\ell,\ell'\in\hat\cB^1}
\end{equation*}
are all bounded.  

\Item (b) \emph{Case $|k|\le c_0$, with $c_0$ being chosen later in 
this case:}\ \ \ 
We first consider those diagonal matrix elements of $d_{n,k}^{(s)}(\ell,\ell')$
having $k,\ell$ such that $|\bbbl^{-1}(k+\ell)|<\tilde c_0$, 
where $\tilde c_0>0$ 
is a small number to be chosen shortly. This, and all other constants chosen
in the course of this argument are to be independent of $L$. Then, by 
Lemma \ref{lemPOCDenppties}.b, we have the following.
\begin{itemize}[leftmargin=*, topsep=2pt, itemsep=0pt, parsep=0pt]
\item
If at least one of $\ell_\nu$, $1\le\nu\le 3$ is nonzero, then 
$\Re \hat\De^{(n)}\big(\bbbl^{-1}(k\!+\!\ell)\big)\ge  \sfrac{c_1}{L^2}$,
provided $\mm_2$ and $\tilde c_0$ are chosen small enough. Here $c_1$ and the constraints
on $\mm_2$ and $\tilde c_0$ depend only on the largest and smallest eigenvalues
of $\big[H_{\nu,\nu'}\big]$, $a_n$ and the $O(|k|^3)$. 
To see this, denote by $\bk$ and $\el$ the spatial parts of $k$ and $\ell$
and observe that
   \begin{itemize}[leftmargin=*, topsep=2pt, itemsep=0pt, parsep=0pt]
     \item    $\big|\Re\bbbl^{-1}(\bk+\el)\big|
              \ge\sfrac{1}{L}\max\big\{\pi,\half|\el|\big\}$,
     \item $\big|\Im\bbbl^{-1}(\bk+\el)\big|=\sfrac{1}{L}\big|\Im\bk\big|
                   \le \sfrac{3\mm_2}{L}$ and 
     \item $\big|\Im\bbbl^{-1}(k+\ell)_0\big|
       =\sfrac{1}{L^2}\big|\Im k_0\big| \le \sfrac{3\mm_2}{L^2}$.
   \end{itemize}
In controlling the
contribution from $O\Big(\big|\bbbl^{-1}(k\!+\!\ell)\big|^3\Big)$ when 
$\sfrac{|\ell_0|}{L^2}$ is larger than $\sfrac{|\el|}{L}$, 
we have to use that, in this case, the real part of 
$\big(\sfrac{1}{a_n}+\sfrac{\veps_n^2}{2}\big)
\big(\bbbl^{-1}(k\!+\!\ell)\big)_0^2$
is at least a strictly positive constant times $\sfrac{\ell_0^2}{L^4}$.

\item 
If $\ell_\nu=0$ for all $1\le\nu\le 3$ but $\ell_0\ne 0$,
then $-\sgn\,\ell_0\ \Im \hat\De^{(n)}\big(\bbbl^{-1}(k\!+\!\ell)\big)
                  \ge  \sfrac{\pi}{2L^2}$,
provided $\mm_2$ and $\tilde c_0$ is chosen small enough. To see this observe that
\begin{itemize}[leftmargin=*, topsep=2pt, itemsep=0pt, parsep=0pt]
     \item   $\big|\Re\bbbl^{-1}(\bk+\el)\big|
               =\sfrac{1}{L}|\Re\bk|
              \le\sqrt{3}\sfrac{\pi}{L}$,
     \item $\big|\Im\bbbl^{-1}(\bk+\el)\big|=\sfrac{1}{L}\big|\Im\bk\big|
                   \le \sfrac{3\mm_2}{L}$,
     \item $\sgn\,\ell_0\ \Re\bbbl^{-1}(k+\ell)_0
           \ge\sfrac{1}{L^2}\max\big\{\pi,\half|\ell_0|\big\}$,
    \item $\big|\Re\bbbl^{-1}(k+\ell)_0\big|
      \le\sfrac{3}{2L^2}|\ell_0|$ and
    \item $\big|\Im\bbbl^{-1}(k+\ell)_0\big|
       =\sfrac{1}{L^2}\big|\Im k_0\big| \le \sfrac{3\mm_2}{L^2}$.
\end{itemize}
\item 
If $\ell\ne 0$, then, by parts (a) and (b) of Lemma \ref{lemPBSuplusppties},
\begin{equation*}
\sfrac{a}{L^2} \big|u_+\big(\bbbl^{-1}(k\!+\!\ell)\big)^{2\fq}\big|
\le \sfrac{a}{L^2}|k|\Big[\smprod\limits_{\nu=0}^3 
                               \sfrac{24}{|\ell_\nu|+\pi}\Big]^{2q}
\end{equation*}

\item
 If $\ell=0$ then
$\Re\big\{\sfrac{a}{L^2} u_+\big(\bbbl^{-1}(k)\big)^{2q}
                        + \hat\De^{(n)}\big(\bbbl^{-1}(k)\big)\big\}
   \ge  \sfrac{a}{2L^2}$,
provided $\mm_2$ and $\tilde c_0$ are chosen small enough. To see this 
observe that
\begin{itemize}[leftmargin=*, topsep=2pt, itemsep=0pt, parsep=0pt]
\item $\big|u_+\big(\bbbl^{-1}(k)\big)-1\big|\le 4^3|k|^2$,
                 by Lemma \ref{lemPBSuplusppties}.c and 
                 $\big|u_+\big(\bbbl^{-1}(k)\big)\big|
                         \le \big(\sfrac{24}{\pi}\big)^4$
                  by Lemma \ref{lemPBSuplusppties}.a.
\item $\Re\hat\De^{(n)}\big(\bbbl^{-1}(k)\big)\ge 
                 -c_2\big(\sfrac{\mm_2}{L^2}+\tilde c_0\big(\sfrac{|k|}{L}\big)^2\big)$
             where $c_2$ depends only on $a_n$, the largest eigenvalue
             of $\big[H_{\nu,\nu'}\big]$,  and the $O(|k|^3)$.
\end{itemize}
\end{itemize}
Note that we have now fixed $\tilde c_0$.
Now we consider the remaining matrix elements.
\begin{itemize}[leftmargin=*, topsep=2pt, itemsep=0pt, parsep=0pt]
\item
 For the remaining diagonal matrix elements we have 
$|\bbbl^{-1}(k+\ell)|\ge \tilde c_0$ and then, by Lemma \ref{lemPOCDenppties}.d, 
$\Re\hat\De^{(n)}\big(\bbbl^{-1}(k\!+\!\ell)\big)\ge \rho(\tilde c_0)$

\item
 Finally, the off--diagonal matrix elements of 
$d_{n,k}^{(s)}$ obey, by parts (a) and (b) of Lemma \ref{lemPBSuplusppties},
\begin{equation*}
\sfrac{a}{L^2} \big|u_+\big(\bbbl^{-1}(k\!+\!\ell)\big)^\fq
               u_+\big(\bbbl^{-1}(k\!+\!\ell')\big)^\fq\big|
\le \sfrac{a}{L^2}|k|\Big[\smprod\limits_{\nu=0}^3 
                               \sfrac{24}{|\ell_\nu|+\pi}\Big]^q
        \Big[\smprod\limits_{\nu=0}^3 
                               \sfrac{24}{|\ell'_\nu|+\pi}\Big]^q
\end{equation*}
Hence the  off--diagonal part of $d_{n,k}^{(s)}$  has 
Hilbert-Schmidt,  matrix and, as $\fq>1$, $L^1$--$L^\infty$ norms 
all bounded by a universal constant times $\sfrac{a}{L^2}|k|$.
\end{itemize}

\noindent Thus $\la\bbbone-d_{n,k}^{(s)}$ has diagonal matrix elements
of magnitude at least 
\begin{equation*}
\min\big\{\sfrac{c_1}{L^2},\sfrac{\pi}{2L^2},\sfrac{a}{2L^2},
                                             \rho(\tilde c_0)\big\}
-\sfrac{a}{L^2}|k|\Big[\smprod\limits_{\nu=0}^3 
                               \sfrac{24}{|\ell_\nu|+\pi}\Big]^{2q}
-\sfrac{\la_0}{L^2}
\end{equation*}
and off diagonal part with  $L^1$--$L^\infty$ norm bounded 
by a universal constant times $\sfrac{a}{L^2}|k|$. It now suffices to 
choose $c_0$ and $\la_0$ small enough that every diagonal matrix element
has magnitude at least 
$\sfrac{1}{2L^2}\min\big\{c_1,\sfrac{\pi}{2},\sfrac{a}{2},
              \rho(\tilde c_0)\big\}$
and the off diagonal part has  $L^1$--$L^\infty$ norm bounded
by $\sfrac{1}{4L^2}\min\big\{c_1,\sfrac{\pi}{2},\sfrac{a}{2},
 \rho(\tilde c_0)\big\}$
and then do a Neumann expansion.

\Item (b) \emph{Case $|k|\ge c_0$, with the $c_0$ just chosen:}\ \ \ 
We may assume that $|\Re k_\nu|\le\pi$ for each $0\le\nu\le3$.
\begin{itemize}[leftmargin=*, topsep=2pt, itemsep=0pt, parsep=0pt]
\item
 If $|\bbbl^{-1}(k+\ell)|<\tilde c_0$ and $\ell_\nu\ne 0$ 
for at least one $1\le\nu\le 3$ or if $|\bbbl^{-1}(k+\ell)|\ge \tilde c_0$, 
then 
$\Re \hat\De^{(n)}\big(\bbbl^{-1}(k+\ell)\big)\ge     
                   \min\big\{\rho(\tilde c_0),\sfrac{c_1}{L^2}\big\}$.
The proof of this given in the case $|k|\le c_0$ applies now too.

\item
 If $|\bbbl^{-1}(k+\ell)|<\tilde c_0$,
$\ell_\nu\!=\!0$ for all $\nu\!\ge\! 1$
then 
\begin{align*}
   \Re \hat\De^{(n)}\big(\bbbl^{-1}(\Re k+\ell)\big)
                  &\ge 0\cr
\big|\hat\De^{(n)}\big(\bbbl^{-1}(k+\ell)\big)
 -\hat\De^{(n)}\big(\bbbl^{-1}(\Re k+\ell)\big)\big|
&\le 4\pi\Gam_7\sfrac{1}{L^2}|\Im k|
\end{align*}
The first bound follows immediately 
from Remark \ref{remPOCde}.c. The second bound follows from 
Lemma \ref{lemPOCDenppties}.c  (for the $\sfrac{1}{L^2}\Im k_0
=\Im(\bbbl^{-1}(k+\ell)_0$ 
contribution) and Lemma \ref{lemPOCDenppties}.g (for the $\sfrac{1}{L}\Im\bk_\nu$
contribution, with $1\le\nu\le3$, --- note that on the line segment from
$\bbbl^{-1}(\Re k+\ell)$ to $\bbbl^{-1}(k+\ell)$, 
$\big|\sfrac{\partial\hat\De^{(n)}}{\partial k_\nu}\big|$
is bounded by $\Gam_7|\bbbl^{-1}(k+\ell)_\nu|
=\sfrac{\Gam_7}{L}|\bk_\nu|\le\Gam_7\sfrac{\pi+3\mm_2}{L} $). Furthermore
\begin{alignat*}{3}
  -\sgn\,\ell_0\ \Im \hat\De^{(n)}\big(\bbbl^{-1}(\Re k+\ell)\big)
                  &\ge \sfrac{\pi}{2L^2} &\qquad
                  &\text{if $\ell_0\ne 0$}\\ 
\big|\hat\De^{(n)}\big(\bbbl^{-1}(\Re k)\big)\big|
    &\ge\gam_1\gam_2 
\min\big\{\sfrac{c_0}{\sqrt{2}}\ ,\ \sfrac{c_0^2}{2}\big\}\ \sfrac{1}{L^2} &
                  &\text{if $\ell_0= 0$}\\
\end{alignat*}
In the case $\ell_0\ne 0$, the proof of the bound given in the 
case $|k|\le c_0$ applies now too. (Just apply it to $\Re k$.) 
The bound for the case $\ell_0=0$ follows from Lemma \ref{lemPOCDenppties}.e and
Lemma \ref{lemPDOhatSzeroppties}.a.

\item For all $\ell$, by parts (a) and (e) of Lemma \ref{lemPBSuplusppties},
\begin{align*}
 \big|u_+\big(\bbbl^{-1}(k\!+\!\ell)\big)^{\fq}\big|
&\le \Big[\smprod\limits_{\nu=0}^3 \sfrac{24}{|\ell_\nu|+\pi}\Big]^{\fq}
\\
 \big|\Im u_+\big(\bbbl^{-1}(k\!+\!\ell)\big)^{\fq}\big|
&\le 16\fq\ |\Im k|\ \Big[\smprod\limits_{\nu=0}^3 \sfrac{24}{|\ell_\nu|+\pi}\Big]^{\fq}
\end{align*}
\end{itemize}
We split $\la\bbbone - d_{n,k}^{(s)}$ into three pieces
\begin{equation*}
\big(\la\bbbone-d_{n,k}^{(s)}\big)(\ell,\ell')
=D(\ell,\ell') -P(\ell,\ell') + I(\ell,\ell')
\end{equation*}
where
\begin{alignat*}{3}
D(\ell,\ell')&= \de_{\ell,\ell'} d_\ell & &\text{with }
d_\ell= \la_--
\begin{cases}
  \hat\De^{(n)}\big(\bbbl^{-1}(k\!+\!\ell)\big)
  &\text{$\ell_\nu\ne 0$ for some $\nu\ge 1$}\\
  \hat\De^{(n)}\big(\bbbl^{-1}(k\!+\!\ell)\big)
  &\text{$|\bbbl^{-1}(k+\ell)|\ge \tilde c_0$} \\
\hat\De^{(n)}\big(\bbbl^{-1}(\Re k\!+\!\ell)\big)
  &\text{otherwise}
\end{cases}
\\
& & &\phantom{with }\la_-=\min\{\Re\la,0\}+i\Im\la
\\
P(\ell,\ell')&= \sfrac{a}{L^2}v(\ell)v(\ell')\  & &\text{with }
v(\ell)= \Re u_+\big(\bbbl^{-1}(k\!+\!\ell)\big)^\fq\cr
\end{alignat*}
We have chosen $\la_-$ so that $\Re\la_-\le 0$ and 
$|\la-\la_-|\le\sfrac{\la_0}{L^2}$.
As $P$ is a rank one operator
\begin{equation*}
(D-P)^{-1}(\ell,\ell')
=\sfrac{1}{d_\ell} \de_{\ell,\ell'}  +\sfrac{1}{1-\ka}\sfrac{a}{L^2}\sfrac{v(\ell)}{d_\ell}\sfrac{v(\ell')}{d_{\ell'}}
\end{equation*}
with
\begin{equation*}
\ka=\sum_{\ell''}\sfrac{a}{L^2}\sfrac{1}{d_{\ell''}}v(\ell'')^2
\end{equation*}
We have shown above that 
\begin{align*}
|d_\ell|&\ge \sfrac{\la_1}{L^2} \\
\Re d_\ell&\le \begin{cases} 
      0 & \text{if $|\bbbl^{-1}(k+\ell)|<\tilde c_0$,
$\ell_\nu\!=\!0$ for all $\nu\!\ge\! 1$}\\      
      -\sfrac{\la_1}{L^2} & \text{otherwise}
              \end{cases} \\
| v(\ell)|&\le 
   \Big[\smprod\limits_{\nu=0}^3 \sfrac{24}{|\ell_\nu|+\pi}\Big]^{q}\\
|I(\ell,\ell')|&\le \big(\sfrac{\la_0}{L^2}
     + 4\pi\Gam_7\sfrac{1}{L^2}|\Im k|\big)\ \de_{\ell,\ell'}
+ 32\fq \ |\Im k|\ \sfrac{a}{L^2}
    \Big[\smprod\limits_{\nu=0}^3 \sfrac{24}{|\ell_\nu|+\pi}\Big]^{\fq}
    \Big[\smprod\limits_{\nu=0}^3 \sfrac{24}{|\ell'_\nu|+\pi}\Big]^{\fq}
\end{align*}
provided we choose 
$0<\la_1\le 
   \min\big\{c_1,\rho(\tilde c_0),\sfrac{\pi}{2},
                \gam_1\gam_2\sfrac{c_0}{\sqrt{2}},
                \gam_1\gam_2\sfrac{c_0^2}{2}\big\}
$.

Since $\Re z\le 0\Rightarrow \Re\sfrac{1}{z}\le 0$, and $v(\ell)\in\bbbr$
for all $\ell$, we have that $\Re\ka\le 0$ so that 
$\big|\sfrac{1}{1-\ka}\big|\le 1$ and
\begin{equation*}
\|(D-P)^{-1}\|_{\ell^1-\ell^\infty}\le\sfrac{L^2}{\la_2}\qquad
\|I\|_{\ell^1-\ell^\infty}\le\sfrac{\Gam'_5}{L^2}\big(\la_0+|\Im k|\big)
\end{equation*}
with the constant $\la_2>0$ depending only on $a$ and $\la_1$
and the constant $\Gam'_5$ depending only on $\Gam_7$, $a$ and $\fq$.
It now suffices to choose $\la_0$ and $\mm_2$ smaller than 
$\sfrac{\la_2}{12\Gam'_5}$ and use a Neumann expansion to give
\begin{equation*}
\big\|\big(\la\bbbone-d_{n,k}^{(s)}\big)^{-1}\big\|_{\ell^1-\ell^\infty}
   \le 2 \sfrac{L^2}{\la_2}
\end{equation*}
\end{proof} 

\pagebreak[2]
\begin{proposition}\label{propPOCsquareroot}
Let $\mm_2,\la_0,\Gam_8$ be as in Lemma \ref{lemPOCakmatrix} and use $\bbbr_-$
to denote the negative real axis in $\bbbc$. Set
\bigskip
\begin{align*}
\cO_C &=\set{z\in\bbbc}{\dist(z,\bbbr_-)>\sfrac{\la_0}{2L^2},\ |z|<\Gam_8+1}
\hskip0.3in\smash{\lower1.2in\hbox{\includegraphics{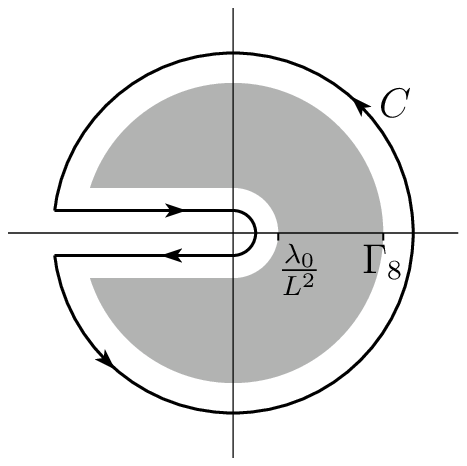}}}\\
\cO&=\set{z\in\bbbc}{\dist(z,\bbbr_-)>\sfrac{\la_0}{3L^2},\ |z|<\Gam_8+2}
\end{align*}
and let
\begin{itemize}[leftmargin=*, topsep=2pt, itemsep=0pt, parsep=0pt]
\item $C=\partial\cO_C$, oriented counterclockwise
\item $f:\cO\rightarrow\bbbc$ be analytic.
\end{itemize}
\medskip
\noindent Then $f\big((\sfrac{a}{L^2}Q^*Q+\De^{(n)})^{(s)}\big)$, defined 
by \cite[(\eqnBOfofA) and Lemma \lemPoPscaling.a]{Bloch}, exists and
there is a constant\footnote{Recall Convention \ref{convPOconstants}.}
$\Gam_9$ such that
\begin{equation*}
\|f\big((\sfrac{a}{L^2}Q^*Q+\De^{(n)})^{(s)}\big)\|_{\mm_2}
   \le \Gam_9 L^7 \sup_{\ze\in C}|f(\ze)|
\end{equation*}

\end{proposition}
\begin{proof} 
Apply \cite[Lemma \lemBOfnbnd]{Bloch} with $\hat a_k(\ell,\ell')
= d_{n,k}^{(s)}(\ell,\ell')$ and
\begin{equation*}
\cX_\fin=\cX_1^{(n)}\quad
\cZ_\fin=\veps_1^2\bbbz\times \veps_1\bbbz^3 \quad
\cX_\crs=\cX_0^{(n+1)}\quad
\cZ_\crs=\bbbz\times \bbbz^3 \quad
\cB=\cB_1 
\end{equation*}
and $m=3\mm_2$, $m'=2\mm_2$ and $m''=\mm_2$.
Then $\vol_c=1$, $\#\hat\cB_1=L^5$.
Observe in particular that, by Lemma \ref{lemPOCakmatrix},
the spectrum of $d_{n,k}^{(s)}$ is contained in
\begin{equation*}
\set{z\in\bbbc}{\dist(z,\bbbr_-)>\sfrac{\la_0}{L^2},\ |z|\le\Gam_8}
\end{equation*}
which is the shaded region in the figure above.
So the lemma gives
\begin{align*}
&\big\|f\big((\sfrac{a}{L^2}Q^*Q+\De^{(n)})^{(s)}\big)\big\|_{m''}\\
   &\hskip0.25in\le\sfrac{C_{m'-m''}}{2\pi\ \vol_c}
       |C|\ \sup_{\ze\in C}|f(\ze)|
      \sup_{|\Im\rk|=m'\atop \ze\in C}\sum_{\ell,\ell'\in\hat\cB_1}
      \big|(\ze\bbbone-\hat d_{n,\rk}^{(s)})^{-1}(\ell,\ell')|\\
   &\hskip0.25in \le\sfrac{C_{m'-m''}}{2\pi} (\Gam_8+1)(3\pi+2)
       \ \sup_{\ze\in C}|f(\ze)|\ L^5
       \sup_{|\Im\rk|=m'\atop \ze\in C}
       \sup_{\ell\in\hat\cB_1}
      \sum_{\ell'\in\hat\cB_1}
      \big|(\ze\bbbone-\hat d_{n,\rk}^{(s)})^{-1}(\ell,\ell')|\\
   &\hskip0.25in \le\sfrac{C_{m'-m''}}{2\pi} (\Gam_8+1)(3\pi+2)
       \ \sup_{\ze\in C}|f(\ze)|L^5\ 
      \Gam_8 L^2
\end{align*}
by Lemma \ref{lemPOCakmatrix}.b.

\end{proof}

Applying Proposition \ref{propPOCsquareroot} with $f(z)=\sfrac{1}{z}$,
$f(z)=\sfrac{1}{\sqrt{z}}$ and $f(z)=\sqrt{z}$,
where $\sqrt{z}$ is the principal value of the square root gives
\begin{corollary}\label{corPOCsquareroot}
The operators $C^{(n)}$, $\sqrt{C^{(n)}}$ and 
$\big(\sqrt{C^{(n)}}\,\big)^{-1}$ all exist. There is a 
constant $\Gam_{10}$ such that
\begin{equation*}
\big\|\bbbl_*^{-1} C^{(n)}\bbbl_*\big\|_{\mm_2}\ ,\ 
\big\|{\textstyle\sqrt{\bbbl_*^{-1} C^{(n)}\bbbl_*}}\,\big\|_{\mm_2}\ ,\ 
\big\|\big({\textstyle\sqrt{\bbbl_*^{-1} C^{(n)}\bbbl_*}}\,\big)^{-1}\,\big\|_{\mm_2}
\le \Gam_{10} L^9
\end{equation*}
\end{corollary}

\newpage
\section{The Green's Functions}\label{secPOgreens}
In this chapter, we discuss the inverses of the operators
\begin{equation*}
D_n+Q^*_n\fQ_nQ_n
\end{equation*}
These inverses, and variations thereof, are constituents of the leading part of the power series expansion of the background fields of \cite{PAR1,PAR2,
BGE}.
See \cite[Proposition \propHTexistencebackgroundfields]{PAR1} 
and \cite[Proposition \propBGEphivepssoln]{BGE}.
In Proposition \ref{POGmainpos}, below, we show that for sufficiently 
small $\mu$, the operators
$\,
D_n+Q^*_n\fQ_nQ_n-\mu
\,$
are invertible, and we estimate the decay of the kernels $S_n(\mu)(x,y)$ of their inverses
\begin{equation*}
S_n(\mu)=\big[D_n+Q^*_n\fQ_nQ_n-\mu\big]^{-1}
\end{equation*}
By Remark \ref{remPBSunderivAlg}
\begin{equation}\label{eqnPOGpartialS}
\partial_\nu \big(S_n(\mu)^*\big)^{-1}=\big(S_{n,\nu}^{(+)}(\mu)\big)^{-1}\partial_\nu
\qquad
\partial_\nu S_n(\mu)^{-1}=\big(S_{n,\nu}^{(-)}(\mu)\big)^{-1}\partial_\nu
\end{equation}
where
\begin{equation}\label{eqnPOGSnnudef}
S_{n,\nu}^{(+)}(\mu)=\big[D_n^*
        +Q_{n,\nu}^{(+)}\fQ_nQ_{n,\nu}^{(-)}-\mu\big]^{-1}
        \qquad
S_{n,\nu}^{(-)}(\mu)=\big[D_n
        +Q_{n,\nu}^{(+)}\fQ_nQ_{n,\nu}^{(-)}-\mu\big]^{-1}
\end{equation}
and $Q_{n,\nu}^{(+)},Q_{n,\nu}^{(-)}$ were defined in \eqref{eqnPBSqnplusminus}.
We shall write
\begin{equation}\label{eqnPOGSnnuzerodef}
S_n=S_n(0)    \qquad\qquad S_{n,\nu}^{(\pm)}=S_{n,\nu}^{(\pm)}(0)
\end{equation} 
The main result extends the statement of 
\cite[Theorem \HTthminvertibleoperators]{PAR1}. It is

\begin{proposition}\label{POGmainpos}
There are constants  $\mu_{\rm up}, \mm_3>0$ and $\Gam_{11}$, depending only on 
$\fq$, $\bh_0$ and $a$, and in particular independent of $n$ and 
$L>\Gam_2$, such that, for $|\mu| \le \mu_{\rm up}$,
the operators $D_n+Q^*_n\fQ_nQ_n-\mu$ and
$\,D_n^*+Q_{n,\nu}^{(+)}\fQ_nQ_{n,\nu}^{(-)}-\mu\,$,
$\,D_n+Q_{n,\nu}^{(+)}\fQ_nQ_{n,\nu}^{(-)}-\mu\,$ are invertible, and their inverses
$\,S_n(\mu)\,$ and $\,S_{n,\nu}^{(+)}(\mu)\,$, $\,S_{n,\nu}^{(-)}(\mu)\,$, 
respectively, fulfill
\begin{align*}
\|\ S_n(\mu) \|_{\mm_3}, \|\ S_{n,\nu}^{(\pm)}(\mu) \|_{\mm_3}& \le \Gam_{11}
\\
\|\ S_n(\mu)-S_n \|_{\mm_3}, 
   \|\ S_{n,\nu}^{(\pm)}(\mu)-S_{n,\nu}^{(\pm)} \|_{\mm_3} &\le |\mu|\,\Gam_{11}
\end{align*}
\end{proposition}
\noindent This Proposition is proven following the proof 
of Lemma \ref{lemPOGSnnuppties}.

\begin{example}\label{exPOGmodelsn}
As a model computation, we evaluate the inverse transform
\begin{align*}
s(x)=\int_{\bbbr\times\bbbr^3} 
       \sfrac{e^{ip\cdot x}}{-ip_0 + m^2 + \bp_1^2+\bp_2^2+\bp_3^2}
       \sfrac{dp_0\,d^3\bp}{(2\pi)^4}
\end{align*}
of $\hat s(p)=\sfrac{1}{-ip_0 + m^2 + \bp^2}$. It is designed
to mimic the behaviour of $S_n$ in the limit $n\rightarrow\infty$.
Write $x=(t,\bx)\in\bbbr\times\bbbr^3$. We first compute the $p_0$ integral.
Observe that the integrand has exactly one pole, which is at 
$p_0=-i(m^2 + \bp^2)$, and that the $e^{ip_0 t}$ in the integrand forces
us to close the contour in the upper half plane when $t>0$ and in the 
lower half plane when $t<0$. Thus
\begin{align*}
\int_{-\infty}^\infty 
       \sfrac{e^{ip_0t+i\bp\cdot \bx}}{-ip_0 + m^2 + \bp^2}
       \sfrac{dp_0}{(2\pi)^4}
=\begin{cases} 0 & \text{if $t>0$}\\
        e^{(m+\bp^2)t+i\bp\cdot\bx}\sfrac{1}{(2\pi)^3} & \text{if $t<0$}
       \end{cases}
\end{align*}
Hence $s(x)=0$ for $t>0$ and, for $t<0$,
\begin{align*}
s(x) & = \int_{\bbbr^{3}} 
       e^{-(m^2+\bp^2)|t|}e^{i\bp\cdot \bx}
       \sfrac{d^3\bp}{(2\pi)^{3}}
=e^{-m^2|t|}\prod_{j=1}^3 
   \int_{-\infty}^\infty e^{-|t|\bp_j^2}e^{i\bp_j\bx_j}\sfrac{d\bp_j}{2\pi}\\
&=\sfrac{1}{2^3(\pi|t|)^{3/2}}e^{-m^2|t|}e^{-\sfrac{\bx^2}{4|t|}}
\end{align*}
Here are some observations about $s(x)$.
\begin{itemize}[leftmargin=*, topsep=2pt, itemsep=0pt, parsep=0pt]
\item 
Since $m^2|t|+\sfrac{\bx^2}{4|t|}\ge m|\bx|$ (the minimum
is at $|t|=\sfrac{|\bx|}{2m}$), $s(x)$ decays exponentially for large $|x|$
in all directions.
\item 
For $\bx\ne 0$, $\lim_{t\nearrow 0}e^{-\sfrac{\bx^2}{4|t|}}=0$,
so $s(x)$ is continuous everywhere except at $x=0$. 

\item 
$s(x)$ has an integrable singularity at $x=0$.
There are a number of ways to see this. For example,
the inequality $e^{-\sfrac{\bx^2}{8|t|}}
\le\const \big(\sfrac{|t|}{|\bx|^2}\big)^{\ka/2}$
implies that $|s(x)|$ is bounded near $x=0$ by a constant times 
$\sfrac{1}{|t|^{(3-\ka)/2}}\sfrac{1}{|\bx|^\ka}$. This is integrable if
$1<\ka<3$. 

\item 
If we send $x\rightarrow 0$ along a curve with 
$\bx^2=-4\ga|t|\ln|t|$, $s(x)\approx\const\sfrac{1}{|t|^{3/2-\ga}}$.

\item 
We see, using $\bx_j^4e^{-\sfrac{\bx_j^2}{8|t|}}\le\const |t|^2$, that $\big(t^2+\sum_{j=1}^3\bx_j^4\big)|s(x)|$
is bounded and exponentially decaying. Note that 
$\sfrac{1}{t^2+\Si_j\bx_j^4}$ has an integrable singularity at the origin,
since $\int_{-\infty}^\infty \sfrac{1}{t^2+\Si_j\bx_j^4} dt
=\sfrac{\const}{\sqrt{\Si_j \bx_j^4}}$.
\end{itemize}
\end{example}

As preparation for and in addition to the position space estimates 
of Proposition \ref{POGmainpos}, we also derive
bounds on the Fourier transforms of these and related operators. 
To convert bounds in momentum space into bounds in position space,
we shall use \cite[Lemma \lemBOlonelinfty]{Bloch}, with
\begin{equation}\label{eqnPOGsubstspaces}
\cX_\fin=\cX_n\qquad
\cZ_\fin=\veps_n^2\bbbz\times\veps_n\bbbz^3 \qquad
\cX_\crs=\cX_{0}^{(n)}\qquad
\cZ_\crs=\bbbz\times \bbbz^3 \qquad
\cB=\cB_n
\end{equation}
We shall routinely use $|p_0|$, $|\bp_\nu|$ and $|\bp|$ to refer to the 
magnitudes of the smallest representatives of $p_0\in\bbbc$, 
$\bp_\nu\in\bbbc$ and $\bp\in\bbbc^3$ in $\bbbc/\sfrac{2\pi}{\veps_n^2}\bbbz$, 
$\bbbc/\sfrac{2\pi}{\veps_n}\bbbz$ and $\bbbc^3/\sfrac{2\pi}{\veps_n}\bbbz^3$, respectively.

The operators $S_n(\mu)$ act on functions on the lattice $\cX_n$, but they are only translation invariant with respect to the sublattice $\cX_0^{(n)}$.
An exponentially decaying operator which is fully translation invariant, and 
has the same local singularity as $S_n$, is the operator
\begin{equation}\label{eqnPOGSntransapprox}
\begin{split}
S_n'&=\big[D_n+a_n\exp\{-\De_n\}\big]^{-1}\qquad\text{where} \\
\De_n&=\partial_0^*\partial_0 
              +\big(\partial_1^*\partial_1+\partial_2^*\partial_2
                     +\partial_3^*\partial_3\big)
\end{split}
\end{equation}
and $a_n=a\sfrac{1-L^{-2}}{1-L^{-2n}}$ as in Proposition \ref{propPBSAnppties}.b, and
the forward derivatives  $\partial_\nu$ are defined in \eqref{eqnPBSforwardDeriv}.
Obviously $S_n'$ has Fourier transform 
\begin{equation*}
\widehat S'_n(p)=\big[\hat\bD_n(p)
                    +a_n\exp\{-\De_n(p)\}\big]^{-1}
\quad\text{where }
\De_n(p)= 
\big[\sfrac{\sin\frac{1}{2}\veps_n^2 p_0}{\frac{1}{2}\veps_n^2}\big]^2 \!
+\smsum_{\nu=1}^3\!\big[\sfrac{\sin\frac{1}{2}\veps_n \bp_\nu}{\frac{1}{2}\veps_n}\big]^2 
\end{equation*}
Before we discuss the properties of $S_n'$ and of the difference 
$\,\de S = S_n-S_n'\,$ we note

\begin{remark}\label{remPOGpropLapl}
The Fourier transform $\De_n(p)$ of the four dimensional Laplacian $\De_n$ 
is entire. For $p\in\bbbr\times \bbbr^3$,
\begin{equation*}
\De_n(p)\ge\sfrac{2}{\pi^2}\big[|p_0|^2+|\bp|^2\big]
\end{equation*}
For $p\in\bbbc\times \bbbc^3$ with $\veps_n^2|\Im p_0|\le 1$ and 
$\veps_n|\Im\bp|\le 1$
\begin{equation*}
|\De_n(p)|\le 4\big[|p_0|^2+|\bp|^2\big]\qquad
|\sfrac{\partial\hfill}{\partial p_\nu}\De_n(p)|\le 4|p_\nu|\qquad
|\sfrac{\partial^\ell\hfill}{\partial p_\nu^\ell}\De_n(p)|
            \le 4\veps_{n,\nu}^{\ell-2}
            \ \text{if $\ell\ge 2$}
\end{equation*}
with the $\veps_{n,\nu}$ of \eqref{eqnPINTlnu}.
For $p\in\bbbc\times \bbbc^3$ with $|\Im p|\le 1$,
\begin{equation*}
\Re \De_n(p) \ge -5\pi^2 + \sfrac{1}{\pi^2}\big[|p_0|^2+|\bp|^2\big]
\end{equation*}
\end{remark}

\begin{proof}
For the first two claims, just apply parts (a) and (b) of 
Lemma \ref{lemPBSsinxoverx}. For the derivatives, use 
$$
\sfrac{d\hfill}{d\th}\big[\sfrac{\sin(\eta\th)}{\eta}\big]^2
=\sfrac{\sin(2\eta\th)}{\eta}\quad\implies\quad
\sfrac{d^\ell\hfill}{d\th^\ell}\big[\sfrac{\sin(\eta\th)}{\eta}\big]^2
=\pm(2\eta)^{\ell-1}\sfrac{1}{\eta}
          \begin{cases}\sin(2\eta\th) & \text{for $\ell$ odd}\\
                  \cos(2\eta\th) & \text{for $\ell$ even}
              \end{cases}
$$
For the final claim, write $p=\sP+i\sQ$ with $\sP,\sQ\in\bbbr\times\bbbr^3$.
Then 
\begin{align*}
\Re \De_n(\sP+i\sQ)
&\ge  \De_n(\sP)
    -\big|\Re\De_n(\sP+i\sQ)- \De_n(\sP)\big|\\
&\ge \sfrac{2}{\pi^2}\big[|\sP_0|^2+|\sbP|^2\big]
   -4|\sQ||\sP+i\sQ|\\
&\ge \sfrac{2}{\pi^2}|\sP+i\sQ|^2
   -2\big(2+\sfrac{2}{\pi^2}\big)\,|\sQ|\,|\sP+i\sQ|
   -\sfrac{2}{\pi^2}|\sQ|^2\\
&\ge \sfrac{1}{\pi^2}|\sP+i\sQ|^2
   -\big\{\pi^2\big(2+\sfrac{2}{\pi^2}\big)^2+\sfrac{2}{\pi^2}\big\}|\sQ|^2
\end{align*}
\end{proof}

By the resolvent identity
\begin{equation*}
\de S = S_n-S_n'=- S'_n\big[Q_n^*\fQ_nQ_n-a_n\exp\{-\De_n\}\big]S_n
\end{equation*}
$S_n$ and $\de S$ are translation invariant with respect to the sublattice  
$\cX_0^{(n)}$ of $\cX_n$. By ``Floquet theory'' 
(see \cite[Lemma \lemBOkervar]{Bloch}),  their Fourier transforms
$\,\widehat S_n(p,p')\,,\,\widehat{\de S}(p,p')\,$, 
$p,p'\in \hat\cX_n$ vanish unless 
$\,\pi_n^{(n,0)}(p)=\pi_n^{(n,0)}(p')$, 
i.e. unless there are $k\in \hat\cX_0^{(n)}$ and $\ell,\ell'\in \hat\cB_n$ such that $p=k+\ell\,$, $\,p'=k+\ell'$.
The blocks 
$\,\widehat S_{n,k}^{-1}(\ell,\ell') = \widehat S_n^{-1}(k+\ell,k+\ell')\,$
and
$\,\widehat{\de S}_k(\ell,\ell') = \widehat{\de S}(k+\ell,k+\ell')\,$
are given by
\begin{equation}\label{eqnPOGdeSnft}
\begin{split}
\widehat S_{n,k}^{-1}(\ell,\ell')
&=\hat\bD_n(k+\ell)\,\de_{\ell,\ell'}
  +u_n(k+\ell)^\fq\hat \fQ_n(k)u_n(k+\ell')^\fq \\
\widehat{\de S}_k(\ell,\ell') 
&=-\hskip-4pt\sum_{\ell''\in\hat\cB_n}\hskip-4pt
     \widehat S'_n(k\!+\!\ell)
    \big[ u_n(k\!+\!\ell)^\fq\hat \fQ_n(k) u_n(k\!+\!\ell'')^\fq
             -a_ne^{-\De_n(k\!+\!\ell)}\de_{\ell,\ell''}\big]
    \widehat{S}_{n,k}(\ell'',\ell')
\end{split}
\end{equation}
where $u_n$ and $\hat\fQ_n$ are given in parts (b) and (e) of 
Remark \ref{remPBSqnft}.

\begin{lemma}\label{lemPOGSnppties}
There are constants\footnote{Recall Convention \ref{convPOconstants}.}
$\mm_4>0$ and $\Gam_{12}$, 
such that the following hold for all $L>\Gam_2$.
\begin{enumerate}[label=(\alph*), leftmargin=*]
\item 
$\widehat{S}'_n(p)$ is analytic in $|\Im p|<3\mm_4$ and obeys
\begin{equation*}
\big|\widehat{S}'_n(p)\big|\le \sfrac{\Gam_{12}}{1+|p_0|+|\bp|^2}
\quad\text{and}\quad
\big|\sfrac{\partial^2\hfill}{\partial p_0^2}\widehat{S}'_n(p)\big|,
\big|\sfrac{\partial^4\hfill}{\partial \bp_\nu^4}\widehat{S}'_n(p)\big|
\le \sfrac{\Gam_{12}}{{(1+|p_0|+|\bp|^2)}^3}
\end{equation*}
for $1\le\nu\le3$ there. 

\item 
For all $\ell,\ell'\in\hat\cB_n$,
$\widehat S_{n,k}(\ell,\ell')$ is analytic in $|\Im k|<3\mm_4$ and obeys
\begin{equation*}
\big|\widehat S_{n,k}(\ell,\ell')\big|
\le\sfrac{\Gam_{12}}{1+|\ell_0|+\Si_{\nu=1}^3|\ell_\nu|^2}
\Big\{\de_{\ell,\ell'}
+\sfrac{1}{1+|\ell'_0|+\Si_{\nu=1}^3|\ell'_\nu|^2}
    \smprod_{\nu=0}^3\! \sfrac{1}{(|\ell_\nu|+1)^\fq}
    \smprod_{\nu=0}^3\! \sfrac{1}{(|\ell'_\nu|+1)^\fq}
\Big\}
\end{equation*}
there.

\item 
For all $\ell,\ell'\in\hat\cB_n$,
$\widehat{\de S}_k(\ell,\ell')$ is analytic in $|\Im k|<3\mm_4$ and obeys
\begin{align*}
\big|\widehat{\de S}_k(\ell,\ell')\big|
&\le\Gam_{12}\exp\big\{-\sfrac{1}{40}\Si_{\nu=0}^3|\ell_\nu|^2\big\}\ \de_{\ell,\ell'}\\
&\hskip0.8in+\sfrac{\Gam_{12}}{1+|\ell_0|+\Si_{\nu=1}^3|\ell_\nu|^2}\ 
   \Big\{\smprod_{\nu=0}^3 \sfrac{1}{(|\ell_\nu|+1)^\fq}
    \smprod_{\nu=0}^3 \sfrac{1}{(|\ell'_\nu|+1)^\fq}\Big\} \ \ 
     \sfrac{1}{1+|\ell'_0|+\Si_{\nu=1}^3|\ell'_\nu|^2}
\end{align*}
there.

\item 
For all $u,u'\in\cX_n$,
\begin{align*}
\big|S_n(u,u')-S'_n(u,u')\big|&\le \Gam_{12} e^{-2\mm_4 |u-u'|}\\
\big|S'_n(u,u')\big|&\le \Gam_{12} \min\Big\{
             \sfrac{e^{-2\mm_4 |u-u'|}}{|u_0-u'_0|^2+|\bu-\bu'|^4}\,,\,
             L^{5n}\Big\}
\end{align*}
\end{enumerate}
\end{lemma}
\begin{proof}  (a)
Obviously $\widehat{S}'_n(p)^{-1}$ is entire. 
For real $p$
\begin{equation}\label{eqnPOGsprimelowerbound}
\big|\widehat{S}'_n(p)^{-1}\big|
\ge a_n\exp\{-\De_n(p)\}+\const\big\{|p_0|+|\bp|^2\big\}
\ge \const\big\{1+|p_0|+|\bp|^2\big\}
\end{equation} 
by Remark \ref{remPOGpropLapl}, Lemma \ref{lemPDOhatSzeroppties}.a, and the 
fact that $\De_n(p),\Re \hat\bD_n(p)\ge 0$ for real $p$.
The bound on $\big|\sfrac{\partial\hfill}{\partial p_\nu} \hat\bD_n(p)\big|$ 
of Lemma \ref{lemPDOhatSzeroppties}.c and Remark \ref{remPOGpropLapl} shows
that \eqref{eqnPOGsprimelowerbound} is valid for all $|\Im p|<3\mm_4$, if $\mm_4$
is chosen sufficiently small.

We now bound the derivatives. For any $0\le\nu\le3$ and $\ell\in\bbbn$,
$\sfrac{\partial^\ell\hfill}{\partial p_\nu^\ell}\widehat{S'_n}(p)$
is a finite linear combination of terms of the form
\begin{equation*}
\widehat{S'_n}(p)^{1+j}\prod_{i=1}^j
      \sfrac{\partial^{\ell_i}\hfill}{\partial p_\nu^{\ell_i}}
      \big[\hat\bD_n(p)
                    +a_n\exp\{-\De_n(p)\}\big]
\end{equation*}
with each $\ell_i\ge 1$ and $\sum_{i=1}^j\ell_i=\ell$. 
By Remark \ref{remPOGpropLapl}, all derivatives of 
$\exp\{-\De_n(p)\}$ of order $\ell_i\in\bbbn$ are bounded by 
$\sfrac{\cst{}{\ell_i}}{{[1+|p_0|+|\bp|^2]}^{\ell_i}}$. Hence, by
Lemma \ref{lemPDOhatSzeroppties}.c,
\begin{equation*}
\big|\sfrac{\partial^{\ell_i}\hfill}{\partial p_\nu^{\ell_i}}
      \big[\hat\bD_n\!(p)
                    +a_n\exp\{-\De_n(p)\}\big]\big|
\le\const\!\begin{cases}\!\sfrac{1}{[1+|p_0|+|\bp|^2]^{\ell_i-1}} & 
                                   \text{if $\nu=0$, $\ell_i=1,2$}\\
                  \!\sfrac{1}{[1+|p_0|+|\bp|^2]^{\ell_i/2-1}} & 
                           \text{if $\nu\ge 1$, $1\le\ell_i\le4$}
              \end{cases}
\end{equation*}
As $\big|\widehat{S}'_n(p)\big|\le \sfrac{\const}{1+|p_0|+|\bp|^2}$,
\begin{align*}
&\Big|\widehat{S'_n}(p)^{1+j}\prod_{i=1}^j
      \sfrac{\partial^{\ell_i}\hfill}{\partial p_\nu^{\ell_i}}
      \big[\hat\bD_n(p)
                    +a_n\exp\{-\De_n(p)\}\big]\Big|\\
&\hskip1in\le \sfrac{\const}{1+|p_0|+|\bp|^2}\prod_{i=1}^j
  \begin{cases}\!\sfrac{1}{[1+|p_0|+|\bp|^2]^{\ell_i}} & 
                            \text{if $\nu=0$, $\ell_i=1,2$}\\
                  \!\sfrac{1}{[1+|p_0|+|\bp|^2]^{\ell_i/2}} & 
                             \text{if $\nu\ge 1$, $1\le\ell_i\le4$}
              \end{cases}
\end{align*}
and the claim follows.

\Item (b) 
For any $c''_0>0$ we have, for $|k|\ge c''_0$ and $|\Im k_0|<3\mm_4$, 
analyticity and the bound
\begin{equation*}
\big|\widehat S_{n,k}(\ell,\ell')\big|
\le\sfrac{\Gam'_{10}}{1+|\ell_0|+\Si_{\nu=1}^3|\ell_\nu|^2}\de_{\ell,\ell'}
+\sfrac{\Gam'_{10}}{1+|\ell_0|+\Si_{\nu=1}^3|\ell_\nu|^2}\!
    \smprod_{\nu=0}^3\! \sfrac{1}{(|\ell_\nu|+1)^\fq}\!
    \smprod_{\nu=0}^3\! \sfrac{1}{(|\ell'_\nu|+1)^\fq} \ \
     \sfrac{1}{1+|\ell'_0|+\Si_{\nu=1}^3|\ell'_\nu|^2}
\end{equation*}
(with $\Gam'_{10}$ depending on $c''_0$) which follows from the representation
\begin{equation*}
S_n = D_n^{-1} -\,D_n^{-1}\,Q_n^*\De^{(n)} Q_nD_n^{-1} 
\end{equation*}
(see \cite[Remark \remBSedA.b]{BlockSpin})
and Lemmas \ref{lemPDOhatSzeroppties}.d, \ref{lemPBSunppties}.a and \ref{lemPOCDenppties}.c.

For $|k|< c'_0$, with $c'_0$ to be shortly chosen sufficiently small,
and $|\Im k|<3\mm_4$ we use the representation
\begin{align}\label{eqnPERSnkell}
&\hat S_{n,k}^{-1}(\ell,\ell')
=\hat\bD_n(k+\ell)\,\de_{\ell,\ell'}
  +u_n(k+\ell)^\fq\hat \fQ_n(k)u_n(k+\ell')^\fq
=D_{\ell,\ell'}+B_{\ell,\ell'}
\end{align}
with
\begin{align*}
D_{\ell,\ell'}
&=\hat\bD_n(k+\ell)\, \de_{\ell,\ell'}
  +\begin{cases}a_n & \text{if $\ell,\ell'=0$}\\
                   \noalign{\vskip0.02in}
                   0 & \text{otherwise}
              \end{cases}\\
\noalign{\vskip0.1in}
B_{\ell,\ell'}
&=\begin{cases}\hat \fQ_n(k)u_n(k)^{2\fq}-a_n  & \text{if $\ell=\ell'=0$}\\
               \noalign{\vskip0.02in}
           u_n(k+\ell)^\fq \hat \fQ_n(k) u_n(k+\ell')^\fq  & \text{otherwise}
   \end{cases}
\end{align*}
By parts (c) and (d) of Lemma \ref{lemPDOhatSzeroppties}, assuming that $|k|<c'_0$
with $c'_0$ small enough, $D$ is invertible and the inverse is a diagonal 
matrix with every diagonal matrix element obeying
\begin{equation*}
\big|D^{-1}_{\ell,\ell}\big|
\le \sfrac{\Gam''_{10}}{1+|\ell_0|+\Si_{\nu=1}^3|\ell_\nu|^2}
\end{equation*}
for some $\Gam''_{10}$ which is independent of $c'_0$.
By parts (b) and (c) of Lemma \ref{lemPBSunppties} 
and parts (a) and (b) of Proposition \ref{propPBSAnppties},
\begin{equation*}
\big|B_{\ell,\ell'}\big|\le\cst{}{a,\fq}\ |k|\smprod_{\nu=0}^3  \big(\sfrac{24}{|\ell_\nu|+\pi}\big)^\fq
                 \smprod_{\nu=0}^3  \big(\sfrac{24}{|\ell'_\nu|+\pi}\big)^\fq
\end{equation*}
So if $|k|<c'_0$ with $c'_0$ small enough, $D+B$ is invertible with
the inverse given by the Neumann expansion 
$D^{-1}+\sum_{p=1}^\infty (-1)^p D^{-1}\big(BD^{-1}\big)^p$. 
Since $D$ and $B$ are both analytic on $|\Im k|<2$ and
\begin{equation*}
\sum_{\ell\in 2\pi\bbbz^4} \cst{}{a,\fq}\ |k|
    \Big(\smprod_{\nu=0}^3  \sfrac{24}{|\ell_\nu|+\pi}\Big)^\fq
     \sfrac{\Gam''_{10}}{1+|\ell_0|+\Si_{\nu=1}^3|\ell_\nu|^2}
    \Big(\smprod_{\nu=0}^3  \sfrac{24}{|\ell_\nu|+\pi}\Big)^\fq
<\half
\end{equation*}
if $c'_0$ is small enough, we again get the desired analyticity and
bound on $\big|\widehat S_{n,k}(\ell,\ell')\big|$.

\Item (c) Just apply Remark \ref{remPOGpropLapl} and parts (a) and (b) of 
this lemma,  Lemma \ref{lemPBSunppties}.a and  Proposition \ref{propPBSAnppties}.a 
and the fact that
\begin{equation*}
\sum_{\ell\in 2\pi\bbbz^4} 
    \Big(\smprod_{\nu=0}^3  \sfrac{1}{|\ell_\nu|+\pi}\Big)^\fq
     \sfrac{1}{1+|\ell_0|+\Si_{\nu=1}^3|\ell_\nu|^2}
    \Big(\smprod_{\nu=0}^3  \sfrac{1}{|\ell_\nu|+\pi}\Big)^\fq
\end{equation*}
is bounded uniformly in $n$ and $L$ to \eqref{eqnPOGdeSnft}.

\Item (d) The bound on $\big|S_n(u,u')-S'_n(u,u')\big|$
follows from part (c) by \cite[Lemma \lemBOlonelinfty.b]{Bloch} with the
replacements \eqref{eqnPOGsubstspaces}.
The bound on $\big|S'_n(u,u')\big|$ follows from part (a),
noting in particular that $\sfrac{\Gam_{12}}{{(1+|p_0|+|\bp|^2)}^3}
\in L^1(\hat\cZ_\fin)$, and
\begin{equation*}
|u_\nu-u'_\nu|^jS'_n(u,u')
=\int_{\hat\cZ_\fin} 
     \sfrac{\partial^j\widehat S'_n}{\partial p_\nu^j}(p)\  
     e^{-ip\cdot(u-u')}\sfrac{d^4p}{(2\pi)^4}
\end{equation*}
and
\begin{align*}
\big|S'_n(u,u')\big|
&\le\bigg|\int_{\hat\cZ_\fin} 
      \widehat S'_n(p)\  
     e^{-ip\cdot(u-u')}\sfrac{d^4p}{(2\pi)^4}\bigg|
\le\Gam_{12}\int_{\hat\cZ_\fin} \sfrac{d^4p}{(2\pi)^4}
= \Gam_{12}\,L^{5n}
\end{align*}
\end{proof}

We now prove the analog of Lemma \ref{lemPOGSnppties} for the operators
$S_{n,\nu}^{(+)}$ and $S_{n,\nu}^{(-)}$ of \eqref{eqnPOGSnnuzerodef}.
As in \eqref{eqnPOGSntransapprox}--\eqref{eqnPOGdeSnft}, we decompose
$$
S_{n,\nu}^{(+)} ={S_n'}^{\!*} +\de S_\nu^{(+)}\qquad
S_{n,\nu}^{(-)} =S_n' +\de S_\nu^{(-)}
$$
with
\begin{align*}
\de S_\nu^{(+)} &= S_{n,\nu}^{(+)}-{S_n'}^{\!*}
          =- {S_n'}^{\!*}\big[Q_{n,\nu}^{(+)}\fQ_nQ_{n,\nu}^{(-)}
                       -a_n\exp\{-\De_n\}\big]S_{n,\nu}^{(+)} \\
\de S_\nu^{(-)} &= S_{n,\nu}^{(-)}-S_n'
          =- S'_n\big[Q_{n,\nu}^{(+)}\fQ_nQ_{n,\nu}^{(-)}
                       -a_n\exp\{-\De_n\}\big]S_{n,\nu}^{(-)}
\end{align*}
They have Fourier representations
\begin{equation}\label{eqnPOGdeSnnuft}
\begin{split}
\widehat{\de S_{\nu,k}^{(+)}}(\ell,\ell')
&=-\hskip-8pt\sum_{\ell''\in\hat\cB_n}\hskip-4.5pt
     \overline{\widehat S'_n(k\!+\!\ell)}
    \big[ U^{(+)}_{n,\nu}(k,\ell)\hat \fQ_n(k) U^{(-)}_{n,\nu}(k,\ell'')
             -a_ne^{-\De_n(k\!+\!\ell)}\de_{\ell,\ell''}\big]
    \widehat{S}_{n,\nu,k}^{(+)}(\ell''\!,\ell')\\
\widehat{\de S_{\nu,k}^{(-)}}(\ell,\ell')
&=-\hskip-8pt\sum_{\ell''\in\hat\cB_n}\hskip-4.5pt
     \widehat S'_n(k\!+\!\ell)
    \big[ U^{(+)}_{n,\nu}(k,\ell)\hat \fQ_n(k) U^{(-)}_{n,\nu}(k,\ell'')
             -a_ne^{-\De_n(k\!+\!\ell)}\de_{\ell,\ell''}\big]
    \widehat{S}_{n,\nu,k}^{(-)}(\ell''\!,\ell')
\end{split}
\end{equation}
where, by \eqref{eqnPBSqnplusminus},
\begin{align*}
U^{(+)}_{n,\nu}(k,\ell)&=\ze_{n,\nu}^{(+)}(k,\ell)u_{n,\nu}^{(+)}(k+\ell)
                                                u_n(k+\ell)^{\fq-1}
\\
U^{(-)}_{n,\nu}(k,\ell'')&=\ze_{n,\nu}^{(-)}(k,\ell'')u_{n,\nu}^{(-)}(k+\ell'')
                                                  u_n(k+\ell'')^{\fq-1} 
\end{align*}

\begin{lemma}\label{lemPOGSnnuppties}
There are constants $\mm_5>0$ and $\Gam_{13}$ 
such that the following hold for all $L>\Gam_2$.
\begin{enumerate}[label=(\alph*), leftmargin=*]
\item  
For all $\ell,\ell'\in\hat\cB_n$,
$\widehat S_{n,\nu,k}^{(\pm)}(\ell,\ell')$ is analytic in $|\Im k|<3\mm_5$ and obeys
\begin{align*}
\big|\widehat S_{n,\nu,k}^{(\pm)}(\ell,\ell')\big|
&\le \sfrac{\Gam_{13}}{1+|\ell_0|+\Si_{\nu=1}^3|\ell_\nu|^2}
\Big\{\de_{\ell,\ell'}
+\sfrac{1}{1+|\ell'_0|+\Si_{\nu=1}^3|\ell'_\nu|^2}
    \smprod_{\nu=0}^3\! \sfrac{1}{(|\ell_\nu|+1)^{\fq-1}}
    \smprod_{\nu=0}^3\! \sfrac{1}{(|\ell'_\nu|+1)^\fq}
\Big\}
\end{align*}
there.

\item 
For all $\ell,\ell'\in\hat\cB_n$,
$\widehat{\de S}_k(\ell,\ell')$ is analytic in $|\Im k|<3\mm_5$ and obeys
\begin{align*}
\big|\widehat{\de S_{\nu,k}^{(\pm)}}(\ell,\ell')\big|
&\le\Gam_{13}\exp\big\{-\sfrac{1}{40}\Si_{\nu=0}^3|\ell_\nu|^2\big\}\ \de_{\ell,\ell'}\\
&\hskip0.8in+\sfrac{\Gam_{13}}{1+|\ell_0|+\Si_{\nu=1}^3|\ell_\nu|^2}\ 
    \smprod_{\nu=0}^3 \sfrac{1}{(|\ell_\nu|+1)^{\fq-1}}
    \smprod_{\nu=0}^3 \sfrac{1}{(|\ell'_\nu|+1)^\fq} \ \ 
     \sfrac{1}{1+|\ell'_0|+\Si_{\nu=1}^3|\ell'_\nu|^2}
\end{align*}
there.

\item
For all $u,u'\in\cX_n$,
$
\ \big|S_{n,\nu}^{(\pm)}(u,u')-S'_n(u,u')\big|\le \Gam_{13} e^{-2\mm_5 |u-u'|}
$.
\end{enumerate}
\end{lemma}
\begin{proof}  
(a) For any $c''_0>0$ we have, for $|k|\ge c''_0$ and $|\Im k_0|<3\mm_5$, 
analyticity and the desired bound
follows from the representations 
(apply \cite[Remark \remBSedA.b]{BlockSpin} with
$\,R= Q_{n,\nu}^{(-)}\,$, $\,R_*= Q_{n,\nu}^{(+)}\,$ and use 
Remark \ref{remPBSunderivAlg} to give $\,R\,D^{-1} R_* = Q_-D^{-1}Q_-^*\,$)
\begin{equation*}
S_{n,\nu}^{(+)} = {D_n^*}^{-1} - 
        {D_n^*}^{-1}\,Q_{n,\nu}^{(+)}\De^{(n)} Q_{n,\nu}^{(-)}{D_n^*}^{-1} \qquad
S_{n,\nu}^{(-)} = D_n^{-1} - D_n^{-1}\,Q_{n,\nu}^{(+)}\De^{(n)} Q_{n,\nu}^{(-)}D_n^{-1} 
\end{equation*}
and Lemmas \ref{lemPDOhatSzeroppties}.d, \ref{lemPBSunppties}.a, 
\ref{lemPBSunderiv}.b and \ref{lemPOCDenppties}.c.

For $|k|< c'_0$, with $c'_0$ to be shortly chosen sufficiently small,
and $|\Im k|<3\mm_5$ we use the representation
\begin{align}\label{eqnPERSnnukell}
&\big(\hat S_{n,\nu,k}^{(-)}\big)^{-1}(\ell,\ell')
=\hat\bD_n(k+\ell)\,\de_{\ell,\ell'}
  +U^{(+)}_{n,\nu}(k,\ell)\hat \fQ_n(k)U^{(-)}_{n,\nu}(k,\ell')
=D_{\ell,\ell'}+B_{\ell,\ell'}
\end{align}
with
\begin{align*}
D_{\ell,\ell'}
&=\hat\bD_n(k+\ell)\, \de_{\ell,\ell'}
  +\begin{cases} a_n & \text{if $\ell,\ell'=0$}\\
                   \noalign{\vskip0.02in}
                   0 & \text{otherwise}
    \end{cases}
\\ \noalign{\vskip0.1in}
B_{\ell,\ell'}
&=\begin{cases}\hat \fQ_n(k)u_n(k)^{2\fq}-a_n  & \text{if $\ell=\ell'=0$}\\
               \noalign{\vskip0.02in}
          U^{(+)}_{n,\nu}(k,\ell) \hat \fQ_n(k) U^{(-)}_{n,\nu}(k,\ell')  & 
                           \text{otherwise}
              \end{cases}
\end{align*}
and the obvious analog for $\big(\hat S_{n,\nu,k}^{(+)}\big)^{-1}$
--- just replace $\hat\bD_n(k+\ell)$ with its
complex conjugate.
As in the proof of Lemma \ref{lemPOGSnppties}.b, $D$ is invertible and the 
inverse is a diagonal matrix with every diagonal matrix element obeying
\begin{equation*}
\big|D^{-1}_{\ell,\ell}\big|
\le \sfrac{\Gam'_{11}}{1+|\ell_0|+\Si_{\nu=1}^3|\ell_\nu|^2}
\end{equation*}
for some $\Gam'_{11}$ which is independent of $c'_0$.
By parts (b) and (c) of Lemma \ref{lemPBSunppties},
parts (a) and (b) of Proposition \ref{propPBSAnppties}, and part (b)
of Lemma \ref{lemPBSunderiv},
\begin{equation*}
\big|B_{\ell,\ell'}\big|\le\cst{}{a,\fq}\ |k|\smprod_{\nu=0}^3  \big(\sfrac{24}{|\ell_\nu|+\pi}\big)^{\fq-1}
                 \smprod_{\nu=0}^3  \big(\sfrac{24}{|\ell'_\nu|+\pi}\big)^\fq
\end{equation*}
So if $|k|<c'_0$ with $c'_0$ small enough, $D+B$ is invertible with
the inverse given by the Neumann expansion 
$D^{-1}+\sum_{p=1}^\infty (-1)^p D^{-1}\big(BD^{-1}\big)^p$. 
As $\fq>1$, the desired analyticity and the desired bound on 
$\big|\widehat S_{n,\nu,k}^{(\pm)}(\ell,\ell')\big|$ follow as
in the proof of Lemma \ref{lemPOGSnppties}.b

\Item (b) is proven just as Lemma \ref{lemPOGSnppties}.c.

\Item (c) follows from part (b) by \cite[Lemma \lemBOlonelinfty.b]{Bloch}.
\end{proof}

\begin{proof}[Proof of Proposition \ref{POGmainpos}]
Set $\mm_3=\min\{\mm_4,\mm_5\}$.
As $\sfrac{1}{|u_0|^2+|\bu|^4}$ is locally integrable in $\bbbr^4$,
the pointwise bounds on $\big|S_n(u,u')-S'_n(u,u')\big|$ and
$\big|S'_n(u,u')\big|$, given in Lemma \ref{lemPOGSnppties}.d, and on 
$\big|S_{n,\nu}^{(\pm)}(u,u')-S'_n(u,u')\big|$, given in 
Lemma \ref{lemPOGSnnuppties}.c, imply
$$
\|\ S_n\|_{\mm_3}, \|\ S_{n,\nu}^{(\pm)}\|_{\mm_3} \le \tilde \Gam_{11}
$$
with $ \tilde \Gam_{11}$ a constant, depending only on $\mm_3$, times 
$\max\{\Gam_{12},\Gam_{13}\}$.
Setting $\mu_{\rm up}= \sfrac{1}{2\tilde \Gam_{11}}$ and 
$\Gam_{11}$ to be the maximum of $2\tilde \Gam_{11}$ (for 
$\|\ S_n(\mu) \|_{\mm_3}$ and$\|\ S_{n,\nu}^{(\pm)}(\mu) \|_{\mm_3}$)
and $2\tilde \Gam_{11}^2$ (for 
$\|\ S_n(\mu)-S_n \|_{\mm_3}$ and$\|\ S_{n,\nu}^{(\pm)}(\mu)
-S_{n,\nu}^{(\pm)} \|_{\mm_3}$)
a  Neumann expansion gives the specified bounds.

\end{proof}

We now formulate and prove two more technical lemmas that will be used
elsewhere.
\begin{lemma}\label{lemPOGsnQstar}
There are constants $\mm_6>0$ and $\Gam_{14}$ 
such that, for all $L>\Gam_2$,
\begin{equation*}
\big|\big(S_nQ_n^*\big)(y,x)\big|\le \Gam_{14}e^{-2\mm_6 |x-y|}\qquad
\|S_nQ_n^*\|_{\mm_6} \le \Gam_{14}
\end{equation*}
\end{lemma}
\begin{proof} From the definitions of $S_n$,  in \eqref{eqnPOGSnnuzerodef}, and 
$\De^{(n)}$, at the beginning of \S\ref{secPOcovariance}, one sees directly
that
\begin{equation}\label{eqnPOGsq}
S_n^{(*)}Q_n^*= {D_n^{(*)}}^{-1}Q_n^*\De^{(n)(*)}\fQ_n^{-1}: L^2(\cX_\crs) \rightarrow L^2(\cX_\fin)
\end{equation}
The Fourier transform of the kernel, $b(y,x)$, of the operator $S_nQ_n^*$
is
\begin{equation*}
\hat b_\rk(\ell) = \hat\bD_n^{-1}(\rk+\ell)u_n(\rk+\ell)^\fq
                        \hat\De^{(n)}(\rk)\sfrac{1}{\hat \fQ_n(\rk)}
\end{equation*}
By Lemma \ref{lemPOCDenppties}.a,c,f,\ \  
   Remark \ref{remPDOftDn}.b and Lemma \ref{lemPDOhatSzeroppties}.d,\ \ 
   Proposition \ref{propPBSAnppties}.a,\ \ 
   Remark \ref{remPBSqnft}.d  and Lemma \ref{lemPBSunppties}.a, 
$\hat b_\rk(\ell)$ is analytic in $|\Im \rk|<3\mm_6$ and
\begin{equation*}
\big|\hat b_\rk(\ell)\big|
\le \frac{\sfrac{5}{4a}\Gam_6}{1+|\rk_0+\ell_0|+\smsum_{\nu=1}^3 |\rk_\nu+\ell_\nu|^2}
     \prod_{\nu=0}^3 \Big(\frac{24}{|\ell_\nu|+\pi}\Big)^{\fq}
\end{equation*}
The bound is uniform in $n$ and $L$ and is summable in $\ell$, so the claims
follow from \cite[Lemma \lemBOlonelinfty.c]{Bloch}.
\end{proof}

\begin{lemma}\label{lemPOGrightinverse}
There are constants $\mm_7>0$ and $\Gam_{15}$
such that, for all $L>\Gam_2$, the operators 
\begin{align*}
&\big(S_n(\mu)Q_n^*\fQ_n\big)^*\big(S_n(\mu)Q_n^*\fQ_n\big),\ 
\big(S_n(\mu)^*Q_n^*\fQ_n\big)^*\big(S_n(\mu)^*Q_n^*\fQ_n\big)\\
&\big(S_{n,\nu}^{(\pm)}(\mu)Q_{n,\nu}^{(+)} \fQ_n\big)^*
       \big(S_{n,\nu}^{(\pm)}(\mu)Q_{n,\nu}^{(+)} \fQ_n\big)
\end{align*}
all have bounded inverses. The $\|\ \cdot\ \|_{\mm_7}$ norms of the inverses
are all bounded by $\Gam_{15}$.

\end{lemma}
\begin{proof} We first consider the case that $\mu=0$.
By \eqref{eqnPOGsq}, the operator 
\begin{equation*}
S_nQ_n^*\fQ_n=D_n^{-1}Q_n^*\De^{(n)}:L^2(\cX_\crs)\rightarrow L^2(\cX_\fin)
\end{equation*}
has Fourier transform 
\begin{equation*}
\tilde b_k(\ell) = \hat\bD_n^{-1}(k+\ell)u_n(k+\ell)^\fq
                        \hat\De^{(n)}(k)
\end{equation*}
The operator $(S_nQ_n^*\fQ_n)^*(S_nQ_n^*\fQ_n)$ 
maps $L^2(\cX_\crs)$ to $L^2(\cX_\crs)$ and has Fourier transform 
\begin{equation*}
\sum_{\ell\in\hat\cB}
      \tilde b_{-k}(-\ell)\tilde b_k(\ell) 
= \sum_{\ell\in\hat\cB}
       \hat\bD_n^{-1}(-k-\ell)
        \hat\bD_n^{-1}(k+\ell)u_n(k+\ell)^{2\fq}
                        \hat\De^{(n)}(-k)\hat\De^{(n)}(k)
\end{equation*}
For $k$ real,
\begin{align*}
\sum_{\ell\in\hat\cB}
      \tilde b_{-k}(-\ell)\tilde b_k(\ell) 
&= \sum_{\ell\in\hat\cB} \big|\hat\bD_n^{-1}(k+\ell)u_n(k+\ell)^{\fq}\hat\De^{(n)}(k)\big|^2
\\
&\ge \big|\hat\bD_n^{-1}(k)u_n(k)^{\fq}\hat\De^{(n)}(k)\big|^2
\\
&\ge \gam_2^2\inf_{|k_\nu|\le\pi}|u_n(k)|^{2\fq}
\\
&\ge \gam_2^2\big(\sfrac{2}{\pi}\big)^{8\fq}
\end{align*}
by Lemmas \ref{lemPOCDenppties}.e and \ref{lemPBSunppties}.f. To show that half this
the lower bound extends into a strip along the real axis that has 
width independent of $n$ and $L$, we observe that
\begin{itemize}[leftmargin=*, topsep=2pt, itemsep=0pt, parsep=0pt]
\item
all first order derivatives of $u_n(k+\ell)$ are 
uniformly bounded by $2 \smprod_{\nu=0}^3 \sfrac{24}{|\ell_\nu|+\pi} $
on such a strip by the Cauchy integral formula and Remark \ref{remPBSqnft}.d and 
Lemma \ref{lemPBSunppties}.a and 
\item
$u_n(k+\ell)$ itself is 
uniformly bounded by $\smprod_{\nu=0}^3 \sfrac{24}{|\ell_\nu|+\pi} $
on such a strip by Lemma \ref{lemPBSunppties}.a and 
\item
all first order derivatives of $\hat\De^{(n)}(k)$ are 
uniformly bounded on such a strip by Lemma \ref{lemPOCDenppties}.c and
\item
for $\ell\ne 0$, all first order derivatives of 
$\hat\bD_n^{-1}(k+\ell)$ are uniformly bounded on such a 
strip by parts (c) and (d) of Lemma \ref{lemPDOhatSzeroppties} and
\item
for $\ell= 0$, all first order derivatives of 
$\hat\bD_n^{-1}(k)\hat\De^{(n)}(k)$ are uniformly bounded on such a 
strip by Lemma \ref{lemPOCDenppties}.f.
\end{itemize}

\noindent
The operator $S_n^*Q_n^*\fQ_n= {D_n^{-1}}^*Q_n^*\De^{(n)*}$
has Fourier transform $\tilde b_{-k}(-\ell)$. So
\begin{equation*}
(S_nQ_n^*\fQ_n)^*(S_nQ_n^*\fQ_n)=(S_n^*Q_n^*\fQ_n)^*(S_n^*Q_n^*\fQ_n)
\end{equation*}
The operators $S_{n,\nu}^{(-)}Q_{n,\nu}^{(+)} \fQ_n
                  =D_n^{-1}Q_{n,\nu}^{(+)}\De^{(n)}$ 
and 
$S_{n,\nu}^{(+)}Q_{n,\nu}^{(+)} \fQ_n
                  ={D_n^{-1}}^*Q_{n,\nu}^{(+)}\De^{(n)*}$
map $L^2(\cX_\crs)$ to $L^2(\cX_\fin)$ and have Fourier transforms 
\begin{equation*}
\tilde c_k(\ell) = \hat\bD_n^{-1}(\pm k\pm\ell)
        \ze_{n,\nu}^{(+)}(k,\ell)u_{n,\nu}^{(+)}(k+\ell)u_n(k+\ell)^{\fq-1}
                        \hat\De^{(n)}(\pm k)
\end{equation*}
So the operators 
$(S_{n,\nu}^{(\pm)}Q_{n,\nu}^{(+)} \fQ_n)^*(S_{n,\nu}^{(\pm)}Q_{n,\nu}^{(+)} \fQ_n)$
both map $L^2(\cX_\crs)$ to $L^2(\cX_\crs)$ and have Fourier transform 
\begin{align*}
&\sum_{\ell\in\hat\cB}
      \tilde c_{-k}(-\ell)\tilde c_k(\ell) \\
&\hskip0.5in= \sum_{\ell\in\hat\cB}
       \hat\bD_n^{-1}(-k-\ell)
        \hat\bD_n^{-1}(k+\ell)u_{n,\nu}^{(+)}(k+\ell)^2
                        u_n(k+\ell)^{2(\fq-1)}
                        \hat\De^{(n)}(-k)\hat\De^{(n)}(k)
\end{align*}
This is bounded just as $\sum_{\ell\in\hat\cB}
      \tilde b_{-k}(-\ell)\tilde b_k(\ell)$ was. 
The specified bounds, in the special case that $\mu=0$ follow.

By Proposition \ref{POGmainpos}, a Neumann expansion
gives the desired bounds when $\mu$ is nonzero.

\end{proof}

\newpage
\section{The Degree One Part of the Critical Field}\label{secPOleadingOrder}

In \cite[Proposition \propCFpsisoln\ and (\eqnBGEpsiNexp)]{BGE} we derive 
an expansion for the critical fields of the form 
$$
\psi_{(*)n}(\th_*,\th,\mu,\cV) 
=\sfrac{a}{L^2} C^{(n)}(\mu)^{(*)} Q^*\th_{(*)}
   +\psi_{(*)n}^{(\ge 3)}(\th_*,\th,\mu,\cV) 
$$
with the $C^{(n)}(\mu)$ of \cite[Proposition \propHTexistencecriticalfields]{PAR1}
and with $\psi_{(*)n}^{(\ge 3)}$ being of degree at least $3$ in $\th_{(*)}$.
In this section we derive bounds on a scaled version of 
$C^{(n)}(\mu)^{(*)} Q^*\th_{(*)}$ and some related operators. To do so
we use the representation 
\begin{equation}\label{eqnPOLanotherC}
\sfrac{a}{L^2} C^{(n)}(\mu)^{(*)} Q^*
=\big(\sfrac{a}{L^2}Q^* Q+\fQ_n\big)^{-1}
    \big\{\sfrac{a}{L^2}Q^*  
   +\fQ_n Q_n \check S_{n+1}(\mu)^{(*)}\check Q_{n+1}^* \check\fQ_{n+1}
   \big\}
\end{equation}
of \cite[(\eqnBGEanotherC)]{BGE}. Here, as in \cite[Lemma \lemSCacheckOne\ and 
(\eqnBGEcheckSmu)]{BGE},
\begin{alignat}{5}\label{eqnPOLcheckDefs}
\check Q_{n+1}
     &= \bbbl_*Q_{n+1}\bbbl_*^{-1}
     &&=QQ_n 
     &&:\cH_n\rightarrow\cH_{-1}^{(n+1)}\notag\\
\check\fQ_{n+1}^{-1}
     &=L^2\, \bbbl_*\fQ_{n+1}^{-1}\bbbl_*^{-1}
     &&=\sfrac{1}{aL^{-2}}\bbbone+Q\fQ_n^{-1}Q^*
     &&:\cH_{-1}^{(n+1)}\rightarrow\cH_{-1}^{(n+1)}\notag\\
\check S_{n+1}(\mu)
      &=  L^2\, \bbbl_*S_{n+1}(L^2\mu)\bbbl_*^{-1}
     &&=\big\{D_n-\mu
        +\check Q_{n+1}^*\check\fQ_{n+1}\check Q_{n+1}\big\}^{-1}
     &&:\cH_n\rightarrow\cH_n 
\end{alignat}
These operators are all translation invariant with respect to
$\cX_{-1}^{(n+1)}$. As
\begin{equation*}
\check S_{n+1}(\mu) = \check S_{n+1} +\mu \check S_{n+1} \check S_{n+1}(\mu)
\qquad\text{with }\check S_{n+1}=\check S_{n+1}(0)
\end{equation*}
we have
\begin{equation}\label{eqnPOLcmunomu}
\sfrac{a}{L^2} C^{(n)}(\mu)^{(*)} Q^*
=\sfrac{a}{L^2} {C^{(n)}}^{(*)} Q^*
+\mu A_{\psi,\phi} \check S_{n+1}^{(*)} \check S_{n+1}(\mu)^{(*)}
            \check Q_{n+1}^* \check\fQ_{n+1}
\end{equation}
where
\begin{equation}\label{eqnPOLapsiphi}
A_{\psi,\phi}=(aL^{-2}Q^* Q+\fQ_n)^{-1}\fQ_nQ_n 
   :\cH_n\rightarrow\cH_0^{(n)}
\end{equation}
The operator $A_{\psi,\phi}$ is also used in the course of bounding 
$\psi_{(*)n}^{(\ge 3)}$ in \cite[Proposition \propCFpsisoln]{BGE}.

The main results of this section are
\begin{proposition}\label{propPOLmain}
There are constants\footnote{Recall Convention \ref{convPOconstants}.}
$\mm_8>0$ and $\Gam_{16}, \Gam_{17}$
such that the following holds, for each $L>\Gam_{17}$ and each $\mu$
obeying $|L^2\mu|\le\mu_{\rm up}$.

\begin{enumerate}[label=(\alph*), leftmargin=*]

\item \ \ \ 
$
\big\|\bbbl_*^{-1}A_{\psi,\phi}\bbbl_*\big\|_{m=1}\le\Gam_{16}
$
\ \ \ and \ \ \ 
$
\big\|\bbbl_*^{-1}\ \sfrac{a}{L^2} C^{(n)}(\mu)^{(*)} Q^*\  \bbbl_*\big\|_{\mm_8}\le\Gam_{16}
$

\item   
Let $0\le\nu\le3$. There are operators 
$A_{\psi,\phi,\nu}$ and $A_{\psi_{(*)}\th_{(*)}\nu}(\mu)$
such that
\begin{equation*}
\partial_\nu A_{\psi,\phi} = A_{\psi,\phi,\nu}\partial_\nu\qquad
\partial_\nu \sfrac{a}{L^2} C^{(n)}(\mu)^{(*)} Q^*
       = A_{\psi_{(*)}\th_{(*)}\nu}(\mu)\ \partial_\nu
\end{equation*}
and
\begin{equation*}
\|\bbbl_*^{-1} A_{\psi,\phi,\nu}\bbbl_*\|_{m=1}\le\Gam_{16} \qquad
\big\|\bbbl_*^{-1}A_{\psi_{(*)}\th_{(*)}\nu}(\mu)\bbbl_*\big\|_{\mm_8}
   \le\Gam_{16}
\end{equation*}
\end{enumerate}
\end{proposition}
\noindent
This proposition is proven at the end of this section, 
after Lemma \ref{lemPOLcheckDeppties}.
In this proof we write 
$\sfrac{a}{L^2} {C^{(n)}}^{(*)}Q^* = A_{\psi_{(*)},\th_{(*)}}$
so that
\begin{align*}
A_{\psi_{(*)},\th_{(*)}}&=(aL^{-2}Q^* Q+\fQ_n)^{-1}
    \big\{aL^{-2}Q^* 
          +\fQ_nQ_n\check S^{(*)}_{n+1}\check Q_{n+1}^*\check\fQ_{n+1}\big\}
   :\cH_{-1}^{(n+1)}\rightarrow\cH_0^{(n)}
\end{align*}

\begin{remark}\label{remPOLderivAlg}
\ 
\begin{enumerate}[label=(\alph*), leftmargin=*]
\item  $A_{\psi_{(*)}\th_{(*)}}=\fQ_n^{-1}Q^* \check\fQ_{n+1} 
     + A_{\psi,\phi}{D_n^{-1}}^{(*)}\check Q_{n+1}^*\ \check\De^{(n+1)(*)}$
with
\begin{align*}
\check\De^{(n+1)(*)}
   &=\sfrac{1}{L^2}\bbbl_*\De^{(n+1)(*)}\bbbl_*^{-1}
   =\big\{\bbbone
        +\check\fQ_{n+1}\check Q_{n+1}{D_n^{-1}}^{(*)}
            \check Q_{n+1}^*\big\}^{-1}\check\fQ_{n+1}
\end{align*}
being a fully translation invariant operator on $\cH_{-1}^{(n+1)}$.

\item  Let $0\le\nu\le3$. We have
$\partial_\nu A_{\psi,\phi}=A_{\psi,\phi,\nu}\partial_\nu$
and
$
\partial_\nu A_{\psi_{(*)}\th_{(*)}}(\mu)
   = A_{\psi_{(*)}\th_{(*)}\nu}(\mu) \partial_\nu
$
where
\begin{align*}
A_{\psi,\phi,\nu}
 &=\big[\bbbone - \fQ_n^{-1}Q^{(+)}_{+,\nu}\check\fQ_{n+1} Q^{(-)}_{+,\nu}
   \big]Q_{n,\nu}^{-} \\
A_{\psi_{(*)}\th_{(*)}\nu} 
 &= \fQ_n^{-1}Q_{+,\nu}^{(+)}\check  \fQ_{n+1} 
     + A_{\psi,\phi,\nu}\bbbl_*D_{n+1}^{{-1}(*)}Q_{n+1,\nu}^{(+)}\De^{(n+1)(*)}
                \bbbl_*^{-1} \\
A_{\psi_*\th_*\nu}(\mu)
 &= A_{\psi_*\th_*\nu}
  +L^2\mu A_{\psi,\phi,\nu}\bbbl_* S_{n+1,\nu}^{(+)} S_{n+1,\nu}^{(+)}(L^2\mu)
            Q_{n+1,\nu}^{(+)}  \fQ_{n+1}\bbbl_*^{-1} \\
A_{\psi\th\nu}(\mu)
 &= A_{\psi\th\nu}
 +L^2\mu A_{\psi,\phi,\nu}\bbbl_*  S_{n+1,\nu}^{(-)} S_{n+1,\nu}^{(-)}(L^2\mu)
            Q_{n+1,\nu}^{(+)} \fQ_{n+1}\bbbl_*^{-1} 
\end{align*}
\end{enumerate}
\end{remark}
\begin{proof} 
(a) First observe that, by \eqref{eqnPOGsq},
\begin{align}\label{eqnPOLapsith}
Q_n\check S_{n+1}^{(*)}\check Q_{n+1}^*\check\fQ_{n+1}
&= Q_n\bbbl_* {D_{n+1}^{(*)}}^{-1} Q_{n+1}^*\De^{(n+1)(*)}\bbbl_*^{-1} 
= (Q_n {D_n^{-1}}^{(*)}Q_n^*)\ Q^*\ \check\De^{(n+1)(*)}
\end{align}
Using \eqref{eqnPOLapsith}, the operator
\begin{align*}
A_{\psi_{(*)}\th_{(*)}} 
 &= (aL^{-2}Q^* Q+\fQ_n)^{-1}
    \big\{aL^{-2}Q^*
      +\fQ_nQ_n{D_n^{-1}}^{(*)}\check Q_{n+1}^*\ \check\De^{(n+1)(*)}\big\} 
      \\
 &= \fQ_n^{-1}(Q^* Q\fQ_n^{-1}+\sfrac{1}{aL^{-2}}\bbbone)^{-1} Q^*
        + A_{\psi,\phi}{D_n^{-1}}^{(*)}\check Q_{n+1}^*\ \check\De^{(n+1)(*)} 
      \\
 &=  \fQ_n^{-1}Q^* (Q\fQ_n^{-1}Q^*+\sfrac{1}{aL^{-2}}\bbbone)^{-1}
       + A_{\psi,\phi}{D_n^{-1}}^{(*)}\check Q_{n+1}^*\ \check\De^{(n+1)(*)} 
      \\
 &=  \fQ_n^{-1}Q^* \check\fQ_{n+1} 
     + A_{\psi,\phi}{D_n^{-1}}^{(*)}\check Q_{n+1}^*\ \check\De^{(n+1)(*)}
\end{align*}

\Item (b)
By Remark \ref{remPBSunderivAlg},
\begin{align*}
\partial_\nu A_{\psi,\phi}
&=\partial_\nu(\bbbone+aL^{-2}\fQ_n^{-1}Q^* Q)^{-1}Q_n \\
&=\partial_\nu Q_n 
-\partial_\nu aL^{-2}\fQ_n^{-1}Q^* (\bbbone+aL^{-2}Q\fQ_n^{-1}Q^*)^{-1}QQ_n \\
&= Q_{n,\nu}^{-}\partial_\nu \!
- aL^{-2}\fQ_n^{-1}Q^{(+)}_{+,\nu}
     (\bbbone+aL^{-2}Q\fQ_n^{-1}Q^*)^{-1}Q^{(-)}_{+,\nu}Q^{(-)}_{n,\nu}
         \partial_\nu \\
&=\big[\bbbone - \fQ_n^{-1}Q^{(+)}_{+,\nu}\check\fQ_{n+1} Q^{(-)}_{+,\nu}
   \big]Q_{n,\nu}^{-} \partial_\nu
\end{align*}
Therefore 
by part (a), \eqref{eqnPDOdndef}, \eqref{eqnPOLcheckDefs} and Remark \ref{remPBSunderivAlg},
\begin{align*}
\partial_\nu A_{\psi_{(*)}\th_{(*)}} 
 &= \partial_\nu\big[\fQ_n^{-1}Q^* \check\fQ_{n+1} 
    + A_{\psi,\phi}\bbbl_*D_{n+1}^{{-1}(*)}Q_{n+1}^*\De^{(n+1)(*)}\bbbl_*^{-1}
   \big]
  =  A_{\psi_{(*)}\th_{(*)\nu}} \partial_\nu
\end{align*}
since
$
\partial_\nu \bbbl_*= \sfrac{1}{L_\nu}\bbbl_*\partial_\nu
$
by \cite[Remark \remSCscaling.a,b]{PAR1}.
To get 
$
\partial_\nu A_{\psi_{(*)}\th_{(*)}}(\mu)
   = A_{\psi_{(*)}\th_{(*)}\nu}(\mu) \partial_\nu
$
when $\mu\ne 0$, write, using \eqref{eqnPOLcheckDefs},
\begin{align*}
A_{\psi_{(*)}\th_{(*)}}(\mu)
&=A_{\psi_{(*)}\th_{(*)}} 
+\mu A_{\psi,\phi} \check S_{n+1}^{(*)} \check S_{n+1}(\mu)^{(*)}
            \check Q_{n+1}^* \check\fQ_{n+1}\\
&=A_{\psi_{(*)}\th_{(*)}} 
+L^2\mu A_{\psi,\phi} \bbbl_*S_{n+1}^{(*)} S_{n+1}(L^2\mu)^{(*)}
             Q_{n+1}^* \fQ_{n+1}\bbbl_*^{-1}
\end{align*}
and use \eqref{eqnPOGpartialS}, Remark \ref{remPBSunderivAlg} and the fact that
$\fQ_{n+1}$ is fully translation invariant.

\end{proof}

The operators of principal interest, $A_{\psi_*,\th_*}$ and $A_{\psi,\th}$,
act from $L^2(\cX_\crs)=\cH_{-1}^{(n+1)}$ to
$L^2(\cX_\fin)=\cH_0^{(n)}$ with
\begin{equation*}
\cX_\fin=\cX_{0}^{(n)}\qquad
\cX_\crs=\cX_{-1}^{(n+1)}\qquad
\cB=\cB^+
\end{equation*}
We now give a bunch of Fourier transforms (in the sense of 
\cite[(\eqnBOifkerkell) and  (\eqnBOasymmetric)]{Bloch} -- but we shall 
suppress the $\hat{\ }$ from 
the notation). All of the operators above are periodized in the sense of
\cite[Definition \defBOperiodization]{Bloch}. As before we denote
\begin{itemize}[leftmargin=*, topsep=2pt, itemsep=0pt, parsep=0pt]
\item 
    momenta dual to the $L$--lattice $L^2\bbbz\times L\bbbz^3$ by
    $\fk\in \big(\bbbr/\sfrac{2\pi}{L^2}\bbbz\big)
            \times \big(\bbbr^3/\sfrac{2\pi}{L}\bbbz^3\big)$, 
\item
momenta dual to the unit lattice $\bbbz\times \bbbz^3$ by
$k\in \big(\bbbr/2\pi\bbbz\big) \times \big(\bbbr^3/2\pi\bbbz^3\big)$
and decompose $k=\fk+\ell$ or $k=\fk+\ell'$ with $\fk$
in a fundamental cell for $\big(\bbbr/\sfrac{2\pi}{L^2}\bbbz\big)
\times \big(\bbbr^3/\sfrac{2\pi}{L}\bbbz^3\big)$
and $\ell,\ell'\in \big(\sfrac{2\pi}{L^2}\bbbz/2\pi\bbbz\big)
      \times\big(\sfrac{2\pi}{L}\bbbz^3/2\pi\bbbz^3\big)=\hat\cB^+$ and 
\item
momenta dual to the $\veps_j$--lattice
$\veps_j^2\bbbz\times \veps_j\bbbz^3$ by
$p_j\in \big(\bbbr/\sfrac{2\pi}{\veps_j^2}\bbbz\big)
      \times \big(\bbbr^3/\sfrac{2\pi}{\veps_j}\bbbz^3\big)$
and decompose $p_j=k+\ell_j$ with $k$ in a fundamental cell for $\big(\bbbr/2\pi\bbbz\big) \times \big(\bbbr^3/2\pi\bbbz^3\big)$
and $\ell_j\in\big(2\pi\bbbz/\sfrac{2\pi}{\veps_j^2}\bbbz\big)
    \times\big(2\pi\bbbz^3/\sfrac{2\pi}{\veps_j}\bbbz^3\big)=\hat\cB_j$.
Here $1\le j\le n$.
\end{itemize}
The Fourier transform of $A_{\psi_{(*)}\th_{(*)}}$ is
\begin{equation}\label{eqnPOLapsithft}
\begin{split}
\big(A_{\psi_{(*)}\th_{(*)}}\big)_{\fk}(\ell)
&=\sum_{\ell'\in\hat\cB^+}
     \big(\sfrac{a}{L^2}Q^*Q+\fQ_n\big)^{-1}_{\fk}(\ell,\ell')
  \cr \noalign{\vskip-0.1in}&\hskip0.6in
        \Big\{\sfrac{a}{L^2}Q^*_{\fk}(\ell')
   +\fQ_n(\fk\!+\!\ell')\,
    \big(Q_n {D_n^{-1}}^{(*)}Q_n^*\big)(\fk\!+\!\ell')
    Q_{\fk}^*(\ell')
    \check\De^{(n+1)(*)}(\fk)\Big\}
\end{split}
\end{equation}
where, by Remark \ref{remPBSqnft}, Remark \ref{remPOLderivAlg}.a and \eqref{eqnPOLcheckDefs}
\begin{equation*}
\big(aL^{-2}Q^*Q+\fQ_n\big)_{\fk}(\ell,\ell')
=\fQ_n(\fk+\ell)\de_{\ell,\ell'}
        +aL^{-2}u_+(\fk+\ell)^\fq u_+(\fk+\ell')^\fq 
\end{equation*}
and
\begin{align*}
Q^*_{\fk}(\ell)&=u_+(\fk+\ell)^\fq \\
\fQ_n(k)&=a\Big[1+\sum_{j=1}^{n-1}
   \smsum_{\ell_j\in\hat\cB_j}\sfrac{1}{L^{2j}}u_j(k+\ell_j)^{2\fq}\Big]^{-1}\displaybreak[0]\\
(Q_n \bD_n^{{-1}(*)}Q_n^*)(k)&= \sum_{\ell_n\in\hat\cB_n}
    u_n(k+\ell_n)^{2\fq} \bD_n^{{-1}(*)}(k+\ell_n)\displaybreak[0]\\
\check\De^{(n+1)(*)}(\fk)&=\frac{  \check\fQ_{n+1}(\fk) }
{ 1 + \check\fQ_{n+1}(\fk)\hskip-4pt
    \sum\limits_{\ell_n\in\hat\cB_n\atop\ell\in\hat\cB^+}\hskip-4pt
u_n(\fk\!+\!\ell\!+\!\ell_n)^{2\fq}\, u_+(\fk\!+\!\ell)^{2\fq}\,
  \bD_n^{{-1}(*)}(\fk\!+\!\ell\!+\!\ell_n) }
\displaybreak[0]\\
\check\fQ_{n+1}(\fk)&=\sfrac{a}{L^2}\Big[1+
    \smsum_{\ell\in\hat\cB^+}\sfrac{a}{L^2}u_+(\fk+\ell)^{2\fq}
                   \fQ_n(\fk+\ell)^{-1}\Big]^{-1}\displaybreak[0]\\
\end{align*}

\begin{lemma}\label{lemPOLcheckAnppties}
Let $|\Im \fk_\nu|\le \sfrac{2}{L_\nu}$ for each $0\le\nu\le3$.
There is a constant $\Gam_{17}$, depending only on $\fq$, such that
the following hold for all $L>\Gam_{17}$.
\begin{enumerate}[label=(\alph*), leftmargin=*]
\item 
We have
$
\sfrac{6}{7}\sfrac{a}{L^2}\le |\check\fQ_{n+1}(\fk)| 
  \le \sfrac{6}{5} \sfrac{a}{L^2}
$
and
$
\Re \check\fQ_{n+1}(\fk)\ge \sfrac{a}{2L^2}
$.

\item
We have
\begin{align*}
&\Big|\big(aL^{-2}Q^*Q+\fQ_n\big)_{\fk}^{-1}(\ell,\ell')
                -\fQ_n(\fk+\ell)^{-1}\de_{\ell,\ell'}\Big|\\
&\hskip1.5in\le \sfrac{2}{aL^2}
                \smprod_{\nu=0}^3 
                               \big(\sfrac{24}{L_\nu|\ell_\nu|+\pi}\big)^\fq
                    \smprod_{\nu=0}^3 
                               \big(\sfrac{24}{L_\nu|\ell'_\nu|+\pi}\big)^\fq
                \ \begin{cases}
                      \qquad 1 & \text{for all $\ell$}\\
                        \noalign{\vskip0.05in}
                       \smprod\limits_{\atop{0\le\nu\le3}{\ell_\nu\ne 0}} 
                                            L_\nu|\fk_\nu|& \text{if $\ell\ne 0$}
                 \end{cases}
\end{align*}

\item Let $0\le\nu\le3$. Then
\begin{equation*}
\big|\big[\fQ_n^{-1}Q^{(+)}_{+,\nu}\check\fQ_{n+1} Q^{(-)}_{+,\nu}
   \big]_\fk(\ell,\ell')\big|\le 
        \sfrac{3e^4}{2L^2}
     \big(\sfrac{24}{\pi}\big)^4
  \smprod\limits_{\nu=0}^3 \big(\sfrac{24}{L_\nu|\ell_\nu|+\pi}\big)^{\fq-1}
   \smprod\limits_{\nu=0}^3 \big(\sfrac{24}{L_\nu|\ell'_\nu|+\pi}\big)^\fq
\end{equation*}
\end{enumerate}
\end{lemma}

\begin{proof} 
(a) is proven much as Proposition \ref{propPBSAnppties}.a was.

\Item (b) The straight forward Neumann expansion gives
\begin{align*}
&\Big|\big(\sfrac{a}{L^2}Q^*Q+\fQ_n\big)_{\fk}^{-1}(\ell,\ell')
                -\fQ_n(\fk+\ell)^{-1}\de_{\ell,\ell'}\Big|\\
&\hskip0.4in\le \sfrac{a}{L^2}\Big|\fQ_n(\fk+\ell)^{-1}
                u_+(\fk+\ell)^\fq u_+(\fk+\ell')^\fq
                \fQ_n(\fk+\ell')^{-1}\Big|\ 
                \sum_{j=0}^\infty\bigg[
     \sum_{l\in\hat\cB^+}\sfrac{a}{L^2}\Big|\sfrac{u_+(\fk+l)^{2\fq}}
                   {\fQ_n(\fk+l)}\Big|\bigg]^j\\
&\hskip0.4in\le \sfrac{a}{L^2}\big(\sfrac{5}{4a}\big)^2
               \big|u_+(\fk+\ell)^\fq\big|
                    \smprod_{\nu=0}^3 
                               \big(\sfrac{24}{L_\nu|\ell'_\nu|+\pi}\big)^\fq
                \sum_{j=0}^\infty\big[\sfrac{5\Gam'_{15}}{4L^2}\big]^j\\
&\hskip0.4in\le \sfrac{2}{aL^2}
                \smprod_{\nu=0}^3 
                               \big(\sfrac{24}{L_\nu|\ell_\nu|+\pi}\big)^\fq
                    \smprod_{\nu=0}^3 
                               \big(\sfrac{24}{L_\nu|\ell'_\nu|+\pi}\big)^\fq
              \ \begin{cases}\qquad 1 & \text{for all $\ell$}\\
                        \noalign{\vskip0.05in}
                       \smprod\limits_{\atop{0\le\nu\le3}{\ell_\nu\ne 0}}
                                            L_\nu|\fk_\nu|& if \text{$\ell\ne 0$}
                  \end{cases}
\end{align*}
by part (a), Lemma \ref{lemPBSuplusppties}.a,b and 
Proposition \ref{propPBSAnppties}.a, with 
$L$ satisfying the conditions of part (a).

\Item (c) The specified bound follows from
\eqref{eqnPBSqplusplusminus}, part (a),
Proposition \ref{propPBSAnppties}.a and 
Lemmas \ref{lemPBSunderiv}.b, \ref{lemPBSuplusppties}.a.
\end{proof}

\begin{lemma}\label{lemPOLscaling}
For all
\begin{align*}
\fk&=\bbbl^{-1}(k)\in  \big(\bbbr/\sfrac{2\pi}{L^2}\bbbz\big)
            \times \big(\bbbr^3/\sfrac{2\pi}{L}\bbbz^3\big)  &
k&\in  \big(\bbbr/2\pi\bbbz\big) \times \big(\bbbr^3/2\pi\bbbz^3\big)\\
p&=\bbbl^{-1}(\fp)\in   \big(\bbbr/\sfrac{2\pi}{\veps_n^2}\bbbz\big)
      \times \big(\bbbr^3/\sfrac{2\pi}{\veps_n}\bbbz^3\big) &
\fp&\in  \big(\bbbr/\sfrac{2\pi}{\veps_{n+1}^2}\bbbz\big)
      \times \big(\bbbr^3/\sfrac{2\pi}{\veps_{n+1}}\bbbz^3\big) \\
& &
\ell_{n+1}&\in \big(2\pi\bbbz/\sfrac{2\pi}{\veps_{n+1}^2}\bbbz\big)
    \times\big(2\pi\bbbz^3/\sfrac{2\pi}{\veps_{n+1}}\bbbz^3\big)
\end{align*}
we have
\begin{enumerate}[label=(\alph*), leftmargin=*]
\item
$
u_+\big(\bbbl^{-1}(k)\big)=u_1(k)\quad\text{and}\quad
u_n\big(\bbbl^{-1}(k)\big)u_+\big(\bbbl^{-1}(k)\big)=u_{n+1}(k)
$
for all $n\in\bbbn$.

\item 
 $\check\fQ_{n+1}\big(\bbbl^{-1}(k)\big)=\sfrac{1}{L^2} \fQ_{n+1}(k)$

\item  $\bD_n^{-1}\big(\bbbl^{-1}(\fp)\big)
                                =L^2\bD_{n+1}^{-1}(\fp)$

\item  $\check\De^{(n+1)(*)}\big(\bbbl^{-1}(k)\big)=
              \sfrac{1}{L^2}\De^{(n+1)(*)}(k)$

\item $\check S_{n+1,\bbbl^{-1}(k)}
              \big(\bbbl^{-1}(\ell_{n+1}),\bbbl^{-1}(\ell'_{n+1})\big)
=L^2\,\hat S_{n+1,k}(\ell_{n+1},\ell'_{n+1})$
\end{enumerate}
\end{lemma}
\begin{proof} 
These all follow from \eqref{eqnPOLcheckDefs}, Remark \ref{remPOLderivAlg}.a and
\begin{itemize}[leftmargin=*, topsep=2pt, itemsep=0pt, parsep=0pt]
\item $Q_1=\bbbl_*^{-1} Q\bbbl_*$
                by \eqref{eqnPBSqn}
\item $\bD_{n+1} = L^2\,\bbbl_*^{-1} \bD_n\bbbl_*$
                by \eqref{eqnPDOdndef}
\end{itemize}
and \cite[Lemmas \lemPoPscaling.b and \lemPoPscalingCrs.b]{Bloch}.
\end{proof}

\begin{corollary}\label{corPOLcheckSnppties}
There are constants $\mm_9>0$ and $\Gam_{18}$
such that, for all $L>\Gam_{17}$,
\begin{equation*}
\|\bbbl_*^{-1}\check S_{n+1}\bbbl_*\|_{\mm_9}\le L^2\Gam_{18}\qquad
\|\bbbl_*^{-1}\check S_{n+1}\check Q_{n+1}^*\bbbl_*\|_{\mm_9} \le L^2 \Gam_{18}
\end{equation*}

\end{corollary}
\begin{proof}
Set $\mm_9=\min\{\mm_3,\mm_6\}$.
Just combine \cite[Lemmas \lemPoPscaling.b and \lemPoPscalingCrs.b]{Bloch}  and 
parts (a) and (e)  of Lemma \ref{lemPOLscaling} to yield
\begin{equation*}
\|\bbbl_*^{-1}\check S_{n+1}\bbbl_*\|_{\mm_9}= L^2 \| S_{n+1}\|_{\mm_9}\qquad
\|\bbbl_*^{-1}\check S_{n+1}\check Q_{n+1}^*\bbbl_*\|_{\mm_9} 
       = L^2 \| S_{n+1}Q_{n+1}\|_{\mm_9}
\end{equation*}
and then apply Proposition \ref{POGmainpos} and Lemma \ref{lemPOGsnQstar}.
\end{proof}

\begin{lemma}\label{lemPOLcheckDeppties}
Assume that $L>\Gam_{17}$, the constant of Lemma \ref{lemPOLcheckAnppties}.
There are constants $\mm_{10}>0$ and $\Gam_{19}$ 
such that the following hold for all $k\in\bbbc^4$ with $|\Im k|\le \mm_{10}$.

\begin{enumerate}[label=(\alph*), leftmargin=*]
\item  
$\Big|\big(\bbbl_*^{-1}A_{\psi_{(*)}\th_{(*)}}\bbbl_*\big)_{k}(\ell_1)\Big|
=\Big|\big(A_{\psi_{(*)}\th_{(*)}}\big)_{\bbbl^{-1}(k)}(\bbbl^{-1}(\ell_1))\Big|
 \le\Gam_{19} \prod\limits_{\nu=0}^3 \big(\sfrac{24}{|\ell_{1,\nu}|+\pi}\big)^\fq$

\item 
If $\ell_1\ne 0$, then 
$\Big|\big(\bbbl_*^{-1} A_{\psi_{(*)}\th_{(*)}}\bbbl_*\big)_{k}(\ell_1)\Big|
 \le\Gam_{19} |k| \prod\limits_{\nu=0}^3 \big(\sfrac{24}{|\ell_{1,\nu}|+\pi}\big)^\fq$

\end{enumerate}
\end{lemma}
\begin{proof}
By \eqref{eqnPOLapsithft}, Lemma \ref{lemPOLscaling} and Remark \ref{remPBSqnft}.e,
\begin{align*}
&\big(A_{\psi_{(*)}\th_{(*)}}\big)_{\bbbl^{-1}(k)}(\bbbl^{-1}(\ell_1))\\
&=\sum_{\ell'_1\in\hat\cB_1}
     \big(\sfrac{a}{L^2}Q^*Q+\fQ_n\big)^{-1}_{\bbbl^{-1}(k)}
                        (\bbbl^{-1}(\ell_1),\bbbl^{-1}(\ell'_1))
      \Big\{\sfrac{a}{L^2}Q^*_{\bbbl^{-1}(k)}(\bbbl^{-1}(\ell'_1))
  \cr \noalign{\vskip-0.1in}&\hskip0.5in
   +\fQ_n(\bbbl^{-1}(k+\ell'_1))\,
    \big(Q_n \bD_n^{{-1}(*)}Q_n^*\big)(\bbbl^{-1}(k+\ell'_1))
  \cr \noalign{\vskip-0.1in}&\hskip3.0in
    Q_{\bbbl^{-1}(k)}^*(\bbbl^{-1}(\ell'_1))
    \check\De^{(n+1)(*)}(\bbbl^{-1}(k))\Big\}\\
&=\sum_{\ell'_1\in\hat\cB_1}
     \big(\sfrac{a}{L^2}Q^*Q+\fQ_n\big)^{-1}_{\bbbl^{-1}(k)}
                        (\bbbl^{-1}(\ell_1),\bbbl^{-1}(\ell'_1))\ 
     B_1(k,\ell'_1)\\
&= \fQ_n\big(\bbbl^{-1}(k+\ell_1)\big)^{-1}\  B_1(k,\ell_1)
+\sum_{\ell'_1\in\hat\cB_1}
     C(k,\ell_1,\ell'_1) B_1(k,\ell'_1)
\end{align*}
where
\begin{align*}
B_1(k,\ell'_1)&=u_1(k+\ell'_1)^\fq\Big\{\sfrac{a}{L^2}
   +\fQ_n\big(\bbbl^{-1}(k+\ell'_1)\big)\,
   \sum_{\ell_n\in\hat\cB_n} B_2(k,\ell'_1,\ell_n)\Big\}\cr
B_2(k,\ell'_1,\ell_n)&=u_n(\bbbl^{-1}(k+\ell'_1)+\ell_n)^{2\fq}\  
    \bD_{n+1}^{{-1}(*)}\big(k+\ell'_1+\bbbl(\ell_n)\big)\ \De^{(n+1)(*)}(k)\cr
C(k,\ell_1,\ell'_1)
&=\big(\sfrac{a}{L^2}Q^*Q+\fQ_n\big)^{-1}_{\bbbl^{-1}(k)}
                        (\bbbl^{-1}(\ell_1),\bbbl^{-1}(\ell'_1))
                -\fQ_n\big(\bbbl^{-1}(k+\ell_1)\big)^{-1}\de_{\ell_1,\ell_1'}
\end{align*}

\Item (a)
Choose $\mm_{10}=\min\{2,\mm_1,\bmm(\pi)\}$. Then
\begin{align*}
\Big| \big(\sfrac{a}{L^2}Q^*Q+\fQ_n\big)^{-1}_{\bbbl^{-1}(k)}
                        (\bbbl^{-1}(\ell_1),\bbbl^{-1}(\ell'_1))\Big|
&\le \sfrac{6}{5} a\,\de_{\ell_1,\ell'_1}
+\sfrac{2}{aL^2}
           \smprod_{\nu=0}^3 \big(\sfrac{24}{|\ell_{1,\nu}|+\pi}\big)^\fq
           \smprod_{\nu=0}^3 \big(\sfrac{24}{|\ell'_{1,\nu}|+\pi}\big)^\fq
          \hidewidth\cr
& && & \hskip-85pt
      \text{(by Lemma \ref{lemPOLcheckAnppties}.b, Proposition \ref{propPBSAnppties}.a)}\cr
\big| u_1(k+\ell'_1)^\fq\big|
   &\le \smprod_{\nu=0}^3 \big(\sfrac{24}{|\ell'_{1,\nu}|+\pi}\big)^\fq &&
   &\text{(by Lemma \ref{lemPBSunppties}.a)}\cr
\big|\fQ_n(\bbbl^{-1}(k+\ell'_1))\big|&\le\sfrac{6}{5}a &&
   &\hskip-18pt\text{(by Proposition \ref{propPBSAnppties}.a)}\cr
\big|u_n(\bbbl^{-1}(k+\ell'_1)+\ell_n)^{2\fq}\big|
         &\le \smprod_{\nu=0}^3 \big(\sfrac{24}{|\ell_{n,\nu}|+\pi}\big)^{2\fq} &&
   &\text{(by Lemma \ref{lemPBSunppties}.a)}\cr
\big|\bD_{n+1}^{{-1}(*)}\big(k+\ell'_1+\bbbl(\ell_n)\big)\big|
         &\le \sfrac{1}{\gam_1 \pi}\text{ if }(\ell'_1,\ell_n)\ne (0,0)  &&
   &\text{(by Lemma \ref{lemPDOhatSzeroppties}.d)}\cr
\big|\De^{(n+1)(*)}(k)\big|&\le 2a  &&
   &\text{(by Lemma \ref{lemPOCDenppties}.c)}\cr
\big| \bD_{n+1}^{{-1}(*)}(k)\ \De^{(n+1)(*)}(k)\big|&\le \Gam_6  &&
   &\text{(by Lemma \ref{lemPOCDenppties}.f)}\cr
\end{align*}
So
\begin{align*}
&\Big|\big(A_{\psi_{(*)}\th_{(*)}}\big)_{\bbbl^{-1}(k)}(\bbbl^{-1}(\ell_1))\Big|
\\
&\hskip0.5in\le \sum_{\ell'_1\in\hat\cB_1}
     \Big\{\sfrac{6}{5} a\,\de_{\ell_1,\ell'_1}
         +\sfrac{2}{aL^2}
           \smprod_{\nu=0}^3 \big(\sfrac{24}{|\ell_{1,\nu}|+\pi}\big)^\fq
           \smprod_{\nu=0}^3 \big(\sfrac{24}{|\ell'_{1,\nu}|+\pi}\big)^\fq
    \Big\}
  \\ \noalign{\vskip-0.1in}&\hskip1.3in
         \smprod_{\nu=0}^3 \big(\sfrac{24}{|\ell'_{1,\nu}|+\pi}\big)^\fq
     \Big\{\sfrac{a}{L^2}
   +\sfrac{6}{5}a\,
    \sum_{\ell_n\in\hat\cB_n}
     \smprod_{\nu=0}^3 \big(\sfrac{24}{|\ell_n|+\pi}\big)^{2\fq} 
    \max\big\{\Gam_6,\sfrac{2a}{\gam_1\pi}\big\}\Big\}\\
&\hskip0.5in\le\const \sum_{\ell'_1\in\hat\cB_1}
     \Big\{\sfrac{6}{5} a\,\de_{\ell_1,\ell'_1}
         +\sfrac{2}{aL^2}
           \smprod_{\nu=0}^3 \big(\sfrac{24}{|\ell_{1,\nu}|+\pi}\big)^\fq
           \smprod_{\nu=0}^3 \big(\sfrac{24}{|\ell'_{1,\nu}|+\pi}\big)^\fq
    \Big\}\smprod_{\nu=0}^3 \big(\sfrac{24}{|\ell'_{1,\nu}|+\pi}\big)^\fq\\
&\hskip0.5in\le\Gam_{19} 
     \smprod_{\nu=0}^3 \big(\sfrac{24}{|\ell_{1,\nu}|+\pi}\big)^\fq
\end{align*}

\Item (b) Using the bounds of part (a) together with
\begin{align*}
\big| u_1(k+\ell'_1)^\fq\big|
   &\le |k|\smprod_{\nu=0}^3 \big(\sfrac{24}{|\ell'_{1,\nu}|+\pi}\big)^\fq\ \ 
   \text{ if }\ell'_1\ne 0 &&
   &\text{(by Lemma \ref{lemPBSunppties}.b)}\\
\big|u_n(\bbbl^{-1}(k+\ell'_1)+\ell_n)^{2\fq}\big|
         &\le |\bbbl^{-1}(k+\ell'_1)|
            \smprod_{\nu=0}^3 \big(\sfrac{24}{|\ell_n|+\pi}\big)^{2\fq} 
            \text{ if }\ell_n\ne 0 &&
   &\text{(by Lemma \ref{lemPBSunppties}.b)}
\end{align*}
we have, if $\ell'_1\ne 0$,
\begin{align*}
\big|B_1(k,\ell'_1)\big|
&\le |k|\smprod_{\nu=0}^3 \big(\sfrac{24}{|\ell'_{1,\nu}|+\pi}\big)^\fq
\Big\{\sfrac{a}{L^2}
   +\sfrac{6}{5}a\,
    \sum_{\ell_n\in\hat\cB_n}
     \smprod_{\nu=0}^3 \big(\sfrac{24}{|\ell_n|+\pi}\big)^{2\fq} 
    \max\big\{\Gam_6,\sfrac{2a}{\gam_1\pi}\big\}\Big\}\\
&\le \const |k|\smprod_{\nu=0}^3 \big(\sfrac{24}{|\ell'_{1,\nu}|+\pi}\big)^\fq
\end{align*}
and 
\begin{align*}
\big|B_1(k,0)\big|
&\le \smprod_{\nu=0}^3 \big(\sfrac{24}{\pi}\big)^\fq
\Big\{\sfrac{a}{L^2}
   +\sfrac{6}{5}a\,
    \sum_{\ell_n\in\hat\cB_n}
     \smprod_{\nu=0}^3 \big(\sfrac{24}{|\ell_n|+\pi}\big)^{2\fq} 
    \max\big\{\Gam_6,\sfrac{2a}{\gam_1\pi}\big\}\Big\}
\le \const 
\end{align*}
Using these bounds, the first bound of part (a) and 
Lemma \ref{lemPOLcheckAnppties}.b, and assuming that  $\ell_1\ne 0$,
\begin{align*}
&\big(A_{\psi_{(*)}\th_{(*)}}\big)_{\bbbl^{-1}(k)}(\bbbl^{-1}(\ell_1))
=\sum_{\ell'_1\in\hat\cB_1}
     \big(\sfrac{a}{L^2}Q^*Q+\fQ_n\big)^{-1}_{\bbbl^{-1}(k)}
                        (\bbbl^{-1}(\ell_1),\bbbl^{-1}(\ell'_1))\ 
     B_1(k,\ell'_1)\\
&\hskip0.5in=\big(\sfrac{a}{L^2}Q^*Q+\fQ_n\big)^{-1}_{\bbbl^{-1}(k)}
                        (\bbbl^{-1}(\ell_1),0))\  B_1(k,0)
  +O\Big(|k|\smprod_{\nu=0}^3 
               \big(\sfrac{24}{|\ell_{1,\nu}|+\pi}\big)^\fq\Big)\\
&\hskip0.5in=\sfrac{2}{aL^2}
                \smprod_{\nu=0}^3 
                               \big(\sfrac{24}{|\ell_{1,\nu}|+\pi}\big)^\fq
                    \smprod_{\nu=0}^3 
                               \big(\sfrac{24}{\pi}\big)^\fq
                \smprod\limits_{\atop{0\le\nu\le3}{\ell_{1,\nu}\ne 0}} |k_\nu|
                 \  B_1(k,0)
  +O\Big(|k|\smprod_{\nu=0}^3 
               \big(\sfrac{24}{|\ell_{1,\nu}|+\pi}\big)^\fq\Big)\\
&\hskip0.5in=O\Big(|k|\smprod_{\nu=0}^3 
               \big(\sfrac{24}{|\ell_{1,\nu}|+\pi}\big)^\fq\Big)
\end{align*}
\end{proof}
\begin{proof}[Proof of Proposition \ref{propPOLmain}]
\ 
\Item \emph{Bound on 
   $\big\|\bbbl_*^{-1} A_{\psi,\phi}\bbbl_*\big\|_{m=1}$:}\ \ \ 
By \cite[Lemma \lemPoPscaling.b]{Bloch} and Lemma \ref{lemPOLcheckAnppties}.b,
if $|\Im k_{\nu'}|\le 2$ for each $0\le\nu'\le3$ then
\begin{align*}
&\big|\big[\bbbl_*^{-1}\big\{\big(aL^{-2}Q^*Q+\fQ_n\big)^{-1}
                -\fQ_n^{-1}\big\}\bbbl_*\big]_k(\ell,\ell')\big| \\
&\hskip0.5in=\big|\big[\big(aL^{-2}Q^*Q+\fQ_n\big)^{-1}
        -\fQ_n^{-1}\big]_{\bbbl^{-1}k}(\bbbl^{-1}\ell,\bbbl^{-1}\ell')\big|\\
&\hskip0.5in\le \sfrac{2}{aL^2}
                \smprod_{\nu=0}^3 
                               \big(\sfrac{24}{|\ell_\nu|+\pi}\big)^\fq
                    \smprod_{\nu=0}^3 
                               \big(\sfrac{24}{|\ell'_\nu|+\pi}\big)^\fq
\end{align*}
So, by \cite[Lemma \lemBOlonelinfty.b]{Bloch}, 
\begin{equation*}
\big\|\bbbl_*^{-1}\big\{\big(aL^{-2}Q^*Q+\fQ_n\big)^{-1}
                -\fQ_n^{-1}\big\}\bbbl_*\|_{m=1}
\le\const_{\fq}
\end{equation*}
By Proposition \ref{propPBSAnppties}.a, Lemma \ref{lemPBSunppties}.a and 
\cite[Lemmas \lemBOlonelinfty.b,c]{Bloch},
\begin{equation*}
\|\fQ_n\|_{m=1},\ \|\fQ_n^{-1}\|_{m=1},\ \|Q_n\|_{m=1}\le\const_{\fq}
\end{equation*}
too.  Now just apply 
\cite[Lemmas \lemPoPscaling.c and \lemPoPscalingCrs.c]{Bloch}.

\Item 
\emph{Bound on 
   $\big\|\bbbl_*^{-1} A_{\psi,\phi,\nu}\bbbl_*\big\|_{m=1}$:}\ \ \ 
By \cite[Lemma \lemPoPscaling.b]{Bloch} and Lemma \ref{lemPOLcheckAnppties}.c,
if $|\Im k_{\nu'}|\le 2$ for each $0\le\nu'\le3$ then
\begin{align*}
\big|\big[\bbbl_*^{-1}\big(\fQ_n^{-1}Q^{(+)}_{+,\nu}\check\fQ_{n+1} Q^{(-)}_{+,\nu}
  \big)\bbbl_*\big]_k(\ell,\ell')\big|
&=\big|\big[\fQ_n^{-1}Q^{(+)}_{+,\nu}\check\fQ_{n+1} Q^{(-)}_{+,\nu}
         \big]_{\bbbl^{-1}k}(\bbbl^{-1}\ell,\bbbl^{-1}\ell')\big|\\
&\le \sfrac{3e^4}{2L^2}
     \big(\sfrac{24}{\pi}\big)^4
  \smprod\limits_{\nu=0}^3 \big(\sfrac{24}{|\ell_\nu|+\pi}\big)^{\fq-1}
   \smprod\limits_{\nu=0}^3 \big(\sfrac{24}{|\ell'_\nu|+\pi}\big)^\fq
\end{align*}
As $\fq>2$, \cite[Lemma \lemBOlonelinfty.b]{Bloch} yields 
\begin{equation*}
\big\|\bbbl_*^{-1} \fQ_n^{-1}Q^{(+)}_{+,\nu}\check\fQ_{n+1} Q^{(-)}_{+,\nu}\bbbl_*\|_{m=1}
\le\const_{\fq}
\end{equation*}
By Lemma \ref{lemPBSunderiv}.b and \cite[Lemma \lemBOlonelinfty.c]{Bloch},
$\|Q_{n,\nu}^{(-)}\|_{m=1}\le\const_{\fq}$ too, since $\fq>2$. Now just apply
\cite[Lemma \lemPoPscalingCrs.c]{Bloch}.

\Item 
\emph{Bound on 
   $\big\|\bbbl_*^{-1}A_{\psi_{(*)}\th_{(*)}}\bbbl_*\big\|_{\mm_8}$:}\ \ \ 
This follows from Lemma \ref{lemPOLcheckDeppties}.a by \cite[Lemma \lemBOlonelinfty.c]{Bloch}.

\Item 
\emph{Bound on 
   $\big\|\bbbl_*^{-1}\ \sfrac{a}{L^2} C^{(n)}(\mu)^{(*)} Q^*\  \bbbl_*\big\|_{\mm_8}$:}\ \ \ 
By \eqref{eqnPOLcmunomu}
\begin{equation*}
\bbbl_*^{-1}\ \sfrac{a}{L^2} C^{(n)}(\mu)^{(*)} Q^*\  \bbbl_*
=\bbbl_*^{-1}A_{\psi_{(*)}\th_{(*)}}\bbbl_*
+L^2\mu \bbbl_*^{-1}A_{\psi,\phi}\bbbl_* 
    S_{n+1}^{(*)}  S_{n+1}(L^2\mu)^{(*)} Q_{n+1}^* \fQ_{n+1}
\end{equation*}
Now just apply Proposition \ref{POGmainpos}, Lemma \ref{lemPBSunppties} and
Proposition \ref{propPBSAnppties}.c.

\Item 
\emph{Bound on 
   $\big\|\bbbl_*^{-1} A_{\psi_{(*)}\th_{(*)}\nu}\bbbl_*\big\|_{\mm_8}$:}\ \ \ 
It suffices to bound
\begin{itemize}[leftmargin=*, topsep=2pt, itemsep=0pt, parsep=0pt]
\item
$\bbbl_*^{-1} A_{\psi,\phi,\nu}\bbbl_*$ as above,
\item
 bound $\bbbl_*^{-1} \fQ_n^{-1}Q_{+,\nu}^{(+)}\check  \fQ_{n+1} \bbbl_*$ using 
\begin{align*}
\big|\big[\bbbl_*^{-1}\big(\fQ_n^{-1}Q^{(+)}_{+,\nu}\check\fQ_{n+1}
  \big)\bbbl_*\big]_k(\ell)\big|
&=\big|\big[\fQ_n^{-1}Q^{(+)}_{+,\nu}\check\fQ_{n+1} 
         \big]_{\bbbl^{-1}k}(\bbbl^{-1}\ell)\big|\\
&\le \sfrac{3e^2}{2L^2}
     \big(\sfrac{24}{\pi}\big)^3
  \smprod\limits_{\nu=0}^3 \big(\sfrac{24}{|\ell_\nu|+\pi}\big)^{\fq-1}
\end{align*}
(by \cite[Lemma \lemPoPscalingCrs.b]{Bloch}, \eqref{eqnPBSqplusplusminus}, 
Proposition \ref{propPBSAnppties}.a and 
Lemmas \ref{lemPBSunderiv}.b, \ref{lemPBSuplusppties}.a,
\ref{lemPOLcheckAnppties}.a) and \cite[Lemma \lemBOlonelinfty.c]{Bloch}, and

\item 
bound $D_{n+1}^{{-1}(*)}Q_{n+1,\nu}^{(+)}\De^{(n+1)(*)}$
using \begin{align*}
&\big|\big(\hat\bD_{n+1}^{{-1}(*)}Q_{n+1,\nu}^{(+)}\De^{(n+1)(*)}
  \big)_k(\ell_{n+1})\big|\\
&\hskip0.3in\le \big|\hat\bD_{n+1}^{{-1}(*)}\big(k+\ell_{n+1})\big)
\ze_{n+1,\nu}^{(+)}(k,\ell_{n+1})
u_{n+1,\nu}^{(+)}(k+\ell_{n+1})u_{n+1}(k+\ell_{n+1})^{\fq-1} \\
&\hskip4.3in\hat\De^{(n+1)(*)}(k)\big|
\end{align*}
(by \eqref{eqnPBSqnplusminus}) and
\begin{align*}
\big|\ze_{n+1,\nu}^{(+)}(k,\ell_{n+1})u_{n+1,\nu}^{(+)}(k+\ell_{n+1})\big|
         &\le e^2\big(\sfrac{24}{\pi}\big)^3 &
   &\text{(by Lemma \ref{lemPBSunderiv}.b)}\\
\big|u_{n+1}(k+\ell_{n+1})^{\fq-1}\big|
         &\le \smprod_{\nu=0}^3 \big(\sfrac{24}{|\ell_{n+1}|+\pi}\big)^{\fq-1} &
   &\text{(by Lemma \ref{lemPBSunppties}.a)}\\
\big|\hat\bD_{n+1}^{{-1}(*)}\big(k+\ell_{n+1}\big)\big|
         &\le \sfrac{1}{\gam_1 \pi}\text{ if }\ell_{n+1}\ne 0  &
   &\text{(by Lemma \ref{lemPDOhatSzeroppties}.d)}\\
\big|\hat\De^{(n+1)(*)}(k)\big|&\le 2a  &
   &\text{(by Lemma \ref{lemPOCDenppties}.c)}\\
\big|\hat\bD_{n+1}^{{-1}(*)}(k)\ \hat\De^{(n+1)(*)}(k)\big|&\le \Gam_6  &
   &\text{(by Lemma \ref{lemPOCDenppties}.f)}
\end{align*}
and \cite[Lemma \lemBOlonelinfty.c]{Bloch}.
\end{itemize}

\Item 
\emph{Bound on 
   $\big\|\bbbl_*^{-1} A_{\psi_{(*)}\th_{(*)}\nu}(\mu)\bbbl_*\big\|_{\mm_8}$:}\ \ \
This follows from the previous bounds of this Proposition, 
Remark \ref{remPOLderivAlg}.b, Proposition \ref{POGmainpos}, Lemma \ref{lemPBSunderiv}.c and
Proposition \ref{propPBSAnppties}.c.

\end{proof}

\newpage
\appendix
\section{Trigonometric Inequalities}\label{appTrigIneq}

\begin{lemma}\label{lemPBSsinxoverx}
\ 
\begin{enumerate}[label=(\alph*), leftmargin=*]
\item 
For $x,y$ real with $|x|\le\sfrac{\pi}{2}$,
\begin{equation*}
\big|\sin(x+iy)\big|\ge \sfrac{\sqrt{2}}{\pi}|x+iy|
\end{equation*}

\item  
For $x,y$ real with $|y|\le 1$,
\begin{equation*}
\sfrac{|\sin(x+iy)|}{|x+iy|}\le 2\min\big\{1\,,\,\sfrac{1}{|x+iy|}\big\}\qquad
\Big|\Im\sfrac{\sin(x+iy)}{x+iy}\Big|
   \le 2|y|\min\big\{|x|\,,\,\sfrac{2}{|x+iy|}\big\}   
\end{equation*}

\item 
For $0<\veps\le 1$ and $x,y$ real with $|\veps x|\le\pi$, $|y|\le 2$
\begin{align*}
\bigg|\frac{\sin\half (x+iy)}{\sfrac{1}{\veps}\sin\half\veps (x+iy)}\bigg|
    \le 4\min\left\{1,\frac{2}{|x|}\right\}\qquad
\bigg|\Im\frac{\sin\half (x+iy)}{\sfrac{1}{\veps}\sin\half\veps (x+iy)}\bigg|
    \le 6 |y|\min\left\{|x|\,,\,\frac{8}{|x|}\right\}
\end{align*}

\item 
For $x$ real with $|x|\le\sfrac{\pi}{2}$,
\begin{equation*}
\sfrac{2}{\pi}\le\sfrac{\sin x}{x} \le 1
\end{equation*}

\item 
For any complex number $z$ obeying $|z|\le 2$, 
\begin{equation*}
\big|\sfrac{\sin z}{z}-1\big|\le\half |z|^2
\end{equation*}
\end{enumerate}
\end{lemma}
\begin{proof} 
By the standard trig identity
\begin{equation*}
\sin(x+iy)=\sin(x)\cos(iy)+\cos(x)\sin(iy)
=\sin(x)\cosh(y)+i\cos(x)\sinh(y)
\end{equation*}

\Item (a)
For $x,y$ real with $|x|\le\sfrac{\pi}{2}$
\begin{align*}
\big|\Re \sin(x+iy)\big| &= \big|\sin(x)\cosh(y)\big|
                         \ge|\sin(x)|
                         \ge\sfrac{2}{\pi}|x|\\
\big|\Re \sin(x+iy)\big| &= \big|\sin(x)\cosh(y)\big|
                         \ge\big|\sin(x)\sinh(y)\big|
                         \ge|\sin(x)|\,|y| \\
\big|\Im \sin(x+iy)\big| &= \big|\cos(x)\sinh(y)\big|
                         \ge|\cos(x)|\,|y|
\end{align*}
so that
\begin{equation*}
\big|\sin(x+iy)\big|\ge\max\big\{\sfrac{2}{\pi}|x|\,,\,|y|\big\}
\ge \sfrac{\sqrt{2}}{\pi}|x+iy|
\end{equation*}

\Item (b)
For $x,y$ real with $|y|\le 1$,
\begin{align*}
|\sin(x+iy)|&=\big|\sin(x)\cosh(y)+i\cos(x)\sinh(y)\big|
           \le\cosh(1)\big|\sin(x)+i\cos(x)\big| \\
           &=\cosh(1)
\end{align*}
and, since 
$
\big|\sinh(y)\big|=\Big|\sum_{n=0}^\infty\sfrac{y^{2n+1}}{(2n+1)!}\Big|
\le |y|\sum_{n=0}^\infty\sfrac{|y|^{2n}}{(2n)!}
=|y|\cosh(y)\le\cosh(1)|y|
$,
\begin{align*}
|\sin(x+iy)|&=\big|\sin(x)\cosh(y)+i\cos(x)\sinh(y)\big|
           \le\cosh(1)\big||x|+i|y|\big| \\
            &=\cosh(1) |x+iy|
\end{align*}
Thus
\begin{equation*}
\sfrac{|\sin(x+iy)|}{|x+iy|}\le\cosh(1)\min\big\{1\,,\,\sfrac{1}{|x+iy|}\big\}
\le 2\min\big\{1\,,\,\sfrac{1}{|x+iy|}\big\}
\end{equation*}
giving the first bound.

For the second bound
\begin{align*}
\Big|\Im\sfrac{\sin(x+iy)}{x+iy}\Big|
=\Big|\sfrac{x\,\Im\sin(x+iy)-y\,\Re\sin(x+iy)}{x^2+y^2}\Big|
=\Big|\sfrac{x\,\cos(x)\sinh(y)-y\,\sin(x)\cosh(y)}{x^2+y^2}\Big|
\end{align*}
Using
\begin{alignat*}{3}
\cos x-1&=-\int_0^x dt\ \sin t 
        &&=-\int_0^x dt\int_0^t ds\ \cos s\\
\sin x-x&= \int_0^x dt\ [\cos t-1] 
      &&=-\int_0^x dt\int_0^t ds\int_0^s du\ \cos u\\
\cosh y-1&=\int_0^y dt\ \sinh t 
       &&=\int_0^y dt\int_0^t ds\ \cosh s\\
\sinh y-y&= \int_0^y dt\ [\cosh t-1] 
      &&=\int_0^y dt\int_0^t ds\int_0^s du\ \cosh u
\end{alignat*}
and $\cosh(1) < 2$, we have
\begin{alignat*}{3}
\cos x &= 1 + \al(x) \sfrac{x^2}{2}\qquad&
\sin x &= x + \be(x) \sfrac{x^3}{6}\\
\cosh y &= 1 +\ga(y) y^2\qquad&
\sinh y &= y + \de(y)\sfrac{y^3}{3}
\end{alignat*}
with, for $|y|\le 1$, $|\al(x)|\,,\,|\be(x)|\,,\,|\ga(y)|\,,\,|\de(y)|\ \le 1$.
Consequently,
\begin{equation*}
\Big|\Im\sfrac{\sin(x+iy)}{x+iy}\Big|\le 2|xy|
\end{equation*}
Alternatively, using $|\sin(x)|\le |x|$, $|\sinh(y)|\le 2|y|$ 
and $|\cosh(y)|\le 2$,
\begin{equation*}
\Big|\Im\sfrac{\sin(x+iy)}{x+iy}\Big|\le \sfrac{4|xy|}{x^2+y^2}
\le \sfrac{4|y|}{\sqrt{x^2+y^2}}
\end{equation*}

\Item  (c)
For $0<\veps\le 1$, $x,y$ real and $|\veps x|\le\pi$, $|y|\le 2$, 
\begin{align*}
\bigg|\frac{\sin\half (x+iy)}{\sfrac{1}{\veps}\sin\half\veps (x+iy)}\bigg|
&=\left|\frac{\sfrac{\sin\half (x+iy)}{\half (x+iy)}}
     {\sfrac{\sin\half\veps (x+iy)}{\half\veps (x+iy)}}\right|
\le\sfrac{\pi}{\sqrt{2}}\cosh(1)\min\left\{1,\frac{2}{|x+iy|}\right\}
\le 4\min\left\{1,\frac{2}{|x|}\right\}
\end{align*}
and 
\begin{align*}
\bigg|\Im\frac{\sin\half (x+iy)}{\sfrac{1}{\veps}\sin\half\veps (x+iy)}\bigg|
&=\frac{\big|\Im\sfrac{\sin\half (x+iy)}{\half (x+iy)}
               \Re\sfrac{\sin\half\veps (x+iy)}{\half\veps (x+iy)}
             -\Re\sfrac{\sin\half (x+iy)}{\half (x+iy)}
               \Im\sfrac{\sin\half\veps (x+iy)}{\half\veps (x+iy)}\big|}
     {\big|\sfrac{\sin\half\veps (x+iy)}{\half\veps (x+iy)}\big|^2}\\
&\le\frac{|y|\min\big\{\half|x|\,,\,\sfrac{4}{|x+iy|}\big\}
             +4\min\big\{1,\sfrac{2}{|x|}\big\}
               \big|\Im\sfrac{\sin\half\veps (x+iy)}{\half\veps (x+iy)}\big|}
     {\big|\sfrac{\sin\half\veps (x+iy)}{\half\veps (x+iy)}\big|}\\
&\le\frac{|y|\min\big\{\half|x|\,,\,\sfrac{4}{|x+iy|}\big\}
             +4\min\big\{1,\sfrac{2}{|x|}\big\}
               \veps|y|\min\big\{\half\veps|x|\,,\,\sfrac{4}{\veps|x+iy|}\big\}}
     {\sfrac{\sqrt{2}}{\pi}}\\
&\le \sfrac{\pi}{\sqrt{2}}|y|\min\big\{
         \big(\half+2\veps^2\big)|x|\,,\,\sfrac{4+16}{|x+iy|}\big\}\\
&\le 6 |y|\min\big\{|x|\,,\,\sfrac{8}{|x+iy|}\big\}
\end{align*}

\Item (d) and (e) are standard.
\end{proof}

\newpage
\section{Lattice and Operator Summary}\label{appLatOpSummary}
The following table gives, for most of the
operators considered in this paper,
\begin{itemize}[leftmargin=*, topsep=2pt, itemsep=0pt, parsep=0pt]
\item 
the definition of the operator
\item
a reference to where in \cite{PAR1,PAR2}, the operator is introduced and
\item
the translation invariance properties of the operator.
\end{itemize}
A later table will specify where, in this paper, bounds
on the operators are proven.

\begin{center}
\renewcommand{\arraystretch}{1.3}
  \begin{tabular}{|cl|c|c| }
    \hline
     Operator& 
     &Definition
     &Tiwrt \\ \hline
    $D_0=-e^{-h_0}\partial_0 +\big[\bbbone-e^{-h_0}\big]$
   &$\!\!\!\!:\cH_0\rightarrow\cH_0$
   &\S\sectINTstartPoint 
   &$\cX_0$ \\ \hline
   $D_n = L^{2n}\ \bbbl_*^{-n} \, D_0\, \bbbl_*^n$
  &$\!\!\!:\cH_n\rightarrow\cH_n$
  &Def \defHTbackgrounddomaction.a
  &$\cX_n$ \\ \hline
   $\fQ_n=a\big(\bbbone
             +\smsum_{j=1}^{n-1}\sfrac{1}{L^{2j}}Q_jQ_j^*\big)^{-1}$
  &$\!\!\!:\cH_0^{(n)}\rightarrow\cH_0^{(n)}$
  &Def \defHTbackgrounddomaction.b
  &$\cX_0^{(n)}$ \\ \hline
   $\De^{(0)}=D_0$
  &$\!\!\!:\cH_0\rightarrow\cH_0$
  &(\eqnHTden)
  &$\cX_0$ \\ \hline
   $\De^{(n)}=\big(\bbbone+\fQ_n Q_nD_n^{-1} Q_n^*\big)^{-1}\fQ_n,\ n\ge 1$
  &$\!\!\!:\cH_0^{(n)}\rightarrow\cH_0^{(n)}$
  &(\eqnHTden)
  &$\cX_0^{(n)}$ \\ \hline
   $C^{(n)}=\big(\sfrac{a}{L^2} Q^*Q+\De^{(n)}\big)^{-1}$
  &$\!\!\!:\cH_0^{(n)}\rightarrow\cH_0^{(n)}$
  &(\eqnHTcn)
  &$\cX_{-1}^{(n+1)}$ \\ \hline
   $S_n^{-1}=D_n+Q^*_n\fQ_nQ_n$
  &$\!\!\!:\cH_n\rightarrow\cH_n$
  &Thm \HTthminvertibleoperators
  &$\cX_0^{(n)}$ \\ \hline
   $S_n(\mu)^{-1}=D_n+Q^*_n\fQ_nQ_n-\mu$
  &$\!\!\!:\cH_n\rightarrow\cH_n$
  &Thm \HTthminvertibleoperators
  &$\cX_0^{(n)}$ \\ \hline
   $\check\fQ_{n+1}=\big(\sfrac{L^2}{a}\bbbone+Q\fQ_n^{-1}Q^*\big)^{-1}$
  &$\!\!\!:\cH_{-1}^{(n+1)}\!\rightarrow\!\cH_{-1}^{(n+1)}$
  &Lem \lemSCacheckOne.b
  &$\cX_{-1}^{(n+1)}$ \\ \hline
   $\check S_{n+1}(\mu)=\big\{D_n
     \!+\!\check Q_{n+1}^*\check\fQ_{n+1}\check Q_{n+1}\!-\!\mu\big\}^{-1}\!$
  &$\!\!\!:\cH_n\rightarrow\cH_n$
  &(\eqnBGEcheckSmu)
  &$\cX_{-1}^{(n+1)}$ \\ \hline
   $A_{\psi,\phi}$
  &$\!\!\!:\cH_n\rightarrow\cH_0^{(n)}$
  &Prop \propCFpsisoln
  &$\cX_{-1}^{(n+1)}$ \\ \hline
  \end{tabular}
\renewcommand{\arraystretch}{1.0}
\end{center}

\noindent 
The references in the above table are to \cite{PAR1,PAR2} and 
``Tiwrt'' stands for ``translation invariant with respect to''.

\bigskip

\begin{center}
\renewcommand{\arraystretch}{1.3}
  \begin{tabular}{|cl|c|c| }
    \hline
     Operator& 
     &Definition
     &Tiwrt \\ \hline
   $S_nQ_n^*= D_n^{-1}Q_n^*\big(\bbbone+\fQ_n Q_nD^{-1}_n Q_n^*\big)^{-1}$
  &$\!\!\!:\cH_0^{(n)}\rightarrow\cH_n$
  &Lemma \ref{lemPOGsnQstar}
  &$\cX_0^{(n)}$ \\ \hline
   $\check S_{n+1}=\check S_{n+1}(0)$
  &$\!\!\!:\cH_n\rightarrow\cH_n$
  &after \eqref{eqnPOLcheckDefs}
  &$\cX_{-1}^{(n+1)}$ \\ \hline
   $A_{\psi,\phi}$
  &$\!\!\!:\cH_n\rightarrow\cH_0^{(n)}$
  &\eqref{eqnPOLapsiphi}
  &$\cX_{-1}^{(n+1)}$ \\ \hline
   $A_{\psi,\phi,\nu}$
  &$\!\!\!:\cH_n\rightarrow\cH_0^{(n)}$
  &Remark \ref{remPOLderivAlg}
  &$\cX_{-1}^{(n+1)}$ \\ \hline
   $A_{\psi_{(*)},\th_{(*)}}=\sfrac{a}{L^2} {C^{(n)}}^{(*)}Q^*$
  &$\!\!\!:\cH_{-1}^{(n+1)}\rightarrow\cH_0^{(n)}$
  &before Rmk \ref{remPOLderivAlg}
  &$\cX_{-1}^{(n+1)}$ \\ \hline
   $A_{\psi_*,\th_*,\nu},\ A_{\psi,\th,\nu},\ 
          A_{\psi_*,\th_*,\nu}(\mu),\ A_{\psi,\th,\nu}(\mu)$
  &$\!\!\!:\cH_{-1}^{(n+1)}\rightarrow\cH_0^{(n)}$
  &Remark \ref{remPOLderivAlg}
  &$\cX_{-1}^{(n+1)}$ \\ \hline
  \end{tabular}
\renewcommand{\arraystretch}{1.0}
\end{center}

\noindent
 The lattices involved are
\begin{alignat*}{3}
\cX_n&=\big(\veps_n^2\bbbz/\veps_n^2L_\tp\bbbz\big)\times
         \big(\veps_n\bbbz^3/\veps_nL_\sp\bbbz^3\big) &
\hat\cX_n
  &=\big(\sfrac{2\pi}{\veps_n^2L_\tp}\bbbz/\sfrac{2\pi}{\veps_n^2}\bbbz\big)
    \!\times\!
 \big(\sfrac{2\pi}{\veps_nL_\sp}\bbbz^3/\sfrac{2\pi}{\veps_n}\bbbz^3\big) 
  \\
\cX_{0}^{(n)}&=\big(\bbbz/\veps_n^2L_\tp\bbbz\big)\times
         \big(\bbbz^3/\veps_nL_\sp\bbbz^3\big) &
\hat\cX_{0}^{(n)}&=\big(\sfrac{2\pi}{\veps_n^2L_\tp}\bbbz/2\pi\bbbz\big)
       \!\times\!
         \big(\sfrac{2\pi}{\veps_nL_\sp}\bbbz^3/2\pi\bbbz^3\big)\\
\cX_{-1}^{(n+1)}\!&=\big(L^2\bbbz/\veps_n^2L_\tp\bbbz\big)\!\times\!
         \big(L\bbbz^3/\veps_nL_\sp\bbbz^3\big)\quad\  &
\!\!\!\hat\cX_{-1}^{(n+1)}\!
  &=\big(\sfrac{2\pi}{\veps_n^2L_\tp}\bbbz/\sfrac{2\pi}{L^2}\bbbz\big)
\!\times\!         \big(\sfrac{2\pi}{\veps_nL_\sp}\bbbz^3/\sfrac{2\pi}{L}\bbbz^3\big)
\end{alignat*}
where $\veps_n=\sfrac{1}{L^n}$.
The ``single period'' lattices are
\begin{alignat*}{3}
\cB_n&=\big(\veps_n^2\bbbz/\bbbz\big)\times
         \big(\veps_n\bbbz^3/\bbbz^3\big)\qquad &
\hat\cB_n
  &=\big(2\pi\bbbz/\sfrac{2\pi}{\veps_n^2}\bbbz\big)
    \!\times\!
 \big(2\pi\bbbz^3/\sfrac{2\pi}{\veps_n}\bbbz^3\big) 
  \\
\cB^+&=\big(\bbbz/L^2\bbbz\big)\times
         \big(\bbbz^3/L\bbbz^3\big) &
\hat\cB^+&=\big(\sfrac{2\pi}{L^2}\bbbz/2\pi\bbbz\big)
       \!\times\!
         \big(\sfrac{2\pi}{L}\bbbz^3/2\pi\bbbz^3\big)
\end{alignat*}

\noindent
The following table specifies where, in this paper, bounds
on the various operators are proven.

\begin{center}
\renewcommand{\arraystretch}{1.3}
  \begin{tabular}{|c|c|}
    \hline
     Operator
     &Bound \\ \hline
   $Q_n$
  &Lemma \ref{lemPBSunppties}.a\\ \hline
   $\fQ_n$
  &Proposition \ref{propPBSAnppties}\\ \hline
   $Q_{n,\nu}^{(\pm)}$
  &Lemma \ref{lemPBSunderiv} \\ \hline
   $D_n$
  &Lemma \ref{lemPDOhatSzeroppties}\\ \hline
   $\De^{(n)}$
  &Lemma \ref{lemPOCDenppties}\\ \hline
   $C^{(n)}$
  &Corollary \ref{corPOCsquareroot}\\ \hline
   $D^{(n)}$
  &Corollary \ref{corPOCsquareroot}\\ \hline
   $S_n(\mu),\ S_n$
  &Proposition \ref{POGmainpos}\\ \hline
   $S_{n,\nu}^{(\pm)}(\mu),\ S_{n,\nu}^{(\pm)}$
  &Proposition \ref{POGmainpos}\\ \hline
   $A_{\psi,\phi}$
  &Proposition \ref{propPOLmain}\\ \hline
   $A_{\psi,\phi,\nu}$
  &Proposition \ref{propPOLmain}\\ \hline
   $A_{\psi_{(*)}\th_{(*)}}$
  &Proposition \ref{propPOLmain}\\ \hline
   $A_{\psi_{(*)}\th_{(*)}\nu}(\mu)$
  &Proposition \ref{propPOLmain}\\ \hline
   $\check S_{n+1}$
  &Corollary \ref{corPOLcheckSnppties}\\ \hline
  \end{tabular}
\renewcommand{\arraystretch}{1.0}
\end{center}

\newpage
\bibliographystyle{plain}
\bibliography{refs}

\begin{thebibliography}{1}

\bibitem{fnlint1}
T.~Balaban, J.~Feldman, H.~Kn{\"o}rrer, and E.~Trubowitz.
\newblock {A Functional Integral Representation for Many Boson Systems. I: The
  Partition Function}.
\newblock {\em Annales Henri Poincar{\'e}}, 9:1229--1273, 2008.

\bibitem{fnlint2}
T.~Balaban, J.~Feldman, H.~Kn{\"o}rrer, and E.~Trubowitz.
\newblock {A Functional Integral Representation for Many Boson Systems. II:
  Correlation Functions}.
\newblock {\em Annales Henri Poincar{\'e}}, 9:1275--1307, 2008.

\bibitem{CPC}
T.~Balaban, J.~Feldman, H.~Kn{\"o}rrer, and E.~Trubowitz.
\newblock {Power Series Representations for Complex Bosonic Effective Actions.
  I. A Small Field Renormalization Group Step}.
\newblock {\em Journal of Mathematical Physics}, 51:053305, 2010.

\bibitem{UV}
T.~Balaban, J.~Feldman, H.~Kn{\"o}rrer, and E.~Trubowitz.
\newblock {The Temporal Ultraviolet Limit for Complex Bosonic Many-body
  Models}.
\newblock {\em Annales Henri Poincar{\'e}}, 11:151--350, 2010.

\bibitem{Bloch}
T.~Balaban, J.~Feldman, H.~Kn{\"o}rrer, and E.~Trubowitz.
\newblock {Bloch Theory for Periodic Block Spin Transformations}.
\newblock Preprint, 2016.

\bibitem{BlockSpin}
T.~Balaban, J.~Feldman, H.~Kn{\"o}rrer, and E.~Trubowitz.
\newblock {The Algebra of Block Spin Renormalization Group Transformations}.
\newblock Preprint, 2016.

\bibitem{PAR1}
T.~Balaban, J.~Feldman, H.~Kn{\"o}rrer, and E.~Trubowitz.
\newblock {The Small Field Parabolic Flow for Bosonic Many--body Models: Part 1
  --- Main Results and Algebra}.
\newblock Preprint, 2016.

\bibitem{PAR2}
T.~Balaban, J.~Feldman, H.~Kn{\"o}rrer, and E.~Trubowitz.
\newblock {The Small Field Parabolic Flow for Bosonic Many--body Models: Part 2
  --- Fluctuation Integral and Renormalization}.
\newblock Preprint, 2016.

\bibitem{BGE}
T.~Balaban, J.~Feldman, H.~Kn{\"o}rrer, and E.~Trubowitz.
\newblock {The Small Field Parabolic Flow for Bosonic Many--body Models: Part 4
  --- Background and Critical Field Estimates}.
\newblock Preprint, 2016.

\end{thebibliography}

\end{document}